\documentclass[11pt,english]{article}

\usepackage[margin=1in]{geometry}
\usepackage{bbm}
\usepackage{graphicx,color}
\usepackage{amsmath,amssymb,amsthm}
\usepackage{enumerate}
\usepackage{enumitem}
\usepackage{fullpage}
\usepackage{algorithm}
\usepackage{mathrsfs}
\usepackage[noend]{algcompatible}
\usepackage[noblocks]{authblk}
\usepackage{varwidth}
\usepackage{mathrsfs}
\usepackage{babel}
\usepackage{accents}
\usepackage{mathtools}
\usepackage{soul}
\usepackage{tipa}
\usepackage[most]{tcolorbox}
\usepackage{wrapfig}
\usepackage{thm-restate}
\usepackage{longtable}
\usepackage[backref=page, pagebackref=true]{hyperref}
\usepackage[capitalize]{cleveref} 
\Crefname{enumtry}{try1}{}   

\DeclarePairedDelimiter{\nor2}{\lVert}{\rVert}

\newtheorem{goal}{Goal}

\def\denseformat{
\setlength{\textheight}{9in}
\setlength{\textwidth}{6.9in}
\setlength{\evensidemargin}{-0.2in}
\setlength{\oddsidemargin}{-0.2in}
\setlength{\headsep}{10pt}
\setlength{\topmargin}{-0.3in}
\setlength{\columnsep}{0.375in}
\setlength{\itemsep}{0pt}
}

\newtheorem{theorem}{Theorem}[section]

\newtheorem{definition}[theorem]{Definition}
\newtheorem{claim}[theorem]{Claim}
\newtheorem{lemma}[theorem]{Lemma}

\newtheorem{corollary}[theorem]{Corollary}

\newtheorem{remark}[theorem]{Remark}

\newtheorem{observation}[theorem]{Observation}

\def\boldhead#1:{\par\vskip 7pt\noindent{\bf #1:}\hskip 10pt}
\def\ithead#1:{\par\vskip 7pt\noindent{\it #1:}\hskip 10pt}

\def\inline#1:{\par\vskip 7pt\noindent{\bf #1:}\hskip 10pt}
\def\midinline#1:{\par\noindent{\bf #1:}\hskip 10pt}
\def\dnsinline#1:{\par\vskip -7pt\noindent{\bf #1:}\hskip 10pt}
\def\ddnsinline#1:{\newline{\bf #1:}\hskip 10pt}
\def\largeinline#1:{\par\vskip 7pt\noindent{\large\bf #1:}\hskip 10pt}
\long\def\commhide #1\commhideend{}
\long\def\commfull #1\commend{#1}
\long\def\commabs #1\commenda{}
\long\def\commtim #1\commendt{#1}
\long\def\commb #1\commbend{}
\long\def\commedit #1\commeditend{} 

\long\def\commB #1\commBend{}       

\long\def\commex #1\commexend{}     

\long\def\commsiena #1\commsienaend{}  

\long\def\commBI #1\commBIend{}  

\long\def\CProof #1\CQED{}
\def\qed{\mbox{}\hfill $\Box$\\}
\def\blackslug{\hbox{\hskip 1pt \vrule width 4pt height 8pt
    depth 1.5pt \hskip 1pt}}
\def\QED{\quad\blackslug\lower 8.5pt\null\par}

\long\def\PPP#1{\noindent{\bf Proof:}{ #1}{\quad\blackslug\lower 8.5pt\null}}

\long\def\denspar #1\densend
{#1}

\newif\ifnotesw\noteswtrue
   {\ifnotesw\marginpar[\hfill\(\top\)]{\(\top\)}\fi}%
   {\ifnotesw\marginpar[\hfill\(\bot\)]{\(\bot\)}\fi}

\newcommand{\mnote}[1]%
    {\ifnotesw\marginpar%
        [{\scriptsize\it\begin{minipage}[t]{\marginparwidth}
        \raggedleft#1%
                        \end{minipage}}]%
        {\scriptsize\it\begin{minipage}[t]{\marginparwidth}
        \raggedright#1%
                        \end{minipage}}%
    \fi}

\def\MathF{\hbox{\rm I\kern-2pt F}}
\def\MathP{\hbox{\rm I\kern-2pt P}}
\def\MathR{\hbox{\rm I\kern-2pt R}}
\def\MathZ{\hbox{\sf Z\kern-4pt Z}}
\def\MathN{\hbox{\rm I\kern-2pt I\kern-3.1pt N}}
\def\MathC{\hbox{\rm \kern0.7pt\raise0.8pt\hbox{\footnotesize I}
\kern-4.2pt C}}
\def\MathQ{\hbox{\rm I\kern-6pt Q}}

\newsavebox{\ttop}\newsavebox{\bbot}

\newcommand{\emp}{\varnothing}
\def\eps{\epsilon}
\def\epsi{\varepsilon}

\newcommand{\mst}{\mathrm{MST}}

\newcommand{\be}{\delta}

\newcommand{\poly}{\mathsf{poly}}
\newcommand{\tw}{\mathsf{tw}}
\newcommand{\ssa}{\mathsf{SSA}}
\newcommand{\sso}{\mathsf{SSO}}
\newcommand{\gsso}{\mathsf{GSSO}}

\newcommand{\pathg}{\mathsf{path~greedy}}
\newcommand{\greedy}{\mathsf{greedy}}
\newcommand{\msttilde}{\widetilde{\mst}}
\newcommand{\Ftilde}{\widetilde{F}}
\newcommand{\Ttilde}{\widetilde{T}}
\newcommand{\Ptilde}{\widetilde{P}}
\newcommand{\Qtilde}{\widetilde{Q}}
\newcommand{\Fbar}{\overline{F}}
\newcommand{\Tbar}{\overline{T}}
\newcommand{\Pbar}{\overline{P}}
\newcommand{\Qbar}{\overline{Q}}
\newcommand{\Ibar}{\overline{I}}

\newcommand{\bmx}{\overline{\mathcal{X}}}
\newcommand{\uncontract}{\mathsf{uctrt}}

\newcommand{\high}{\mathsf{high}}
\newcommand{\highp}{\mathsf{high}^+}

\newcommand{\lowm}{\mathsf{low}^-}

\newcommand{\prune}{\mathsf{pruned}}

\newcommand{\minor}{\mathsf{Minor}}

\newcommand{\oracle}{\mathsf{Oracle}}
\newcommand{\internal}{\mathsf{intrnl}}
\newcommand{\prefix}{\mathsf{pref}}
\newcommand{\light}{\mathsf{light}}

\newcommand{\take}{\mathsf{take}}
\newcommand{\redunt}{\mathsf{redunt}}
\newcommand{\reduce}{\mathsf{reduce}}

\def\eps{\epsilon}

\DeclareMathAlphabet{\mathpzc}{OT1}{pzc}{m}{it}

\newcommand{\dm}{\mathsf{Dm}}
\newcommand{\adm}{\mathsf{Adm}}

\newcommand {\ignore} [1] {}

\denseformat
\newcommand{\doverline}[1]{\widehat{#1}}
\newcommand{\dbar}[1]{\widehat{#1}}

\newcommand{\ma}{\mathcal{A}}
\newcommand{\mb}{\mathcal{B}}
\newcommand{\mc}{\mathcal{C}}
\newcommand{\md}{\mathcal{D}}
\newcommand{\mh}{\mathcal{H}}

\newcommand{\mg}{\mathcal{G}}
\newcommand{\mv}{\mathcal{V}}
\newcommand{\me}{\mathcal{E}}
\newcommand{\ms}{\mathcal{S}}
\renewcommand{\mp}{\mathcal{P}}
\newcommand{\mq}{\mathcal{Q}}
\newcommand{\mk}{\mathcal{K}}
\newcommand{\mt}{\mathcal{T}}
\newcommand{\mx}{\mathcal{X}}
\newcommand{\my}{\mathcal{Y}}
\newcommand{\mz}{\mathcal{Z}}

\newcommand{\mi}{\mathcal{I}}

\newcommand{\wsp}{\mathtt{Ws}}
\newcommand{\ssp}{\mathtt{Ss}}

\newcommand{\mbe}{\mathbf{e}}

\newcommand{\treeClustering}{\textsc{TreeClustering}}
\newcommand{\real}{\mathbb{R}}
\DeclareMathOperator{\MST}{\mathrm{MST}}

\DeclareMathOperator{\defi}{\overset{\mathrm{def.}}{=} }
\DeclareMathOperator{\pr}{\mathtt{Pr}}

\newcommand{\bnu}{\bar\nu}
\newcommand{\bmu}{\bar\mu}

\newcommand{\arc}[1]{{%
		\setbox9=\hbox{#1}%
		\ooalign{\resizebox{\wd9}{\height}{\texttoptiebar{\phantom{A}}}\cr#1}}}

\date{}

\title{A Unified Framework of Light Spanners II: Fine-Grained Optimality}
\author{Hung Le}
\affil{University of Massachusetts Amherst}
\author{Shay Solomon}
\affil{Tel Aviv University}

\begin{document}
\pagenumbering{gobble}
\maketitle
\begin{abstract}
Seminal works on {\em light} spanners over the years provide spanners with optimal {\em lightness} in various graph classes,\footnote{The {\em lightness} is a normalized notion of weight: a graph's lightness is the ratio of its weight to the MST weight.}  such as in general graphs~\cite{CW16}, Euclidean spanners \cite{das1994fast} and minor-free graphs~\cite{BLW17}.
Three shortcomings of previous works on light spanners are: (1) The techniques are ad hoc per graph class, and thus can't be applied broadly.  (2) The runtimes of these constructions are almost always sub-optimal, and usually far from optimal.
(3) These constructions are optimal in the standard and crude sense, but not in a refined sense that takes into account a wider range of involved parameters.

This work aims at addressing these shortcomings by presenting  a {\em unified framework} of light spanners in a variety of graph classes. Informally, the framework boils down to a {\em transformation} from sparse spanners to light spanners; since the state-of-the-art for sparse spanners is much more advanced than that for light spanners, such a transformation is powerful. 
Our framework is developed in two papers. 
{\bf The current paper is the second of the two ---  it  builds on the basis of the unified framework laid in the first paper, 
and then strengthens it to achieve {\em more refined} optimality bounds} for several graph classes, i.e., the bounds remain optimal when taking into account a \emph{wider range of involved parameters},
most notably $\eps$, but also others such as the dimension (in Euclidean spaces) or the minor size (in minor-free graphs). 
Our new constructions are significantly better than the state-of-the-art {\em for every examined graph class}. Among various applications and implications of our framework, we highlight the following:
\noindent
\vspace{5pt}
\\
For $K_r$-minor-free graphs, we provide a  $(1+\epsilon)$-spanner with lightness $\tilde{O}_{r,\epsilon}( \frac{r}{\epsilon} + \frac{1}{\epsilon^2})$,
where $\tilde{O}_{r,\epsilon}$ suppresses $\mathsf{polylog}$ factors of $1/\epsilon$ and $r$,
improving the lightness bound $\tilde{O}_{r,\epsilon}( \frac{r}{\epsilon^3})$ of Borradaile, Le and Wulff-Nilsen~\cite{BLW17}.
We complement our upper bound with a highly nontrivial lower bound construction, for which any $(1+\epsilon)$-spanner must have lightness $\Omega(\frac{r}{\epsilon} + \frac{1}{\epsilon^2})$.
Interestingly, our lower bound is realized by a geometric graph in $\mathbb{R}^2$.
We note that the quadratic dependency on $1/\eps$ we proved here is surprising,
as the prior work suggested that the dependency on $\eps$ should be around $1/\eps$. 
Indeed, for minor-free graphs there is a known upper bound of lightness $O(\log(n)/\eps)$,
whereas subclasses of minor-free graphs, primarily graphs of genus bounded by $g$, are long known to admit spanners of lightness $O(g/\eps)$.  

\end{abstract}

\pagebreak

\tableofcontents

\pagenumbering{arabic}

\clearpage

\section{Introduction}\label{sec:intro}

For a weighted graph $G = (V,E,w)$ and a {\em stretch parameter} $t \ge 1$, a subgraph $H = (V,E')$ of $G$
is called a \emph{$t$-spanner} if $d_H(u,v) \le t \cdot d_G(u,v)$, for every $e = (u,v) \in E$,
where $d_G(u,v)$ and $d_H(u,v)$ are the distances between $u$ and $v$ in $G$ and $H$, respectively.
Graph spanners were introduced in two celebrated papers from 1989 \cite{PS89,PU89} for unweighted graphs,
where it is shown that for any $n$-vertex graph $G = (V,E)$ and integer $k \ge 1$, there is an $O(k)$-spanner with $O(n^{1+ 1/k})$ edges.
We shall sometimes use a normalized notion of size, {\em sparsity}, which is the ratio of the size of the spanner to the size of a spanning tree, namely $n-1$.
Since then, graph spanners have been extensively studied, both for general weighted graphs and for restricted graph families,
such as Euclidean spaces and minor-free graphs.  
In fact, spanners for Euclidean spaces---{\em Euclidean spanners}---were studied implicitly already in the pioneering SoCG'86 paper of Chew~\cite{Chew86}, who showed that any 2-dimensional Euclidean space admits a spanner of $O(n)$ edges and stretch $\sqrt{10}$, and later improved the stretch to 2~\cite{Chew89}.

As with the sparsity parameter, its weighted variant---lightness---has been extremely well-studied; the \emph{lightness} is the ratio of the weight of the spanner to $w(MST(G))$. 
Seminal works on {\em light} spanners over the years provide spanners with optimal {\em lightness} in various graph classes, such as in general graphs~\cite{CW16}, Euclidean spanners \cite{das1994fast} and minor-free graphs~\cite{BLW17}.
{\bf Despite the large body of work on light spanners,  the stretch-lightness tradeoff is not nearly as well-understood as the stretch-sparsity tradeoff}, and the intuitive reason behind that is clear: Lightness seems inherently more challenging to optimize than sparsity, since different edges may contribute disproportionately to the overall lightness due to differences in their weights.  The three shortcomings of light spanners that emerge, when considering the large body of work in this area, are: (1) The techniques are ad hoc per graph class, and thus can't be applied broadly 
(e.g., some require large stretch and are thus suitable to general graphs, while others are naturally suitable to stretch $1 + \eps$). 
(2) The runtimes of these constructions are usually far from optimal.
(3) These constructions are optimal in the standard and crude sense, but not in a refined sense that takes into account a wider range of involved parameters.

In this work, we are set out to address these shortcomings by presenting  a {\em unified framework} of light spanners in a variety of graph classes. Informally, the framework boils down to a {\em transformation} from sparse spanners to light spanners; since the state-of-the-art for sparse spanners is much more advanced than that for light spanners, such a transformation is powerful.

Our framework is developed in two papers. 
{\bf The current paper is the second of the two --- it builds on the {basis of the unified framework} laid in the first paper, and strengthens it to achieve {\em fine-grained optimality}}, i.e., constructions that are {optimal even when taking into account a wider range of involved parameters}. Our ultimate goal is to bridge the gap in the understanding between light spanners and sparse spanners.  This gap is very prominent when considering constructions through the lens of fine-grained optimality.  
Indeed, the state-of-the-art spanner constructions for general graphs, as well as for most restricted graph families, incur a (multiplicative) $(1+\eps)$-factor slack on the stretch with a suboptimal dependence on $\eps$ as well as other parameters in the lightness bound. 
To exemplify this statement, we next survey results on light spanners in several basic graph classes. Subsequently, we present our new constructions, all of which are derived as applications and implications of the unified framework developed in this work. Our constructions are significantly better than the state-of-the-art {\em for every examined graph class}. 
{\bf Our main result is for minor-free graphs, where we achieve tight dependencies on both $\eps$ and the minor size parameter} --- the upper bound follows as an application of the unified framework whereas the lower bound is obtained by different means.

\paragraph{General weighted graphs.~} 
The aforementioned results of \cite{PS89,PU89} for general graphs were strengthened in \cite{ADDJS93}, where it was shown that for every $n$-vertex \emph{weighted} graph $G = (V,E,w)$ and integer $k \ge 1$, there is a {\em greedy} algorithm for constructing a $(2k-1)$-spanner with $O(n^{1+1/k})$ edges, which is optimal under Erd\H{o}s' girth conjecture. Thus, the stretch-sparsity tradeoff is resolved up to the girth conjecture. 

The stretch-lightness tradeoff, on the other hand, is still far from being resolved.  
Alth\"{o}fer et al.~\cite{ADDJS93}   showed that the lightness of the greedy spanner is $O(n/k)$. Chandra et al.~\cite{CDNS92} improved this lightness bound to $O(k \cdot n^{(1+\eps)/{(k-1)}} \cdot (1/\eps)^2)$, for any $\eps > 0$;
another, somewhat stronger, form of this tradeoff from \cite{CDNS92}, is stretch $(2k-1)\cdot(1+\eps)$,
$O(n^{1+1/k})$ edges and lightness $O(k \cdot n^{1/{k}} \cdot (1/\eps)^{2})$.
In a sequence of works from recent years \cite{ENS14,CW16,FS16},
it was shown that the lightness of the greedy spanner is
$O(n^{1/k} (1/\eps)^{3+2/k})$ (this lightness bound is due to \cite{CW16}; the fact that this bound holds   for the greedy spanner is due to \cite{FS16}).
We note that the previous best known dependence on $1/\epsilon$ for near-optimal lightness bound $O(n^{1/{k}})$ (as a function of $n$ and $k$)
is super-cubic~\cite{CW16,FS16}. 

\paragraph{Minor-free graphs}
A graph $H$ is called a \emph{minor} of graph $G$ if $H$ can be obtained from $G$ by deleting edges and vertices and by contracting edges. A graph $G$ is said to be {\em $K_r$-minor-free},
if it excludes  $K_r$ as a minor for some fixed $r$, where $K_r$ is the complete graph on $r$ vertices. (We shall omit the prefix $K_r$ in the term ``$K_r$-minor-free'',   when the value of $r$ is not important.)

The gap between sparsity and lightness is prominent in minor-free graphs, for stretch $1+\eps$.
Indeed, minor-free graphs are sparse to begin with, and no further edge sparsification is possible for stretch $2-\eps$ (let alone $1+\eps$),
thus the sparsity of $(1+\eps)$-spanners in minor-free graphs is trivially $\tilde \Theta(r)$. 
(E.g., consider a path that connects $n/(r-1)$ vertex-disjoint copies of $K_{r-1}$;
the only $(2-\eps)$-spanner of such a graph, which is $K_r$-minor free and has $\Theta(n r)$ edges, is itself.)
On the other hand, for lightness, bounds are much more interesting.
Borradaile, Le, and Wulff-Nilsen~\cite{BLW17} showed that the greedy $(1+\eps)$-spanners of $K_r$-minor-free graphs have lightness $\tilde{O}_{r,\epsilon}(\frac{r}{\epsilon^3})$, where the notation $\tilde{O}_{r,\epsilon}(.)$ hides polylog factors of $r$ and $\frac{1}{\epsilon}$. Moreover, this is the state-of-the-art lightness bound also in some sub-classes of minor-free graphs, particularly bounded treewidth graphs.

Past works provided strong evidence that the dependence of lightness on $1/\epsilon$ of $(1+\epsilon)$-spanners
should be \emph{linear}: $O(\frac{1}{\epsilon})$ in planar graphs by Alth\"{o}fer et al.~\cite{ADDJS93}, 
$O(\frac{g}{\epsilon})$ in bounded genus graphs by Grigni~\cite{Grigni00}, and $\tilde{O}_{r}(\frac{r \log n}{\epsilon})$ in $K_r$-minor-free graphs by Grigni and Sissokho~\cite{GS02}. (The $\log n$ factor in the lightness bound of~\cite{GS02} was removed by~\cite{BLW17} at the cost of a cubic dependence on $1/\epsilon$.) 

\paragraph{Low-dimensional Euclidean spaces.~}
Low-dimensional Euclidean spaces is another class of graphs for which sparsity is much better understood than lightness. The authors of this paper showed in ~\cite{LS19} the existence of point sets $P$ in $\mathbb{R}^d$, $d = O(1)$, for which any $(1+\epsilon)$-spanner for $P$ must have sparsity $\Omega(\epsilon^{-d+1})$ and lightness $\Omega(\epsilon^{-d})$, when  $\epsilon = \Omega(n^{-1/(d-1)})$. The sparsity lower bound matched the long-known upper bound of $O(\epsilon^{-d+1})$, realized by various spanner constructions, including the greedy spanner~\cite{ADDJS93,CDNS92,NS07}, the {\em $\Theta$-graph} and {\em Yao graph} \cite{Yao82, Clarkson87,Keil88,KG92,RS91,ADDJS93}, and the gap-greedy spanner~\cite{Salowe92,AS97}.
While all the aforementioned sparsity upper bounds are tight and rather simple,  the best lightness upper bound prior to \cite{LS19} was $O(\epsilon^{-2d})$ \cite{NS07} (building on \cite{ADDJS93,DHN93,DNS95,RS98}), 
which is quadratically larger than the lower bound; moreover, it uses a very complex argument to analyze the lightness of the greedy spanner. In~\cite{LS19}, the authors improved the analysis of the greedy spanner by \cite{NS07} to achieve a lightness bound of $\tilde{O}(\epsilon^{-d})$,
matching their lower bound (up to a factor of $\log(1/\epsilon)$); the improved upper bound argument of \cite{LS19} is also very complex.
  
In the same paper \cite{LS19}, the authors studied {\em Steiner spanners}, namely, spanners that are allowed to use 
{\em Steiner points}, which are additional points that are not part of the input point set.
It was shown there that Steiner points can be used to improve the sparsity quadratically, i.e., to $O(\epsilon^{\frac{-d+1}{2}})$,
which was shown to be tight for dimension $d = 2$ in \cite{LS19}, and for any $d = O(1)$ by Bhore and T\'{o}th~\cite{BT21B}.

An important question left open in~\cite{LS19} is whether one could use Steiner points to improve the lightness bound to $o(\epsilon^{-d})$.
In~\cite{LS20}, the authors made  the first progress on this question by showing that  any point set $P \in \mathbb{R}^d$ with spread $\Delta(P)$ admits a Steiner $(1+\epsilon)$-spanner with lightness $O(\frac{\log (\Delta(P))}{\epsilon})$ when $d = 2$ and with lightness $\tilde{O}(\epsilon^{-(d+1)/2} + \epsilon^{-2} \log (\Delta(P)))$ when $d\geq 3$~\cite{LS20}. In particular, when $\Delta(P) = \poly(\frac{1}{\epsilon})$, the lightness bounds are $\tilde{O}(\frac{1}{\epsilon})$ when $d = 2$ and $\tilde{O}(\epsilon^{-(d+1)/2})$ when $d \geq 3$.  However, $\Delta(P)$ could be huge, and it could also depend on $n$. Bhore and T\'{o}th~\cite{BT21} removed the dependency on $\Delta(P)$ for   $d=2$ by showing that any point set $P \in \mathbb{R}^2$ admits a Steiner  $(1+\epsilon)$-spanner with lightness $O(\frac{1}{\eps})$. 
The question of whether one can achieve lightness $o(\epsilon^{-d})$ for $d \geq 3$ (for any spread) remains open.

\paragraph{High-dimensional Euclidean metric spaces.~} 
The literature on spanners in high-dimensional Euclidean spaces is surprisingly sparse.
Har-Peled, Indyk and Sidiropoulos~\cite{HIS13} showed that for any set of $n$-point Euclidean space (in any dimension) 
and any parameter $t \geq 2$, there is an $O(t)$-spanner with sparsity $O(n^{1/t^2} \cdot (\log n \log t))$.  Filtser and Neiman~\cite{FN18} gave an analogous but weaker result for lightness, achieving a lightness bound of $O(t^3 n^{\frac{1}{t^2}}\log n)$. They also generalized their results to any $\ell_p$ metric, for $p \in (1,2]$, achieving a lightness bound of $O(\frac{t^{1+p}}{\log^2 t}n^{\frac{\log^2 t}{t^p}}\log n)$.  

\subsection{Research Agenda: From Sparse to Light Spanners}

Thus far we exemplified the statement that the stretch-lightness tradeoff is not as well-understood as the stretch-sparsity tradeoff,
when considering {fine-grained dependencies}. In the companion paper~\cite{LS21}, we exemplified the statement when considering the construction time.  This statement is not to underestimate in any way the exciting line of work on light spanners,
but rather to call for attention to the important research agenda of narrowing this gap and ideally closing it.

\paragraph{Fine-grained optimality.~}  A fine-grained optimization of the stretch-lightness tradeoff, which takes into account the exact dependencies on $\eps$ and the other involved parameters, is a highly challenging goal. For planar graphs, the aforementioned result~\cite{ADDJS93} on the greedy $(1+\eps)$-spanner with lightness $O(1/\eps)$ provides an optimal dependence on $\eps$ in the lightness bound, due to a matching lower bound.
For constant-dimensional Euclidean spaces, the aforementioned result on the greedy $(1+\eps)$-spanner with lightness $\Theta(\eps^{-d})$ was achieved recently \cite{LS19}. We are not aware of any other well-studied graph classes for which such fine-grained optimality is known. 
Achieving fine-grained optimality is of particular importance for graph families that admit light spanners with stretch $1+\eps$, such as minor-free graphs and Euclidean spaces, in spanner applications where precision is a necessity. Indeed, in such applications, the precision is basically determined by $\eps$, hence if it is a tiny (sub-constant) parameter, then improving the $\eps$-dependence on the lightness could lead to significant improvements in the performance.

\begin{goal} \label{g2}
Achieve fine-grained optimality for light spanners in basic graph families. 
\end{goal}

\paragraph{Fast constructions.~}
The companion paper revolves around the following question: Can one achieve {\em fast constructions} of light spanners that {\em match} the corresponding results for sparse spanners? 

\begin{goal} \label{g1}
Achieve {\em fast constructions} of light spanners that {\em match} the corresponding constructions of sparse spanners. 
In particular, achieve (nearly) linear-time constructions of spanners with optimal lightness for basic graph families, such as the ones covered in the aforementioned questions. 
\end{goal}

\paragraph{Unification.~}
Some of the papers on light spanners employ inherently different techniques than others, e.g., the technique of \cite{CW16} requires large stretch while others are naturally suitable to stretch $1+\eps$.
Since the techniques in this area are ad hoc per graph class, they can't be applied broadly.
A unified framework for light spanners would be of both theoretical and practical merit.
\begin{goal} \label{g3}
Achieve a unified framework of light spanners.
\end{goal}

Establishing a thorough understanding of light spanners by meeting (some of) the above goals is not only of theoretical interest, but is also of practical importance, due to the wide applicability of spanners.  Perhaps the most prominent applications of light spanners are to efficient broadcast protocols in the message-passing model of distributed computing \cite{ABP90,ABP91},
to network synchronization and computing global functions \cite{Awerbuch85,PU89,ABP90,ABP91,Peleg00}, and to the TSP \cite{Klein05,Klein06,RS98,GLN02,BLW17,Gottlieb15}.
There are many more applications, such as to data gathering and dissemination tasks in overlay networks \cite{BKRCV02,VWFME03,KV01}, 
to VLSI circuit design \cite{CKRSW91,CKRSW292,CKRSW92,SCRS01},
to wireless and sensor networks \cite{RW04,BDS04,SS10}, 
to routing \cite{WCT02,PU89,PU89b,TZ01}, 
to compute almost shortest paths \cite{Cohen98,RZ11,Elkin05,EZ06,FKMSZ05},
and to computing distance oracles and labels \cite{Peleg00Prox,TZ01b,RTZ05}.

\subsection{Our Contribution} \label{subsec:contribution}
Our work aims at meeting the above goals (\Cref{g2}---\Cref{g3}) by presenting a unified framework for optimal constructions of light spanners in a variety of graph classes.
Basically, we strive to translate results --- in a unified manner --- from sparse spanners to light spanners, without significant loss in any parameter. 

As mentioned, the current paper is the second of two, building on the {\em basis of the framework} laid in the first paper,  aiming to achieve fine-grained optimality. Such a fine-grained optimization is highly challenging, and towards meeting this goal we had to give up on the running time bounds achieved in the companion paper. Thus the current paper achieves \Cref{g2} and \Cref{g3} whereas the companion paper achieves \Cref{g1} and \Cref{g3};
achieving all three goals simultaneously is left open by our work. 

Next, we elaborate on the applications and implications of our framework, and put it into context with previous work. 

\paragraph{$K_r$-minor-free graphs.~}  The most important implication of our framework is to minor-free graphs,
where we improve the $\eps$-dependence in the lightness bound of \cite{BLW17}; as will be asserted in \cref{thm:minor-free-lowerbound},
our improved lightness bound is tight.

\begin{restatable}{theorem}{MinorFree}
	\label{thm:minor-free-opt-lightness}
	Any $K_r$-minor-free graph admits a $(1+\epsilon)$-spanner with lightness $\tilde{O}_{r,\epsilon}(\frac{r}{\epsilon} + \frac{1}{\epsilon^2})$ for any $\epsilon< 1$ and $r\geq 3$. 
\end{restatable}

The $\tilde{O}_{\eps,r}(.)$ notation in \Cref{thm:minor-free-opt-lightness} hides a poly-logarithmic factor of $1/\eps$ and $r$.   The quadratic dependence on $\frac{1}{\epsilon}$ in the lightness bound of \Cref{thm:minor-free-opt-lightness}  may seem artificial;
indeed, as mentioned already, past works \cite{ADDJS93, Grigni00,GS02} provided evidence that the dependence on $1/\epsilon$ 
in the lightness bound of $(1+\epsilon)$-spanners should be \emph{linear}. Surprisingly perhaps, we show that the quadratic dependence on $\frac{1}{\epsilon}$ in the lightness bound of \Cref{thm:minor-free-opt-lightness} is required:

\begin{theorem}\label{thm:minor-free-lowerbound}
For any fixed $r\geq 6$, any  $\epsilon < 1$  and $n \geq r + (\frac{1}{\epsilon})^{\Theta(1/\epsilon)}$, there is an $n$-vertex graph $G$ excluding $K_r$ as a minor for which any $(1+\eps)$-spanner must have lightness $\Omega(\frac{r}{\epsilon} + \frac{1}{\epsilon^2})$.
\end{theorem}	
	
We remark that, in \Cref{thm:minor-free-lowerbound}, the exponential dependence on $1/\epsilon$ in the lower bound on $n$ is unavoidable since, if $n = \mathrm{poly}(1/\epsilon)$, the result of ~\cite{GS02} yields a lightness of
$\tilde{O}_r(\frac{r}{\epsilon}\log(n)) = \tilde{O}_{r,\epsilon}(\frac{r}{\epsilon})$.

Interestingly, our lower bound  applies to a geometric graph, where the vertices correspond to points in $\mathbb R^2$ and the edge weights are the Euclidean distances between the points. The construction is recursive. We start with a basic gadget and then recursively ``stick'' many copies of the same basic gadgets in a fractal-like structure. We use geometric considerations to show that any $(1+\eps)$-spanner must take every edge of this graph, whose total edge weight is $\Omega(1/\eps^2)w(\mst)$. The resulting graph has treewidth at most $4$; by a simple modification, we obtain the lower bound for any $K_r$-minor-free graphs as claimed in \Cref{thm:minor-free-lowerbound}.

\paragraph{General graphs.~} 
For general graphs we prove the following result.

\begin{theorem}\label{thm:light-general-spanner}
Given an edge-weighted graph $G(V,E)$ and two parameters $k \geq 1, \epsilon < 1$, there is a $(2k-1)(1+\epsilon)$-spanner of $G$ with lightness $O(n^{1/k}/{\epsilon})$.
\end{theorem}

The spanner construction provided by \Cref{thm:light-general-spanner}
provides the first improvement over the super-cubic dependence on $1/\eps$ in the lightness bound of $O(n^{1/k}(1/\eps)^{3+2/k})$ of  \cite{CW16}.  Moreover, by substituting $\eps$ with $\eta / k$, for an arbitrarily small constant $\eta \le 1$, we get a stretch arbitrarily close to $2k-1$ with lightness $O({n^{1/k} \cdot k})$,
whereas all previous spanner constructions for general graphs with stretch at most $2k$ have lightness $\Omega(n^{1/k} \cdot k^2 / \log k)$ \cite{CDNS92,ENS14,CW16}, which is bigger by a factor of at least $k / \log k$.

\paragraph{Low-dimensional Euclidean Spaces.~} 
We prove the following result, which provides a near-quadratic improvement over the lightness bound of \cite{LS19}. 

 \begin{theorem}\label{thm:light-Steiner}
 For any $n$-point set $P \in \mathbb{R}^d$ and any $d \ge 3$, $d = O(1)$, there is a Steiner $(1+\epsilon)$-spanner for $P$ with lightness 
 $\tilde{O}(\epsilon^{-(d+1)/2})$ that is constructable in polynomial time.
 \end{theorem}

The lightness bound in \Cref{thm:light-Steiner} has no dependence whatsoever on $\Delta(P)$ for any $d\geq 3$, $d = O(1)$. This lightness bound nearly matches (up to a factor of $\sqrt{1/\eps}$) the recent lower bound of $\Omega(\epsilon^{-d/2})$ by
Bhore and T{\'{o}}th \cite{BT21B}, for any $d = O(1)$.

\paragraph{High dimensional Euclidean metric spaces.~}

We prove the following result. 

\begin{theorem}\label{thm:Euclidean-high} For any $n$-point set $P$ in a Euclidean space and any given $t \ge 2$, there is an $O(t)$-spanner for $P$ with lightness 
	$O(tn^{\frac{1}{t^2}}\log n)$ that is constructable in polynomial time.
\end{theorem}

Recall that the previous state-of-the-art lightness bound is  $O(t^3 n^{\frac{1}{t^2}}\log n)$ \cite{FN18};
e.g., when $t = \sqrt{\log n}$, the lightness of our spanner is $O(\log^{3/2} n)$ while the lightness bound of~\cite{FN18} is $O(\log^{5/2} n)$. 

We also extend \Cref{thm:Euclidean-high} to any $\ell_p$ metric, for $p \in (1,2]$,
which improves over the lightness bound $O(\frac{t^{1+p}}{\log^2 t}n^{\frac{\log^2 t}{t^p}}\log n)$ of~\cite{FN18}.

\begin{theorem}\label{thm:Lp-high} For any $n$-point $\ell_p$ normed space $(X,d_X)$ with $p \in (1,2]$ and any $t \ge 2$, there is an $O(t)$-spanner for $(X,d_X)$ with lightness $O(t n^{\frac{\log^2 t}{t^p}}\log n)$.
\end{theorem}

\subsection{A Unified Framework}\label{subsec:unified-framework-intro}

In this section, we give a high-level overview of our framework for constructing light spanners with stretch $t(1+\eps)$, for some parameter $t$ that depends on the examined graph class; e.g., for Euclidean spaces $t  = 1+\eps$, while for general graphs $t = 2k-1$.  Let $L$ be a positive parameter, and $H_{< L}$ be a subgraph of  $G = (V,E,w)$. Our framework relies on the notion of a {\em cluster graph}, defined as follows.

\begin{definition}[$(L,\eps,\beta)$-Cluster Graph~\cite{LS21}]\label{def:ClusterGraph-Param} An edge-weighted graph $\mg= (\mv,\me,\omega)$ is an \emph{$(L,\eps,\beta)$-cluster graph} w.r.t a subgraph $H_{< L}$ for some constant $\beta$ if:
	\begin{enumerate}
		\item Each node $\varphi_C \in \mv$ corresponds to a subset of vertices $C \in V$, called a \emph{cluster}. For any two different nodes $\varphi_{C_1}, \varphi_{C_2}$ in $\mv$, $C_1\cap C_2 = \emptyset$.
		\item Each edge $(\varphi_{C_1},\varphi_{C_2})\in \me$ corresponds to an edge $(u,v)\in E$ such that $u \in C_1$ and $v\in C_2$. Furthermore, $\omega(\varphi_{C_1},\varphi_{C_2}) = w(u,v)$.
		\item $L \leq \omega(\varphi_{C_1},\varphi_{C_2}) < 2L$ for every edge $(\varphi_{C_1},\varphi_{C_2})\in \me$.
		\item $\dm(H_{< L}[C]) \leq \beta \eps L$ for any cluster $C$ corresponding to a node $\varphi_C \in \mv$.  
	\end{enumerate} 
	Here $\dm(X)$ denotes the diameter of a graph $X$. 
\end{definition} 
Condition (1) asserts that clusters corresponding to nodes of $\mg$ are vertex-disjoint. Furthermore, Condition (4) asserts that they induce subgraphs of low diameter in $H_{< L}$. In particular, if $\beta$ is constant, then the diameter of clusters is roughly $\eps$ times the weight of edges in the cluster graph.  

The idea of using the cluster graph is to select a subset of edges of $G$ to add to a subgraph $H_{< L}$, which will be a spanner for edges of weights less than $L$ in our construction, to obtain a spanner for edges of weights less than $2L$, thereby extending the set of edges whose endpoints' distances are preserved. By repeating the same construction for edges of higher and higher weights, we eventually obtain a spanner that preserves distances for every pair of vertices in $G$.

Our framework assumes the existence of the following algorithm, called \emph{sparse spanner oracle ($\sso$)}, which computes a subset of edges in $G$ to add to $H_{< L}$.

\begin{tcolorbox}
	\hypertarget{SPHigh}{}
	\textbf{$\sso$:} Given an $(L,\eps,\beta)$-cluster graph $\mg(\mv,\me,\omega)$, 
	the $\sso$ outputs a  subset of edges  $F$ in polynomial time such that: 
	\begin{enumerate}[noitemsep]
		\item \textbf{(Sparsity)~} \hypertarget{Sparsity}{} $w(F) \leq \chi|\mv| L_i$ for some $\chi> 0$. 
		\item \textbf{(Stretch)~} \hypertarget{Stretch}{} For each edge $(\varphi_{C_u},\varphi_{C_v})\in \me$, $d_{H_{<2L}}(u,v)\leq t(1+s_{\sso}(\beta)\eps)w(u,v)$ 		where $(u,v)$ is the corresponding edge of $(\varphi_{C_u}, \varphi_{C_v})$ and $s_{\sso}(\beta)$ is some constant that depends on $\beta$  only, and $H_{< 2L}$  is the graph obtained by adding $F$ to $H_{< L}$. 
	\end{enumerate}	
	
\end{tcolorbox}

We can interpret the $\sso$ as a construction of a {\em sparse spanner} in the following way: If $F$ contains only edges of $G$ corresponding to a subset of $\me$, say $\me^{\prune} \subseteq \me$, then, $w(e)\geq L$ for every $e \in F$; in this case $|F| \leq \chi|\mv|$.  Importantly, for all classes of graphs considered in this paper, the implementation of $\sso$ is very simple, as we show in \Cref{sec:app}.   The highly nontrivial part of the framework is given by the following theorem, which provides a {\em black-box transformation} from an $\sso$ to an efficient {\em meta-algorithm} for constructing light spanners. We note that this transformation remains the same across all graphs.

\begin{restatable}{theorem}{Framework}
	\label{lm:framework} Let $L,\eps, t, \beta \geq 1$ be parameters where $\beta$ only takes on constant values, and $\eps \ll 1$. 
	Let $\mathcal{F}$ be an arbitrary graph class.
	If, for any graph $G$ in $\mathcal{F}$, the $\sso$ can take any $(L,\eps,\beta)$-cluster graph $\mg(\mv,\me,\omega)$ corresponding to $G$  as input and return as output a subset of edges $F$ of $G$ satisfying the aforementioned two properties of (\hyperlink{Sparsity}{Sparsity}) and (\hyperlink{Stretch}{Stretch}),
	then for any graph in $\mathcal{F}$ we can construct a spanner with stretch $t(1+(2s_{\sso}(O(1))+O(1))\eps)$, lightness $\tilde{O}_{\eps}((\chi\eps^{-1} + \eps^{-2}))$ when $t = 1+\eps$, and lightness  $\tilde{O}_{\eps}((\chi\eps^{-1}))$ when $t\geq 2$. 
\end{restatable}

We remark the following regarding \Cref{lm:framework}.

\begin{remark}\label{remark:ACTIntro}  Parameter $\beta$ only takes on constant values, and $\eps$ is bounded inversely by  $\beta$. In all constructions in \Cref{sec:app}, $\eps \leq 1/(4\beta)$.
\end{remark}

Our framework here builds on the framework developed in our companion work~\cite{LS21}. In particular, in~\cite{LS21}, we assume the existence of a (nearly) linear-time \emph{sparse spanner algorithm} ($\ssa$) to select edges from the cluster graph.  It was then shown (Theorem 1.7 in \cite{LS21}) that the $\ssa$ can be used as a black-box to obtain a (nearly) linear-time construction of light spanners, in a way that is analogous to how \Cref{lm:framework} uses the $\sso$. Our framework here strengthens the framework in our companion work~\cite{LS21} in three different aspects. First, edges in the set $F$  produced by the $\sso$ may not correspond to edges in $\me$ of $\mg$. This allows for more flexibility in choosing the set of edges to add to $H_{< L}$, and is the key to obtaining a fine-grained optimal dependencies on $\eps$ and the other parameters, such as the Euclidean dimension or the minor size. Second, the $\eps$-dependence in the lightness bound of \Cref{lm:framework} is better than that in Theorem 1.7 from~\cite{LS21}; in the constructions presented in this paper, this dependence is optimal as we explain below. Third, we no longer require the graph $H_{<L}$ used in the definition of \Cref{def:ClusterGraph-Param} to preserve distances less than $L$, as required in ~\cite{LS21}. 

The transformation from sparsity to lightness in \Cref{lm:framework} only looses a factor of $1/\eps$ for stretch $t\geq 2$, and, in addition, another additive term of $+\frac{1}{\eps^2}$ is lost for stretch $t =  1+\eps$. 
Later, we complement this upper bound by  a lower bound (\Cref{sec:lowerbounds}) showing that for $t = 1+\eps$, the additive  term  of $+\frac{1}{\eps^2}$ is unavoidable in the following sense: There is a graph class --- the class of bounded treewidth graphs --- where we can implement an $\sso$ with $\chi = O(1)$ for stretch $(1+\eps)$, and hence the lightness of the transformed spanner is $O(1/\eps^2)$ due to the additive term of $+\frac{1}{\eps^2}$, but any light $(1+\eps)$-spanner for this class of graphs must have lightness $\Omega(1/\eps^2)$.

\Cref{lm:framework} is a powerful tool for constructing light spanners that achieve fine-grained dependency on $\epsilon$ and other parameter in the lightness. 
Its proof builds on the basis of the framework laid   in our companion paper~\cite{LS21}, which itself is highly intricate, but it is even more complex. In particular, we also use  a hierarchy of clusters, potential function, and the notion  of augmented diameter for clusters as in \cite{LS21}. However, our goal here is to minimize the $\eps$-dependence.  To this end, we construct clusters in such a way that (1) a cluster at a higher level should contain as many clusters as possible, called subclusters, at lower levels, and (2) the augmented diameter of the cluster must be within a restricted bound. Condition (1) implies that each cluster has a large potential change, which is used to ``pay'' for spanner edges that the algorithm adds to the spanner, while condition (2) implies that the constructed spanner has the desired stretch. The two conditions are in conflict with each other, since the more subclusters we have in a single cluster, the larger the diameter of the cluster gets. Achieving the right balance between these two conflicting conditions is the main technical contribution of this paper.

Another strength of our framework (provided in \Cref{lm:framework}) is its flexibility. Specifically, in \Cref{subsec:oracle-intro}, we introduce another layer of abstraction via an object that we call \emph{general sparse spanner oracle} ($\gsso$). Informally, GSSO is an algorithm that constructs a sparse spanner for any given \emph{subset of vertices} of the input class of graphs (see \Cref{subsec:oracle-intro} for a formal definition).   A shared property of all graph classes for which we construct a $\gsso$ is that they come from a class of metrics that is closed under taking submetrics. In the construction of $\gsso$ in \Cref{sec:app}, we exploit this property by simply running a known sparse spanner construction on top of the subset of vertices given to the $\gsso$. However, this property does not apply to minor-free graphs. Thus, to establish  the lightness upper bound $\tilde{O}_{r,\epsilon}(\frac{r}{\epsilon} + \frac{1}{\epsilon^2})$ of \Cref{thm:minor-free-opt-lightness}, we directly implement $\sso$.

\subsubsection{General Sparse Spanner Oracles}\label{subsec:oracle-intro}

Next, we introduce the notion of a general sparse spanner oracle ($\gsso$), and show that by feeding the $\gsso$ to our framework in \Cref{lm:framework}, we can obtain light spanners from $\gsso$. 
Our $\gsso$ for stretch $t = (1+\eps)$ coincides with a notion called {\em spanner oracle}, introduced by Le~\cite{LS20}. Our focus in this paper is to optimize the $\eps$-dependence, and to do so while considering a much wider regime of the stretch parameter $t$, which could also depend on $n$. 
As mentioned, $\gsso$ is basically an abstraction layer over the $\sso$, which we use to derive most but not all of the results in this paper;
in particular, to achieve our results for minor-free graphs, we need to work directly on the $\sso$.

\begin{definition}[General Sparse Spanner Oracle]\label{def:oracle} Let $G$ be  an edge-weighted graph and let $t > 1$ be a stretch parameter. A general sparse spanner oracle ($\gsso$) of $G$ for a given stretch $t$ is an algorithm that, given a subset of vertices  $T\subseteq V(G)$ and a distance parameter $L > 0$, outputs in \emph{polynomial time} a subgraph $S$ of $G$ such that for every pair of vertices $x,y \in T, x\not= y$ with $L \leq d_G(x,y) < 2L$:
	\begin{equation}
		d_{S}(x,y)\leq t\cdot d_G(x,y).
	\end{equation}	
	We denote a $\gsso$ of $G$ with stretch $t$ by  $\mathcal{O}_{G,t}$, and its output subgraph is denoted by $\mathcal{O}_{G,t}(T,L)$, given two parameters $T\subseteq V(G)$ and $L >0$.
\end{definition}

\begin{definition}[Sparsity]\label{def:sparsity} Given a $\gsso$ $\mathcal{O}_{G,t}$ of a graph $G$, we define weak sparsity and strong sparsity of $\mathcal{O}_{G,t}$, denoted by $\wsp_{\mathcal{O}_{G,t}}$ and $\ssp_{\mathcal{O}_{G,t}}$ respectively, as follows:
	\begin{equation}\label{eq:wsp-ssp}
		\begin{split}
			\wsp_{\mathcal{O}_{G,t}} &= \sup_{T\subseteq V, L \in \real^+}\frac{w\left(\mathcal{O}_{G,t}(T,L)\right)}{|T|L}\\
			\ssp_{\mathcal{O}_{G,t}} &=  \sup_{T\subseteq V, L \in \real^+} \frac{|E\left( \mathcal{O}_{G,t}(T,L)\right)|}{|T|}
		\end{split}
	\end{equation}	
\end{definition}
\noindent We observe that:
\begin{equation}\label{wsp-vs-ssp}
	\wsp_{\mathcal{O}_{G,t}} \leq t\cdot \ssp_{\mathcal{O}_{G,t}},
\end{equation}
since every edge $E\left( \mathcal{O}_{G,t}(T,L)\right)$ must have weight at most $t\cdot L$; indeed, otherwise we can remove it from $ \mathcal{O}_{G,t}(T,L)$ without affecting the stretch.  Thus, when $t$ is a constant, strong sparsity implies weak sparsity; note, however, that this is not necessarily the case when $t$ is super-constant.

We first show that for stretch $t\geq 2$, we can construct a light spanner with lightness bound roughly $O(\frac{1}{\eps})$ times the sparsity of the spanner oracle.

\begin{restatable}{theorem}{GeneralStretchT}
	\label{thm:general-stretch-2} Let $G$ be an arbitrary edge-weighted graph that admits a $\gsso$ $\mathcal{O}_{G,t}$ of weak sparsity $\wsp_{\mathcal{O}_{G,t}}$ for $t\geq 2$. Then for any $\eps > 0$, we can construct in polynomial time a $t(1+\epsilon)$-spanner for $G$ with lightness $\tilde{O}_{\epsilon}\left(\frac{\wsp_{\mathcal{O}_{G,t}}}{\epsilon}\right)$
\end{restatable}

For stretch $t = 1+\eps$, we can construct a light spanner with lightness bound roughly $O(\frac{1}{\eps})$ times the sparsity of the spanner oracle plus \emph{an additive factor $1/\eps^2$}. It turns out that the additive factor $+1/\eps^2$ is unavoidable.

\begin{restatable}{theorem}{GeneralStretchE}
	\label{thm:general-stretch-1eps} Let $G$ be an arbitrary edge-weighted graph  that admits a $\gsso$ $\mathcal{O}_{G,1+\eps}$ of weak sparsity $\wsp_{\mathcal{O}_{G,1+\eps}}$  for any $\eps > 0$. Then there exists an $(1+O(\epsilon))$-spanner for $G$ with lightness $\tilde{O}_{\epsilon}\left(\frac{\wsp_{\mathcal{O}_{G,t}}}{\epsilon} + \frac{1}{\epsilon^2}\right)$.
\end{restatable}

In both \Cref{thm:general-stretch-2} and~\Cref{thm:general-stretch-1eps}, $\tilde{O}_{\epsilon}(.)$ hides a factor $\log \frac{1}{\epsilon}$.

The bound in \Cref{thm:general-stretch-1eps} improves over the lightness bound due to Le~\cite{Le20} by a $\frac{1}{\epsilon^2}$ factor. The stretch of $S$ in \Cref{thm:general-stretch-1eps} is $1+O(\epsilon)$, but we can scale it down to $(1+\eps)$ while increasing the lightness by a constant factor. 
Moreover, this bound is optimal, as we shall assert next.
First, the additive factor $\frac{\wsp_{\mathcal{O}_{G,t}}}{\epsilon}$ is unavoidable: the authors showed in~\cite{LS19} that there exists a set of $n$ points in $\mathbb R^d$ such that any $(1+\epsilon)$-spanner for it must have lightness $\Omega(\epsilon^{-d})$, while Le~\cite{Le20} showed that point sets in $\mathbb R^d$ have $\gsso$es with weak sparsity $O(\epsilon^{1-d})$. 
Second, the additive factor $\frac{1}{\epsilon^2}$ is tight by the following theorem.

\begin{theorem}\label{thm:lb-oracle-1eps}  
	For any $\epsilon < 1$  and $n \geq (\frac{1}{\epsilon})^{\Theta(\frac{1}{\epsilon})}$, there is an $n$-vertex graph $G$  admitting a  $\gsso$ of stretch $(1+\eps)$ with weak sparsity $O(1)$ such that any $(1+\epsilon)$-spanner of $G$ must have lightness $\Omega(\frac{1}{\epsilon^2})$.  
\end{theorem}

Consequently, there is an inherent difference between the dependence on $\eps$ in the lightness of spanners with stretch at least $2$ and those with stretch $(1+\epsilon)$. 
Again, the exponential dependence on $1/\epsilon$ in the lower bound on $n$ in \Cref{thm:lb-oracle-1eps} is unavoidable, since it is possible to construct a $(1+\epsilon)$-spanner with lightness $O(\log n \cdot \frac{\wsp_{\mathcal{O}_{G,t}}}{\epsilon})$ using standard techniques.

To demonstrate that our framework is unified and applicable, we prove the following theorem, which shows that several graph families admit $\gsso$es, and as a result also light spanners. 

\begin{theorem}\label{thm:graph-oracles}The following $\gsso$es exist.
	\begin{enumerate}[noitemsep]
		\item For any weighted graph $G$ and any $k\geq 2$, $\wsp_{\mathcal{O}_{G,2k-1}} = O(n^{1/k})$.
		\item For the complete weighted graph $G$ corresponding to any Euclidean space (in any dimension) and for any $t\geq 1$, $\wsp_{\mathcal{O}_{G,O(t)}} = O(tn^{\frac{1}{t^2}}\log n)$.
		\item For the complete weighted graph $G$ corresponding to any finite $\ell_p$ normed space for $p \in (1,2]$  and for any $t \ge 1$,  $\wsp_{\mathcal{O}_{G,O(t)}} = O(tn^{\frac{\log t}{t^p}}\log n)$.
	\end{enumerate}
\end{theorem}

\Cref{thm:light-general-spanner} follows directly from \Cref{thm:general-stretch-2} and Item (1) of \Cref{thm:graph-oracles}; \Cref{thm:Euclidean-high} (respectively,~\Cref{thm:Lp-high}) follows directly from \Cref{thm:general-stretch-2} and Item (2) (resp., (3)) of \Cref{thm:graph-oracles} with  $\epsilon = 1/2$; any constant $\epsilon < 1$ works.

To prove \Cref{thm:light-Steiner}, we also use $\gsso$es with stretch $t = 1+\epsilon$, but we do that in a more intricate way. If we work with the complete weighted graph $G$ corresponding to a Euclidean point set $P\in \mathbb{R}^d$ as in \Cref{thm:graph-oracles} and simply construct a light spanner from $\gsso$es for $G$, the resulting spanner will be non-Steiner---hence we cannot hope to obtain the lightness bound of \Cref{thm:light-Steiner} due to a lower bound of $\Omega(\epsilon^{-d})$ by~\cite{LS19}. Our key insight here is to allow the oracle to include Steiner points, i.e., points in $\mathbb{R}^d\setminus P$. Formally, a $\gsso$ with Steiner points, given a subset of points $T\subseteq P$ and a distance parameter $L > 0$,  outputs a Euclidean graph $S(V_S,E_S)$ with $T\subseteq V_S$ such that $d_S(x,y) \le (1+\epsilon) ||x,y||$ for any $x\not=y$ in $T$,\footnote{$||x,y||$ is the Euclidean distance between two points $x,y\in \mathbb{R}^d$.}  where $||x,y||\in [L,2L]$. We denote the oracle by $\mathcal{O}_{P,1+\epsilon}$. We show that  Euclidean spaces admit a $\gsso$ with Steiner points that has sparsity $ \tilde{O}_{\epsilon}(\epsilon^{-(d-1)/2})$. Our construction of the $\gsso$ with Steiner points uses the sparse Steiner $(1+\epsilon)$-spanner from our previous work~\cite{LS19} as a black-box.

\begin{theorem}\label{thm:Euclidean-oracle} Any point set $P$ in $\mathbb{R}^d$ admits a $\gsso$ with Steiner points that has  weak sparsity $\wsp_{\mathcal{O}_{P,t+\epsilon}} = \tilde{O}_{\epsilon}(\epsilon^{-(d-1)/2})$.
\end{theorem}

\Cref{thm:general-stretch-1eps} remains true even when the output of the oracle is not a subgraph of $G$. In this case the resulting spanner may contain vertices not in $G$. For point sets in $\mathbb{R}^d$, the resulting spanner is a \emph{Steiner} spanner, i.e., \Cref{thm:light-Steiner} follows directly from ~\Cref{thm:general-stretch-1eps} and~\Cref{thm:Euclidean-oracle}.

\subsection{Glossary}

\renewcommand{\arraystretch}{1.3}
\begin{longtable}{| l | l|} 
	\hline
	\textbf{Notation} & \textbf{Meaning} \\ \hline
	$t,\eps$ & Stretch parameters, $t\geq 1, \eps \ll 1$.\\ \hline 
	$\nor2{p,q}$ & Euclidean distance between two points $p,q\in \mathbb{R}^d$.\\ \hline 
	$H_{<L}$ & A subgraph of $G$ in the definition of $(L,\eps,\beta)$-cluster graph (\Cref{def:ClusterGraph-Param}).\\ \hline 
	$L,\beta$ & Parameters in $(L,\eps,\beta)$-cluster graph. \\ \hline 
	$\mg = (\mv,\me,\omega)$ &  The $(L,\eps,\beta)$-cluster graph;  $L \leq \omega(\varphi_{C_1},\varphi_{C_2}) < (1+\epsi)L$  $\forall (\varphi_{C_1},\varphi_{C_2})\in \me$. \\ \hline 
	$\varphi_C$ &  The node in $\mg$ corresponding to a cluster $C$.\\ \hline 
	$\sso$ & The \hyperlink{SPHigh}{sparse spanner oracle}.\\ \hline 
	$\chi$ & The  \hyperlink{Sparsity}{sparsity} parameter of $\sso$.\\ \hline 
	$s_{\sso}$ & The  \hyperlink{Stretch}{stretch} function of $\sso$.\\ \hline 
	$F$ &  The set of edges of $G$ returned by $\sso$.\\ \hline 
	$\gsso$ & The \hyperlink{SPHigh}{general sparse spanner oracle}.\\ \hline 
	$\mathcal{O}_{G,t}$ & $\gsso$ of $G$ with stretch $t$.\\ \hline 
	$\wsp_{\mathcal{O}_{G,t}}$ & The weak sparsity of the $\gsso$ $\mathcal{O}_{G,t}$.\\ \hline 
	$\ssp_{\mathcal{O}_{G,t}}$ & The strong sparsity of the $\gsso$ $\mathcal{O}_{G,t}$.\\ \hline 
	\caption{Notation introduced in \Cref{sec:intro}.}
	\label{table:glossray}
\end{longtable}
\renewcommand{\arraystretch}{1}

\section{Preliminaries}\label{sec:prelim}

Let $G$ be an arbitrary weighted graph. We denote by $V(G)$ and $E(G)$ the vertex set and edge set of $G$, respectively. We denote by $w: E(G)\rightarrow \mathbb{R}^+$ the weight function on the edge set.  Sometimes we write $G = (V,E)$ to clearly explicate the vertex set and edge set of $G$, and  $G =(V,E,w)$ to further indicate the weight function $w$ associated with $G$. We denote by $\deg_{G}(v)$ the degree of vertex $v$ in $G$. For a subset of vertices $S$, we denote by $\deg_G(S) = \sum_{v\in S}\deg_G(v)$ the total degree of vertices in $S$.  We use $\mst(G)$ to denote a minimum spanning tree of $G$; when the graph is clear from context, we simply use $\mst$  as a shorthand for $\mst(G)$. W

For a subgraph $H$ of $G$, we use $w(H) \defi \sum_{e\in E(H)}w(e)$ to denote the total edge weight of $H$.   The {\em distance} between two vertices $p,q$ in $G$, denoted by $d_G(p,q)$, is the minimum weight of a path between them in $G$. The diameter of $G$,
denoted by $\dm(G)$, is the maximum pairwise distance in $G$. A \emph{diameter path} of $G$ is a shortest (i.e., of minimum weight) path in $G$ realizing the diameter of $G$, that is, 
it is a shortest path between some pair $u,v$ of vertices in $G$ such that $\dm(G) = d_G(u,v)$.

Sometimes we shall consider graphs with weights on both \emph{edges and vertices}. We define the \emph{augmented weight} of a path to be the total weight of all edges and vertices along the path. The \emph{augmented distance} between two vertices in $G$ is defined as the minimum augmented weight of a path between them in $G$. 
Likewise, the \emph{augmented diameter} of $G$, denoted by $\adm(G)$,
 is the maximum pairwise augmented distance in $G$; 
since we will focus on non-negative weights, the augmented distance and augmented diameter
are no smaller than the (ordinary notions of) distance and diameter. 
An \emph{augmented diameter path} of $G$ is a path of minimum augmented weight realizing the augmented diameter of $G$.

Given a subset of vertices $X\subseteq V(G)$, we denote by $G[X]$ the subgraph of $G$ \emph{induced by $X$}: $G[X]$ has $V(G[X]) = X$ and $E(G[X]) = \{(u,v)\in E(G) ~\vert~   u,v \in X\}$. Let $F\subseteq E(G)$ be a subset of edges of $G$. We denote by $G[F]$ the subgraph of $G$ with $V(G[F]) = V(G)$ and $E(G[F]) = F$.

Let $S$ be a \emph{spanning} subgraph of $G$; weights of edges in $S$ are inherited from $G$. The \emph{stretch} of $S$  is  given by $\max_{x,y \in V(G)} \frac{d_S(x,y)}{d_G(x,y)}$, and it is realized by some edge $e$ of $G$. Throughout we will use the following known observation,
which implies that stretch of $S$ is  equal to $\frac{d_S(u,v)}{w(u,v)}$ for some edge $(u,v) \in E(G)$. 
\begin{observation} \label{stretch:ob}
 $\max_{x,y \in V(G)} \frac{d_S(x,y)}{d_G(x,y)} =  \max_{(x,y) \in E(G)} \frac{d_S(x,y)}{d_G(x,y)}$.
\end{observation}
 We say that $S$ is a \emph{$t$-spanner} of $G$ if the stretch of $S$ is at most $t$. There is a simple greedy algorithm,  called $\pathg$ (or shortly $\greedy$), to find a $t$-spanner of a graph $G$: Examine the edges $e = (x,y)$ in  $G$ in nondecreasing order of weights, and add to the spanner edge $(x,y)$ iff the distance between $x$ and $y$ in the {\em current} spanner is larger than $t\cdot w(x,y)$. This algorithm can be naturally extended to the case where $G$ has weights on both edges and vertices; the distance function considered in this case is the augmented distance.  We say that a subgraph $H$ of $G$ is a $t$-spanner for a \emph{subset of edges $X\subseteq E$} if $\max_{(u,v) \in X} \frac{d_H(u,v)}{d_G(u,v)} \leq t$.

In the context of minor-free graphs, we denote by $G/e$ the graph obtained from $G$ by contracting $e$, where $e$ is an edge in $G$. If $G$ has weights on edges, then every edge in $G/e$ inherits its weight from $G$.

In addition to general and minor-free graphs, this paper studies {\em geometric graphs}.  Let $P$ be a set of $n$ points in $\mathbb{R}^d$. We denote by $\nor2{p,q}$ the Euclidean distance between two points $p,q\in \mathbb{R}^d$.   A  \emph{geometric graph} $G$ for $P$ is a graph where the vertex set corresponds to the point set, i.e., $V(G) = P$, and the edge weights are the Euclidean distances, i.e.,  $w(u,v) = \nor2{u,v}$ for every edge $(u,v)$ in $G$. Note that $G$ need not be a complete graph. If $G$ is a complete graph, i.e., $G = (P, {P \choose 2},\nor2{\cdot})$, then $G$ is equivalent to the {\em Euclidean space} induced by the point set $P$.
For geometric graphs, we use the term \emph{vertex} and \emph{point} interchangeably.

We use $[n]$ and $[0,n]$ to denote the sets $\{1,2,\ldots,n\}$ and $\{0,1,\ldots,n\}$, respectively.

\section{Applications}\label{sec:app}

In this section, we use the framework outlined in \Cref{subsec:unified-framework-intro} to obtain all results started in \Cref{subsec:contribution}. Specifically, in \Cref{subsec:B-Oracle}, we show that the existence of a general sparse spanner oracle implies the existence of light spanners. In \Cref{subsec:Oracle}, we construct $\gsso$es for different class of graphs as claimed in \Cref{thm:graph-oracles}. Finally, in \Cref{subsec:minor-light}, we construct a light spanner for minor-free graphs by directly implementing \hyperlink{SPHigh}{$\sso$}.

\subsection{Light Spanners from General Sparse Spanner Oracles}\label{subsec:B-Oracle}

In this section, we provide an implementation of \hyperlink{SPHigh}{$\sso$} using a $\gsso$. We assume that we are given a $\gsso$ $\mathcal{O}_{G,t}$ with weak sparsity $\wsp_{\mathcal{O}_{G,t}}$. We denote the algorithm by $\sso_{\oracle}$. We assume that every edge in $G$ is a shortest path between its endpoints; otherwise, we can safely remove them from the graph.

\begin{tcolorbox}
	\hypertarget{SPHOracle}{}
	\textbf{$\sso_{\oracle}$:} The input is an $(L,\eps,\beta)$-cluster graph $\mg=(\mv,\me,\omega)$. The output is a set of edges $F$.
	\begin{quote}
		For each node $\varphi_C \in \mv(\mathcal{G})$ corresponding to a cluster $C$, we choose a $v \in C$. Let $S$ be the set of chosen vertices. Let 
		\begin{equation} \label{eq:F-oracle}
			F = E(\mathcal{O}_{G,t}(S,L/2))\cup E(\mathcal{O}_{G,t}(S,L))\cup  E(\mathcal{O}_{G,t}(S,2L))
		\end{equation}
		 be the edge set of the spanner returned by the oracle.   We then return $F$.
	\end{quote}
\end{tcolorbox}

We now show that $\sso_{\oracle}$ has all the properties as described in the abstract  \hyperlink{SPHigh}{$\sso$}.

\begin{lemma}\label{lm:App-Oracle} Let $F$ be the output of \hypertarget{SPHOracle}{$\sso_{\oracle}$}. Then $w(F) = O(\wsp_{\mathcal{O}_{G,t}})L\cdot |\mv|$. Furthermore,  $d_{H_{<2L}}(u,v) \leq t(1+ s_{\sso_{\oracle}}(\beta)\eps)w(u,v)$ for every edge $(u,v)$ corresponding to an edge in $\me$, where  $s_{\sso_{\oracle}}(\beta) = 4\beta$ and $\eps$ is sufficiently smaller than $1$, in particular $\eps \leq 1/(4\beta)$. 
\end{lemma}
\begin{proof} Since we only choose exactly one vertex in $S$ per node in $\mg$, $|S| = |\mv|$. By the definition of the sparsity of an oracle (\Cref{def:sparsity}), $w(F) \leq \wsp_{\mathcal{O}_{G,t}} (L/2)\cdot |S| + \wsp_{\mathcal{O}_{G,t}} L\cdot |S| + \wsp_{\mathcal{O}_{G,t}} 2L\cdot |S| = O(\wsp_{\mathcal{O}_{G,t}})L\cdot |\mv|$; this implies the first claim. 
	
		Let $(u,v)$ be an edge in $G$ corresponding to an edge $(\varphi_{C_u}, \varphi_{C_v}) \in \me$. We have that $L \leq w(u,v)< 2L$ by property 3 in \Cref{def:ClusterGraph-Param}. By the construction of $S$ in \hyperlink{SPHOracle}{$\sso_{\oracle}$}, there are two vertices $u_1 \in C_u$ and $v_1 \in C_v$ that are in $S$. Let $P_{u_1,u}$ ($P_{v_1,v}$) be the shortest path in $H_{<L}[C_u]$ ($H_{<L}[C_v]$) between $u$ and $u_1$ ($v$ and $v_1$). By property 4 in \Cref{def:ClusterGraph-Param}, we have that $\max\{w(P_{u_1,u}), w(P_{v_1,v})\} \leq \beta \eps L$. By the triangle inequality, we have:
	\begin{equation}\label{eq:oracleStretch-uv-up}
		d_G(u_1,v_1)\leq w(u,v) + 2\beta\eps L < (2+2\beta\eps) L \leq 4L,
	\end{equation}
	since $\eps \leq 1/\beta$. Also by the triangle equality, it follows that:
	\begin{equation}\label{eq:oracleStretch-uv-low}
		d_G(u_1,v_1) \geq w(u,v) - 2\beta\eps L\geq (1 - 2\beta\eps L) \geq L/2,
	\end{equation}
	since $\eps \leq \frac{1}{4\beta}$. Thus, $d_G(u_1,v_1) \in [L/2, 2L)$. It follows by the definition of $\gsso$ (\Cref{def:oracle}) that there is a path, say $P_{u_1,v_1}$, of weight at most $t\cdot d_G(u_1,v_1)$  between $u_1$ and $v_1$ in the graph induced by $F$. Let $P_{u,v} = P_{u_1,u}\circ P_{u_1,v_1}\circ P_{v,v_1}$ be the path between $u$ and $v$  obtained by concatenating $P_{u_1,u}, P_{u_1,v_1}, P_{v,v_1}$. By the triangle inequality, it follows that:
	\begin{equation}
		\begin{split}
			w(P_{u,v}) &\leq w(P_{u_1,v_1}) + w(P_{u_1,u}) + w(P_{v_1,v}) \leq t\cdot d_G(u_1,v_1) + 2\eps \beta L \\
			&\stackrel{\mbox{\footnotesize{\cref{eq:oracleStretch-uv-up}}}}{=} t\cdot ( w(u,v) + 2\eps\beta  L) + 2\eps\beta L \\
			&\leq t\cdot (w(u,v) + 4\eps\beta L) \qquad \mbox{(since $t\geq 1$)}\\
			&\leq t\cdot (1 + 4\eps  \beta)w(u,v) \qquad \mbox{(since $w(u,v)\geq L$)},
		\end{split}
	\end{equation}
	as desired. \qed 
\end{proof}

We are now ready to prove \Cref{thm:general-stretch-2} and \Cref{thm:general-stretch-1eps}, which we restate below.

\GeneralStretchT*
\begin{proof}
	By \Cref{lm:framework} and \Cref{lm:App-Oracle}, we can construct in polynomial time a spanner $H$ with stretch $t(1 + (2s_{\sso_{\oracle}}(O(1)) +  O(1))\eps)$ where $s_{\sso_{\oracle}}(\beta) = 8\beta$. Thus, the stretch of $H$ is $t(1 + O(\eps))$; we then can recover stretch $t(1+\eps)$ by scaling. 
	
	The lightness of $H$ is  $\tilde{O}_{\eps}((\chi \eps^{-1}))$ with  $\chi = O(\wsp_{\mathcal{O}_{G,t}})$. That implies a lightness of  $\tilde{O}_{\eps}((\wsp_{\mathcal{O}_{G,t}} \eps^{-1}))$  as claimed. \qed
\end{proof}

\GeneralStretchE*
\begin{proof} The proof follows the same line of the proof of \Cref{thm:general-stretch-2}. The difference is that we  apply \Cref{lm:App-Oracle} and \Cref{lm:framework} with $t =  1+\eps$ to construct $H$. Thus, the stretch of $H$ is $t(1 + O(\eps)) = 1 + O(\eps)$. Since   $\chi  = \wsp_{\mathcal{O}_{G,1+\eps}}$, the lightness is $\tilde{O}_{\epsilon}\left(\frac{\wsp_{\mathcal{O}_{G,t}}}{\epsilon} + \frac{1}{\epsilon^2}\right)$ as claimed.	\qed
\end{proof}

\subsection{General Sparse Spanner Oracles}\label{subsec:Oracle}

In this section, we prove Theorem~\ref{thm:Euclidean-oracle} (Subsection~\ref{subsec:Euclidean})  and Theorem~\ref{thm:graph-oracles} (Section~\ref{subsec:general} and~\ref{subsec:metric}). We say that a pair of terminals  is \emph{critical} if their distance is in $[L, 2L)$.

\subsubsection{Low Dimensional Euclidean Spaces}\label{subsec:Euclidean}

We will use the following result proven in the full version of our previous work~\cite{LS19}:

\begin{theorem}[Theorem 1.3~\cite{LS19}]\label{thm:sparse-Steiner}  Given an $n$-point set $P \in \mathbb{R}^d$, there is a Steiner $(1+\epsilon)$-spanner for $P$ with  $\tilde{O}_{\epsilon}(\epsilon^{-(d-1)/2} |P|)$ edges.
\end{theorem}

Let $T\subseteq P$ be a subset of points given to the oracle and $L$ be the distance parameter. By Theorem~\ref{thm:sparse-Steiner}, we can construct a Steiner $(1+\epsilon)$-spanner $S$ for $T$ with $|E(S)| = \tilde{O}_{\epsilon}(\epsilon^{-(d-1)/2} |T|)$. We observe that:

\begin{observation}\label{obs:remove-heavy-edge} Let $x\not= y$ be two points in $T$ such that $||x,y||\leq 2L$, and $Q$ be a shortest path between $x$ and $y$ in $S$. Then, for any edge $e$ such that $w(e)\geq 4L$, $e\not\in P$ when $\epsilon < 1$.
\end{observation}
\begin{proof}
	Since $S$ is a $(1+\epsilon)$-spanner, $w(P)\leq (1+\epsilon)||x,y|| \leq (1+\epsilon)2 L < 4L$.\qed
\end{proof}
Let $\mathcal{O}_{P,(1+\epsilon)}(T,L)$ be the graph obtained from $S$ by removing every edge $e\in E(S)$ such that $w(e)\geq 4L$. By Observation~\ref{obs:remove-heavy-edge}, $\mathcal{O}_{P,(1+\epsilon)}(T,L)$ is a $(1+\epsilon)$-spanner for $T$. Observe that $$w(\mathcal{O}_{P,(1+\epsilon)}(T,L))~\leq~ 4L |E(\mathcal{O}_{P,(1+\epsilon)}(T,L))| \leq 4L |E(S)| ~=~ \tilde{O}_{\epsilon}(\epsilon^{-(d-1)/2} |T| L).$$ It follows that $\wsp_{\mathcal{O}_{P,1+\epsilon}} = \tilde{O}_{\epsilon}(\epsilon^{-(d-1)/2})$. This completes the proof of Theorem~\ref{thm:Euclidean-oracle}.

\subsubsection{General Graphs}\label{subsec:general}

For a given graph $G(V,E)$ and $T\subseteq V$, we construct another weighted graph $G_T(T, E_T,w_T)$ with vertex set $T$ such that for every two vertices $u,v$ that form a critical pair, we add an edge $(u,v)$ with weight $w_T(u,v) = d_G(u,v)$.

We apply the greedy algorithm~\cite{ADDJS93} to $G_T$ with $t = 2k-1$ and return the output of the greedy spanner, say $S_T$, (after replacing each artificial edge by the shortest path between its endpoints) as the output of the oracle $\mathcal{O}_{G,2k-1}$.  We now bound the weak sparsity of $\mathcal{O}_{G,2k-1}$.

It was shown (Lemma 2 in~\cite{ADDJS93}) that $S_T$ has girth $2k+1$ and hence has at most $|T|^{1+1/k}$ edges.  It follows that $w(S_T) ~\leq~ |T|^{1+1/k}2L ~=~ O(n^{1/k})|T|L$. That implies:

\begin{equation*}
	\wsp_{\mathcal{O}_{G,2k-1}} = \sup_{T\subseteq V, L \in \mathcal{R}^+} \frac{O(n^{1/k})|T|L}{|T|L} = O(n^{1/k}).
\end{equation*}
This implies Item (1) of Theorem~\ref{thm:graph-oracles}.

\subsubsection{Metric Spaces}\label{subsec:metric}

Let $(X,d_X)$ be a metric space and $\mathcal{P}$ be a partition  of $(X,d_X)$ into clusters. We say  that $\mathcal{P}$ is  \emph{$\Delta$-bounded} if $\dm(P) \leq \Delta$ for every $P \in \mathcal{P}$.  For each $x \in X$, we denote the cluster containing $x$ in $\mathcal{P}$ by $\mathcal{P}(x)$. The following notion of $(t,\Delta,\delta$)-decomposition was introduced by Filtser and Neiman~\cite{FN18}.

\begin{definition}[($t,\Delta,\eta$)-decomposition] Given parameters $t \geq 1, \Delta > 0, \eta \in [0,1]$, a distribution $\mathcal{D}$ over partitions of $(X,d_X)$ is a $(t,\Delta,\eta)$-decomposition if:
	\begin{itemize}
		\item[(a)] Every partition $\mathcal{P}$ drawn from $\mathcal{D}$ is $t\cdot\Delta$-bounded.
		\item[(b)] For every $x\not= y \in X$ such that $d_X(x,y) \leq \Delta$, $\pr\limits_{\mathcal{P}\sim \mathcal{D}}[\mathcal{P}(x) = \mathcal{P}(y)] \geq \eta$
	\end{itemize}
\end{definition}

$(X,d)$ is $(t,\eta)$-decomposable if it has a ($t,\Delta,\eta$)-decomposition for any $\Delta > 0$.

\begin{claim}\label{clm:strong-sparse-decomposable} If $(X,d_X)$ is $(t,\eta)$-decomposable, it has a  $\gsso$ $\mathcal{O}_{X,O(t)}$ with sparsity $\wsp_{\mathcal{O}_{X,O(t)}} = O(\frac{t \log |X|}{\eta})$. Furthermore, there is a polynomial time Monte Carlo algorithm constructing  $\mathcal{O}_{X,O(t)}$ with constant success probability.
\end{claim}
\begin{proof}
	Let $T$ be a set of terminals given to the oracle $\mathcal{O}_{X,O(t)}$. Let $\mathcal{D}$ be a $(t, 2L,\eta)$-decomposition of $(X,d_X)$.
	
	Initially the spanner $S$ has $V(S) = T$ and $E(S) = \emptyset$. We sample $\rho = \frac{2\ln |T|}{\eta}$ partitions from $\mathcal{D}$, denoted by $\mathcal{P}_1, \ldots, \mathcal{P}_\rho$. For each $i \in [\rho]$ and each cluster $C \in \mathcal{P}_i$, if $|T\cap C| \geq 2$, we pick a terminal $t\in C$ and add to $S$ edges from $t$ to all other terminals in $C$. We then return $S$ as the output of the oracle.
	
	For each partition $\mathcal{P}_i$, the set of edges added to $S$ forms a forest. That implies we add to $S$ at most $|T|-1$ edges per partition. Thus, $|E(S)| \leq (|T|-1) \rho = O(\frac{|T| \log |T|}{\eta})$. Observe that $w(S) \leq |E(S)| \cdot t 2L = (\frac{2|T| t L \log |T| }{\eta})$ since each edge has weight at most $t\cdot (2L)$. Thus, $\wsp_{\mathcal{O}} = O(\frac{t \log |T|}{\eta})  =  O(\frac{t\log |X|}{\eta})$.
	
	It remains to show that with constant probability, $d_{S}(x,y)  \leq O(t)d_X(x,y)$ for every $x\not= y \in T$  such that $L \leq d_X(x,y) < 2L$. Observe by construction that if $x$ and $y$ fall into the same cluster in any partition, there is a $2$-hop path of length at most $4tL = O(t)d_X(x,y)$. Thus, we only need to bound the probability that $x$ and $y$ are clustered together  in some partition. Observe that the probability that there is no cluster containing both  $x$ and $y$ in $\rho$ partitions is at most:
	\begin{equation*}
		(1-\eta)^\rho  = (1-\eta)^{ \frac{2\ln |T|}{\eta}} \leq \frac{1}{|T|^2}
	\end{equation*}
	Since there are at most $\frac{|T|^2}{2}$ distinct pairs, by union bound, the desired probability is at least $\frac{1}{2}$.\qed
\end{proof}

Filtser and Neiman~\cite{FN18} showed that any $n$-point Euclidean metric is $(t,n^{-O(\frac{1}{t^2})})$-decomposable for any given $t > 1$; this implies Item (2) in Theorem~\ref{thm:graph-oracles}.  If $(X,d_X)$ is an $\ell_p$ metric with $p \in (1,2)$,  Filtser and Neiman~\cite{FN18} showed that it is $(t,n^{-O(\frac{\log t}{t^2})})$-decoposable for any given $t > 1$; this implies Item (3) in Theorem~\ref{thm:graph-oracles}.

\subsection{Light Spanners for Minor-Free Graphs}\label{subsec:minor-light}

In this section, we provide an implementation of \hyperlink{SPHigh}{$\sso$} for minor-free graphs, which we denote by $\sso_{\minor}$. The algorithm simply outputs the edge set $\me$. Note that in this case, we set $t = 1+\eps$.

\begin{tcolorbox}
	\hypertarget{SPHMinor}{}
	\textbf{$\sso_{\minor}$:} The input is an $(L,\eps,\beta)$-cluster graph $\mg=(\mv,\me,\omega)$. The output is a set of edges $F$. 
	\begin{quote}
		 Let $F$ be the subset of edges of $G$ that correspond to edges in $\me$. We then return $F$.
	\end{quote}
\end{tcolorbox}

We now show that $\sso_{\minor}$ has all the properties as described in the abstract  \hyperlink{SPHigh}{$\sso$}. Our proof uses the following result:

\begin{lemma}[Kostochka~\cite{Kostochka82} and Thomason~\cite{Thomason84}]\label{lm:minor-sparsity} Any $K_r$-minor-free graph with $n$ vertices has $O(r\sqrt{\log r}n)$ edges. 
\end{lemma}

We are now ready to prove \Cref{thm:minor-free-opt-lightness}, which we restate below. 
\MinorFree*
\begin{proof}  
	Since we add every edge corresponds to an edge in $\me$ in \hyperlink{SPHMinor}{$\sso_{\minor}$}, $s_{\sso_{\minor}}(\beta) = 0$.
		By \Cref{lm:framework} and \Cref{lm:App-Oracle}, we can construct in polynomial time a spanner $H$ with stretch $t(1 + (2s_{\sso_{\minor}}(O(1)) +  O(1))\eps) = (1 + O(\eps))$; note that $t = (1+\eps)$ in this case. We then can recover stretch $(1+\eps)$ by scaling. 
	
	We observe that $\mg$ is a minor of $G$ and hence is $K_r$-minor-free. Thus, by \Cref{lm:minor-sparsity}, $|\me| = O(r\sqrt{\log r})|\mv|$. It follows that $w(F) = O(r\sqrt{\log r})L\cdot |\mv|$ since every edge in $\mg$ has weight at most $2L$. This gives $\chi = O(r\sqrt{\log r})$. By \Cref{lm:framework} for the case $t = 1+\eps$,
	The lightness of $H$ is  $\tilde{O}_{\eps}((\chi \eps^{-1}) + \eps^{-2}) = \tilde{O}_{\eps,r}(r\eps^{-1} + \eps^{-2})$  as claimed. \qed
	\end{proof}

\section{Lightness Lower Bounds}\label{sec:lowerbounds}

In this section, we provide lower bounds on light $(1+\epsilon)$ spanners to prove Theorem~\ref{thm:minor-free-lowerbound} and Theorem~\ref{thm:lb-oracle-1eps}.  Interestingly, our lower bound construction draws a connection between geometry and graph spanners: we construct a fractal-like geometric graph of weight $\Omega(\frac{\mst}{\epsilon^2})$ such that it has treewidth at most $4$ and any $(1+\epsilon)$-spanner of the graph must take all the edges.

\begin{theorem}\label{thm:treewdith}
	For any $n = \Omega(\epsilon^{\Theta(1/\epsilon)})$ and $\epsilon < 1$, there is an $n$-vertex graph $G$ of treewidth at most $4$ such that any light $(1+\epsilon)$-spanner of $G$ must have lightness $\Omega(\frac{1}{\epsilon^2})$.
\end{theorem}

\noindent Before proving Theorem~\ref{thm:treewdith}, we show its implications in   Theorem~\ref{thm:minor-free-lowerbound} and Theorem~\ref{thm:lb-oracle-1eps}.

\begin{proof}[Proof of Theorem~\ref{thm:lb-oracle-1eps}]
	Le (Theorem 1.3 in~\cite{Le20}), building upon the work of Krauthgamer, Nguy$\tilde{\hat{\mbox{e}}}$n and Zondier~\cite{KNZ14}, showed that graphs with treewidth $\tw$ has a $1$-spanner oracle with weak sparsity $O(\tw^4)$. Since the treewidth of $G$ in Theorem~\ref{thm:treewdith} is $4$, it has a $1$-spanner oracle with weak sparsity $O(1)$; this implies Theorem~\ref{thm:lb-oracle-1eps}. \qed
\end{proof}

\begin{proof}[Proof of Theorem~\ref{thm:minor-free-lowerbound}]
	First, construct a complete graph $H_1$ on $r-1$ vertices whose spanner has lightness $\Omega(\frac{r}{\epsilon})$ as follows: Let $X_1\subseteq V(H_1)$ be a subset of $r/2$ vertices and $X_2 = V(H_1)\setminus X_1$. We assign weight $2\epsilon$ to every edge with both endpoints in $X_1$ or $X_2$, and weight $1$ to every edge between $X_1$ and $X_2$. Clearly $\mst(H_1)  = 1 + (r-2)2\epsilon$. We claim that any $(1+\epsilon)$-spanner $S_1$ of $H_1$ must take every edge between $X_1$ and $X_2$; otherwise, if $e = (u,v)$ is not taken where $u\in X_1,v\in X_2$, then $d_{S_1}(u,v) \geq d_{H_1 - e}(u,v)=1+2\epsilon > (1+\epsilon)d_G(u,v)$. Thus, $w(S_1)\geq |X_1||X_2| = \Omega(r^2)$. This implies $w(S_1) = \Omega(\frac{r}{\epsilon})w(\mst(H_1))$.

	Let $H_2$ be an $(n-r+1)$ vertex graph of treewidth 4 guaranteed by Theorem~\ref{thm:treewdith}; $H_2$ excludes $K_r$ as a minor for any $r\geq 6$. We scale edge weights of $H_1$ appropriately so that. $w(\mst(H_2)) = w(\mst(H_1))$.  Connect $H_1$ and $H_2$ by a single edge of weight $2w(\mst(H_1))$ to form a graph $G$. Then $G$ excludes $K_r$ as minor (for $r\geq 5$) and any $(1+\epsilon)$-spanner must have lightness at least $\Omega(\frac{r}{\epsilon} + \frac{1}{\epsilon^2})$.\qed
\end{proof}

We now focus on proving Theorem~\ref{thm:treewdith}. The core gadget in our construction is depicted in Figure~\ref{fig:core}.  Let $C_r$ be a circle on the plane centered at a point $o$ of radius $r$. We use $\arc{ab}$ to denote an arc of $C_r$ with two endpoints $a$ and $b$. We say \emph{$\arc{ab}$ has angle $\theta$} if $\angle aob = \theta$.We use $|\arc{ab}|$ to denote the (arc) length of $\arc{ab}$, and $||a,b||$ to denote the Euclidean length between $a$ and $b$.

By elementary geometry and Taylor's expansion, one can verify that if $\arc{ab}$ has angle $\theta$, then:

\begin{equation}\label{eq:arc-chord-length}
\begin{split}
|\arc{ab}| &= \theta r \\
||a,b|| &= 2r \sin(\theta/2) = r\theta(1-\theta^2/24 + o(\theta^3)) \\
||a,b|| &= \frac{2\sin(\theta/2)}{\theta} |\arc{ab}| = (1-\theta^2/24 + o(\theta^3))|\arc{ab}|
\end{split}
\end{equation}

\paragraph{Core Gadget.~}  The construction starts with an arc $ab$ of angle $\sqrt{\epsilon}$ of a circle $C_r$.~W.l.o.g., we assume that $\frac{1}{\epsilon}$ is an odd integer.  Let $k = \frac{1}{2}(\frac{1}{\epsilon}+1)$.  Let $\{a \equiv x_1, x_2, \ldots, x_{2k} \equiv b\}$ be the set of points, called \emph{break points}, on the arc $ab$ such that $\angle x_iox_{i+1} = \epsilon^{3/2}$ for any $1\leq i \leq 2k-1$.

Let $H_r$ be a graph with vertex set $V(H_r) = \{x_1,\ldots, x_{2k}\}$. We call $x_1$ and $x_{2k}$ two \emph{terminals} of $H_r$. For each $i \in [2k-1]$, we add an edge $x_ix_{i+1}$ of weight $w(x_ix_{i+1}) = ||x_i,x_{i+1}||$ to $E(H_r)$. We refer to edges between $x_ix_{i+1}$ for $i \in [2k-1]$ as \emph{short edges}.  For each $i \in [k]$, we add an edge $x_ix_{i+k}$ of weight $||x_{i},x_{i+k}||$.  We refer to these edges as \emph{long edges}. Finally, we add edge $||x_1,x_k||$ of $E(H_r)$, that we refer to as the \emph{terminal edge} of $H_r$.   We call $H_r$ a \emph{core gadget} of scale $r$.  See Figure~\ref{fig:core}(a) for a geometric visualization of $H_r$ and Figure~\ref{fig:core}(b) for an alternative view of $H_r$.

\begin{figure}
	\centering
	\vspace{-20pt}
	\includegraphics[scale = 0.8]{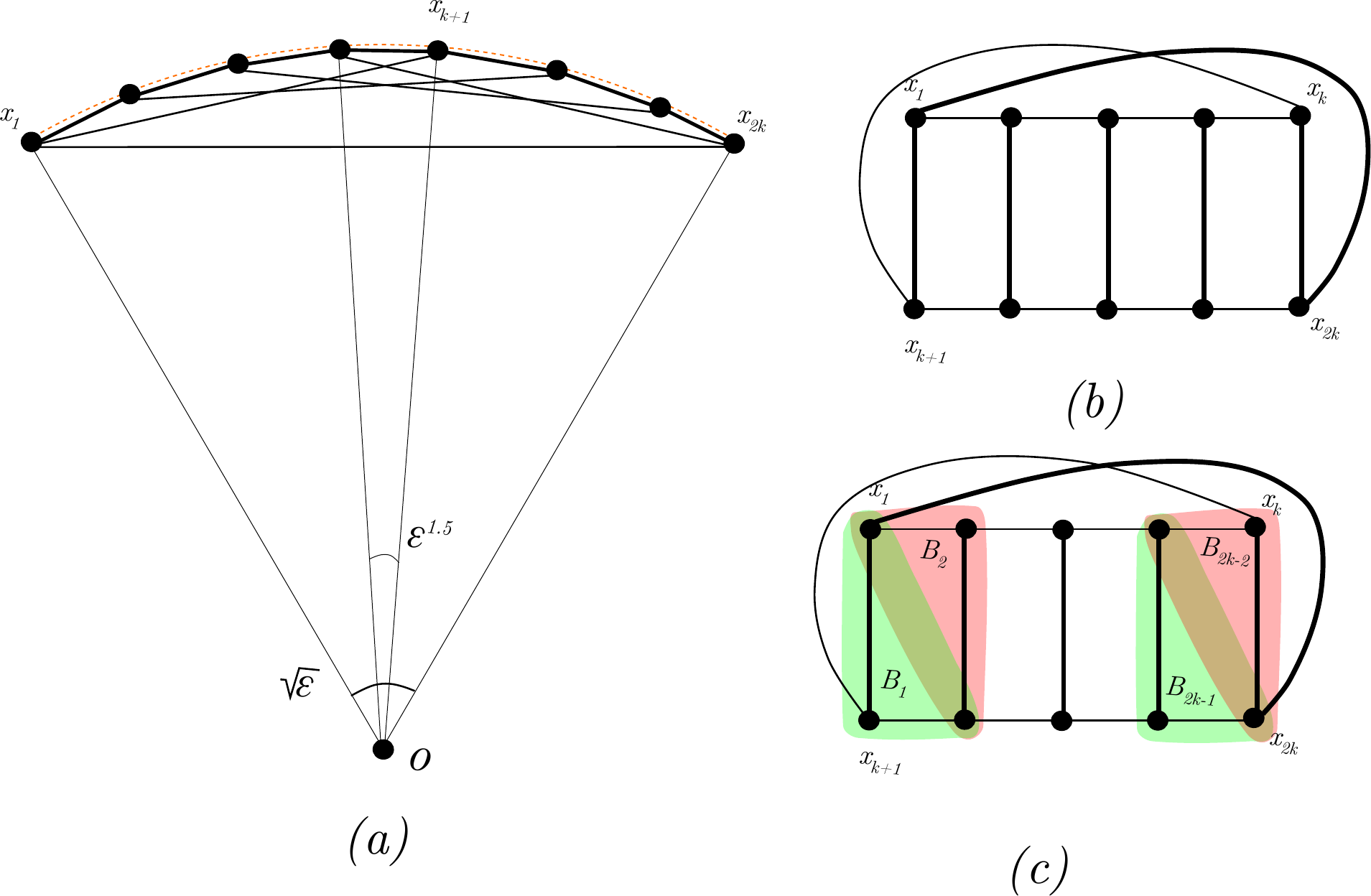}
	\caption{\footnotesize{(a) The core gadget. (b) A different view of the core gadget. (c) A tree decomposition of the core gadget. }}
	\label{fig:core}
\end{figure}

\noindent We observe that:

\begin{observation}\label{obs:edge-length} $H_r$ has the following properties:
	\begin{enumerate}
		\item  For any edge $e \in E(H_r)$, we have:
		\begin{equation}
		w(e) = \begin{cases} 2r\sin(\epsilon^{3/2}/2) &\text{if $e$ is a short edge}\\
		2r\sin(k\epsilon^{3/2}/2) &\text{if $e$ is a long edge}\\ 2r\sin(\sqrt{\epsilon}/2) &\text{if $e$ is the terminal edge}
		\end{cases}
		\end{equation}
		\item $w(\mst(H_r)) \leq r\sqrt{\epsilon}$.
		\item $w(H_r) \geq \frac{r}{6\sqrt{\epsilon}}$ when $\epsilon \ll 1$.
	\end{enumerate}
	\begin{proof}
		We only verify (3); other properties can be seen by direct calculation. By Taylor's expansion, each long edge of $H_r$ has weight $w(e) = 2\sin(\frac{1}{4}(\sqrt{\epsilon} +\epsilon^{3/2})) =  \frac{r}{2}(\sqrt{\epsilon} + o(\epsilon)) \geq r\sqrt{\epsilon}/3$ when $\epsilon \ll 1$. Since $H_r$ has $k$ long edges, $w(H_r) \geq k r \sqrt{\epsilon}/3 \geq  \frac{r}{6\sqrt{\epsilon}}$.\qed
	\end{proof}
\end{observation}

\noindent Next, we claim that $H_r$ has small treewidth.

\begin{claim} \label{clm:treewidth} $H_r$ has treewidth at most $4$.
\end{claim}
\begin{proof}
	We construct a tree decomposition of width $4$ of $H_r$. In fact, we can construct a path decomposition of width $4$ for $H_r$. Let $B_1,\ldots, B_{2k-2}$ be set of vertices where $B_{2i-1}= \{x_{2i-1},x_{2i+k-1}, x_{2i+k}\}$  and $B_{2i}= \{x_{2i-1}, x_{2i+k}, x_{2i}\}$  for each $i \in [k-1]$ (see Figure~\ref{fig:core}(c)). We then add $x_1$ and $x_{k}$ to every $B_i$. Then, $\mathcal{P} = \{B_1,\ldots, B_{2k-2}\}$ is a path decomposition of $H_r$ of width $4$. \qed
\end{proof}

\noindent \textbf{Remark:} It can be seen that $H_r$ has $K_4$ as a minor, thus has treewidth at least $3$. Showing that $H_r$ has treewidth at least $4$ needs more work.

	\begin{wrapfigure}{r}{0.4\textwidth}
	\vspace{-30pt}
	\begin{center}
		\includegraphics[width=0.4\textwidth]{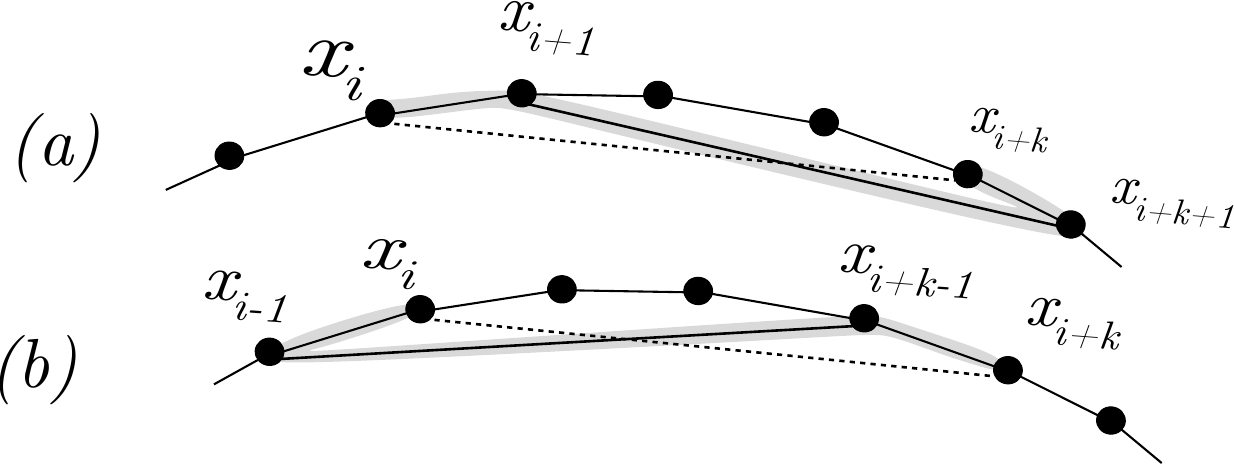}
	\end{center}
	\caption{\footnotesize{Paths $P_e$ between $x_i$ and $x_{i+k}$ are highlighted.}}
	\vspace{-5pt}
	\label{fig:long-edge}
\end{wrapfigure}

\begin{lemma}\label{lm:spanner}
	There is a constant $c$ such that any $(1+\epsilon/c)$-spanner of $H_r$ must have weight at least $$ \frac{w(\mst(H_r))}{6\epsilon}.$$
\end{lemma}
\begin{proof}
	Let $e$ be a long edge of $H_r$ and $G_e = H_r\setminus \{e\}$. We claim that the shortest path between $e$'s endpoints in $G_e$ must have length at least $(1+\epsilon/c)w(e)$ for some constant $c$. That implies any $(1+\epsilon/c)$-spanner of $H_r$ must include all long edges. The lemma then follows from Observation~\ref{obs:edge-length} since $H_r$ has at least $1/2\epsilon$ long edges, and each has length at least $w(\mst(H_r))/3$ for $\epsilon \ll 1$.
	
	Suppose that $e = x_{i}x_{i+k}$.  Let $P_e$ is a shortest path between $x_i$ and $x_{i+k}$ in $G_e$. Suppose that $w(P_e) \leq (1+\epsilon/c)w(e)$.  Since the terminal edge has length at least $3/2 w(e)$, $P_e$ cannot  contain the terminal edge. For the same reason, $P_e$ cannot contain two long edges. It remains to consider two cases:
	
	\begin{enumerate}
		\item $P_e$ contains exactly one long edge. Then, it must be that $P_e = \{x_{i},x_{i+1},x_{i+k+1},x_{i+k}\}$\footnote{indices are mod $2k$.} (Figure~\ref{fig:long-edge}(a)) or $P_e = \{x_{i},x_{i-1},x_{i+k-1},x_{i+k}\}$ (Figure~\ref{fig:long-edge}(b)). In both case, $w(P_i) = w(e) + 4r\sin(\epsilon^{3/2}/2)  \geq w(e)(1  + 2\frac{\sin(\epsilon^{3/2}/2)}{\sin(k\epsilon^{3/2}/2)}) \geq (1+2\epsilon)w(e)$.
		\item $P_e$ contains no long edge. Then, $P_e = \{x_i,x_{i+1}, \ldots,x_{i+k}\}$. Thus we have:
		\begin{equation*}
		\begin{split}
		\frac{w(P_e)}{w(e)} ~=~ \frac{2kr\sin(\epsilon^{3/2}/2)}{2r\sin(k \epsilon^{3/2}/2)} ~=~ 1 + \epsilon/96 + o(\epsilon)  ~\geq~ 1+ \epsilon/100
		\end{split}
		\end{equation*}
	\end{enumerate}
	Thus,  by choosing $c = 100$, we derive a contradiction.\qed
\end{proof}

\begin{figure}[h]
	\centering
	\vspace{0pt}
	\includegraphics[scale = 1.0]{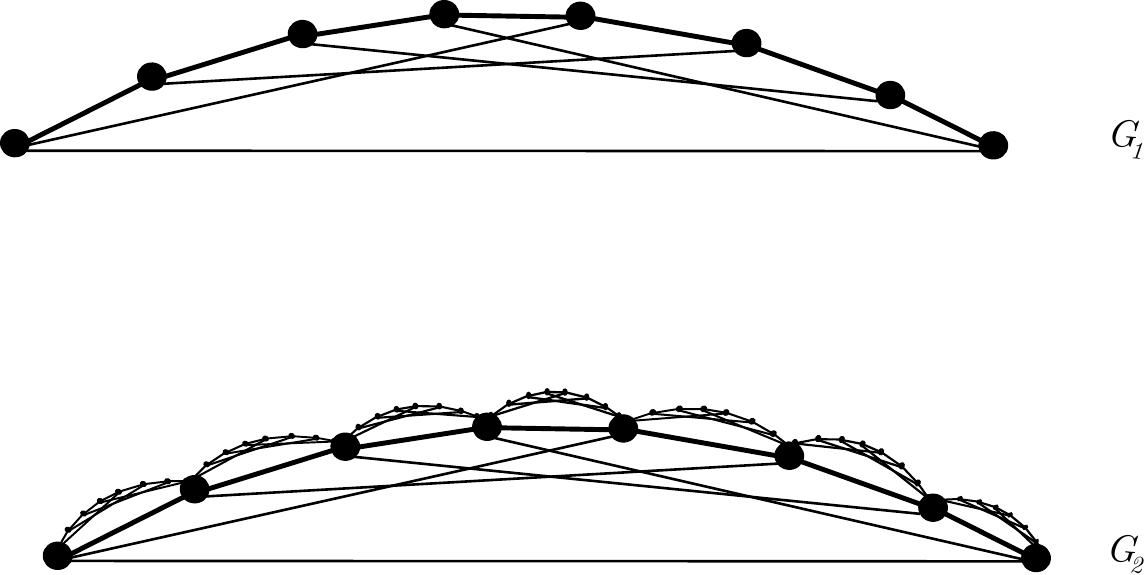}
	\caption{\footnotesize{An illustration of the recursive construction of $G_L$ with two levels.}}
	\label{fig:fractal}
\end{figure}

\paragraph{Proof of Theorem~\ref{thm:treewdith}.~} The construction is recursive. Let $H_1$ the core gadget of scale $1$.  Let $s_1$ ($\ell_1$) be the length of short  edges (long edges) of $H_1$. Let $x^1_1,\ldots, x^1_k$ be break points of $H_1$. Let $\be$ be the ratio of  the length of a short edge to the length of the terminal edge. That is:
\begin{equation}
\be = \frac{||x^1_1,x^1_2||}{||x^1_1,x^1_{2k}||} = \frac{\sin(\epsilon^{3/2}/2)}{\sin(\sqrt{\epsilon}/2)} = \epsilon + o(\epsilon)
\end{equation}
 Let $L = \frac{1}{\epsilon}$. We construct a set of graphs $G_1,\ldots, G_L$ recursively; the output graph is $G_L$. We refer to $G_i$ is the level-$i$ graph.

\noindent \textbf{Level-$1$ graph} $G_1 = H_1$. We refer to breakpoints of $H_1$ as breakpoints of $G$.

\noindent \textbf{Level-$2$ graph} $G_2$ obtained from $G_1$ by: (1) making $2k-1$ copies of the core gadget $H_{\be}$ at scale $\be$ (each $H_\delta$ is obtained by scaling every edge the core gadget by $\delta$), (2) for each $i \in [2k-1]$, attach each copy of $H_{\be}$ to $G_1$ by identifying the terminal edge of $H_{\be}$ and the edge between two consecutive breakpoints $x^1_ix^1_{i+1}$ of $G_1$. We then refer to breakpoints of all $H_{\be}$ as breakpoints of $G_2$. (See Figure~\ref{fig:fractal}.) Note that by definition of $\be$, the length of the terminal edge of $H_{\be}$ is equal to $||x^1_i,x^1_{i+1}||$. We say two adjacent breakpoints of $G_2$ \emph{consecutive} if they belong to the same copy of $H_{\be}$ in $G_2$ and are connected by one short edge of $H_{\be}$.

\noindent \textbf{Level-$j$ graph} $G_j$ obtained from $G_{j-1}$ by: (1) making $(2k-1)^j$ copies of the core gadget $H_{\be^{j-1}}$ at scale $\be^{j-1}$, (2) for every two consecutive breakpoints of $G_{j-1}$, attach each copy of $H_{\be^{j-1}}$ to $G_{j-1}$ by identifying the terminal edge of $H_{\be^{j-1}}$ and the edge between the two consecutive breakpoints. This completes the construction.

\noindent We now show some properties of $G_L$. We first claim that:

\begin{claim}\label{clm:tw-GL} $G_L$ has treewidth at most $4$.
\end{claim}
\begin{proof}
	Let $T_1$ be the tree decomposition of $G_1$ of width $5$, as guaranteed by Claim~\ref{clm:treewidth}. Note that for every pair of consecutive breakpoints $x^1_i,x^1_{i+1}$ of $G_1$, there is a bag, say $X_i$, of $T_1$ contains both $x^1_i$ and $x^1_{i+1}$. Also, there is a bag of $T_1$ containing both terminals of $T_1$.
	
	We extend the tree decomposition $T_1$ to a tree decomposition $T_2$  of $G_2$ as follows. For each gadget $H_{\be}$ attached to $G_1$ via consecutive breakpoints $x_1^i,x^1_{i+1}$, we add a bag $B = \{ x_1^i,x^1_{i+1}\}$, connect $B$ to $X_i$ of $T_1$ and to the bag containing terminals of the tree decomposition of $H_{\be}$. Observe that the resulting tree decomposition $T_2$ has treewidth at most $4$. The same construction can be applied recursively to construct a tree decomposition of $G_L$ of width at most $4$.\qed
\end{proof}

\begin{claim}\label{clm:mst-GL} $w(\mst(G_L)) = O(1) w(\mst(H_1))$.
\end{claim}
\begin{proof}
	Let $r(\epsilon)$ be the ratio between $\mst(H_1)$ and the length of the terminal edge of $H_1$.  Note that $\mst(H_1)$ is a path of short edges between $x_1^1$ and $x^1_{2k}$. By Observation~\ref{obs:edge-length}, we have:
	\begin{equation}
	r(\epsilon)  \leq \frac{r\sqrt{\epsilon}}{2r\sin(\sqrt{\epsilon}/2)}  = 1+\epsilon/24 + o(\epsilon) \leq  1+\epsilon
	\end{equation}
	when $\epsilon \ll 1$. When we attach copies of $H_{\be}$ to edges between two consecutive breakpoints of $G_1$, by re-routing each edge of $\mst(H_1)$ through the path $\mst(H_{\be})$ between $H_{\be}$'s terminals, we obtain a spanning tree of $G_2$ of weight at most $r(\epsilon)w(\mst(H_1)) \leq (1+\epsilon)w(\mst(H_1))$. By induction, we have:
	\begin{equation*}
	w(\mst(G_j)) \leq (1+\epsilon)w(\mst(G_{j-1})) \leq (1+\epsilon)^{j-1} w(\mst(H_1))
	\end{equation*}
This implies that $w(\mst(G_L)) \leq (1+\epsilon)^{L-1}w(\mst(H_1)) = O(1)w(\mst(H_1))$. \qed
\end{proof}

Let $S$ be an $(1+\epsilon/100)$-spanner of $G_L$ ($c = 100$ in Lemma~\ref{lm:spanner}). By Lemma~\ref{lm:spanner}, $S$ includes every long edge of all copies of $H_r$ at every scale $r$ in the construction. Recall that $||x^1_1,x^{1}_{2k}||$ is the terminal edge of $G_1$. Let $L_j$ be the set of long edges of all copies of $H_{\be^{j-1}}$ added at level $j$. Since $\frac{\mst(G_1)}{||x^1_1,x^{1}_{2k}||} = r(\epsilon)$,  we have:
\begin{equation}
\begin{split}
w(\mst(G_1) &= \frac{r(\epsilon)}{r(\epsilon)-1}\left(w(\mst(G_1)) - ||x^1_1,x^1_{2k}||\right)\geq \frac{24}{\epsilon}\left(w(\mst(G_1)) - ||x^1_1,x^1_{2k}||\right)
\end{split}
\end{equation}
By Lemma~\ref{lm:spanner}, we have:
\begin{equation}
\begin{split}
w(L_1) &\geq \frac{1}{6\epsilon}w(\mst(G_1)) \geq \frac{4}{\epsilon^2}(w(\mst(G_1)) - ||x^1_1,x^1_{2k}||)\\
w(L_2) &\geq   \frac{4}{\epsilon^2}(w(\mst(G_2)) - \mst(G_1))\\
&\ldots \\
w(L_j) &\geq \frac{4}{\epsilon^2}(w(\mst(G_{j})) - w(\mst(G_{j-1})))
\end{split}
\end{equation}
Thus, we have:
\begin{equation*}
w(S) \geq \sum_{j=1}^L w(L_j) \geq \frac{1}{4\epsilon^2}(w(\mst(G_L)) - ||x^1_1,x^1_{2k}||) = \Omega(\frac{1}{\epsilon^2})w(\mst(G_L)
\end{equation*}

By setting $\epsilon \leftarrow \epsilon/100$, we complete the proof of Theorem~\ref{thm:treewdith}. The condition on $n$ follows from the fact that  $G_L$ has $|V(G_L)| = O( (2k-1)^{L}) = O((\frac{1}{\epsilon})^{\frac{1}{\epsilon}})$ vertices. \qed

\section{Unified Framework: Proof of \Cref{lm:framework}}\label{sec:framework}

In this section, we describe the unified framework presented in our companion work~\cite{LS21}; we refer readers to~\cite{LS21} for the details of the proof.   Two important components in our spanner construction is a hierarchy of clusters and a potential function. The idea of using a hierarchy of clusters in spanner constructions dated back to the early 90s \cite{ALGP89,CDNS92}, and was used by most if not all of the works on light spanners (see, e.g.,~\cite{ES16,ENS14,CW16,BLW17,BLW19,LS19}).

First, we need a setup step ``normalize'' the set of edges of $G$ for which we construct a light spanner. Let $\MST$ be the minimum spanning tree of the graph $G(V,E)$ with $n$ vertices and $m$ edges. Let $T_{\MST}$ be the running time to construct $\MST$. By scaling, we assume that the minimum edge-weight is $1$. Let $\bar{w} = \frac{w(\mst)}{m}$. Next, we add every edge of length at most $\bar{w}/\eps$ to the spanner; this incurs only an additive factor $+O(\frac{1}{\eps})$ in the lightness (Observation 3.1 in \cite{LS21}). The remaining set of edges is partition into $O(\frac{1}{\psi}\log\frac{1}{\eps})$ sets of edges $E^{\sigma} \subseteq E$ for $\sigma \in \{1,\ldots, \lceil \frac{\log(1/\eps)}{\log(1+\psi)} \rceil\}$; here  $\psi \in (0,1)$ is a parameter $\psi$ chosen by specific applications of the framework.  $E^{\sigma}$ has the following property: for any two edges $e, e' \in E^{\sigma}$, their weights are either \emph{the same} up to a factor of $1+\psi$ or \emph{far apart} by a factor of at least $\frac{1}{\eps}$. Specifically, $E^{\sigma}$ can be written as $E^{\sigma} = \cup_{i\in \mathbb{Z}^+} E^{\sigma}_i$ where:
\begin{equation}\label{eq:Esigmai}
	E^{\sigma}_i = \left\{e : \frac{L_i}{1+\psi} \leq w(e) < L_i \right\} \mbox{ with } L_i = L_{0}/\eps^i, L_0 = (1+\psi)^{\sigma}\bar{w}~. 
\end{equation}
Since $\sigma,i \geq 1$, every edge in $E^{\sigma}$ has weight at least $\frac{\bar{w}}{\eps}$. We refer edges in $E^{\sigma}_i$ as \emph{level-$i$ edges}. We focus on constructing a $t(1+\eps)$-spanner for edges in $E^{\sigma}$ for a fixed $\sigma \in \{1,\ldots, \lceil \frac{\log(1/\eps)}{\log(1+\psi)} \rceil\}$. In the fast constructions in our companion work~\cite{LS21}, we choose $\psi = \eps$. As we will see later (\Cref{lm:framework-technical}), the value of $\psi$ is factored into the lightness of the spanner. Therefore, to minimize the dependency on $\eps$, we choose $\psi = 1/250$ in this paper.

\paragraph{Subdividing $\mst$.~} We subdivide each edge $e \in \mst$ of weight more than $\bar{w}$ into $\lceil \frac{w(e)}{\bar{w}} \rceil$ edges of weight (of at most $\bar{w}$ and at least $\bar{w}/2$ each) that sums to $w(e)$. (New edges do not have to have equal weights.)  Let $\widetilde{\mst}$ be the resulting subdivided $\mst$.
We refer to vertices that are subdividing the $\mst$  edges as \emph{virtual vertices}. Let $\tilde V$ be the set of vertices in $V$ and virtual vertices; we call $\tilde{V}$  the {\em extended set} of vertices. Let $\tilde G = (\tilde V,\tilde E)$ be the graph that consists of the edges in $\widetilde{\mst}$ and $E^{\sigma}$.

Let $\tilde{G}(\tilde{V}, \tilde{E})$ be a graph obtained from $G(V,E)$ by keeping $\MST$ edges, removing every edge not in $E^{\sigma}$ or having weight (strictly) larger than $w(\MST)$, and subdividing each edge $e \in \mst$ into $\lceil \frac{w(e)}{\bar{w}} \rceil$ edges whose weights sums to $w(e)$. It was observed in~\cite{LS21} (Observation 3.4) that $|\tilde{E}| = O(m)$

The spanner we construct for $E^{\sigma}$ is a subgraph of $\tilde{G}$ containing all edges of $\widetilde{\mst}$. This property can be guaranteed by adding $\widetilde{\mst}$ to the spanner. By replacing the edges of $\widetilde{\mst}$ by those of $\mst$, we can transform any subgraph of $\tilde{G}$ that contains the entire tree $\widetilde{\mst}$ to a subgraph of $G$ that contains the entire tree $\mst$. We denote  by $\tilde H^{\sigma}$ the $t(1+\eps)$-spanner of $E^{\sigma}$ in $\tilde{G}$; by abusing the notation, we will write $H^{\sigma}$ rather than $\tilde H^{\sigma}$ in the sequel, under the understanding that in the end we transform $H^{\sigma}$ to a subgraph of $G$.

\paragraph{Cluster hierarchy.~} Our spanner construction is based on a \emph{hierarchy of clusters} $\mathcal{H} = \{\mathcal{C}_1,\mathcal{C}_2, \ldots \}$\footnote{The number of clusters is not defined beforehand; it depends on the construction.} with three properties: 

\begin{itemize}  [noitemsep] 	
	\item \textbf{(P1)~} 	\hypertarget{P1}{} For any $i\geq 1$, each $\mathcal{C}_i$ is a partition of $\tilde{V}$. When $i$ is large enough, $\mathcal{C}_i$ contains a single set $\tilde{V}$ and $\mathcal{C}_{i+1} = \emptyset$.
	\item \textbf{(P2)~} \hypertarget{P2}{} $\mathcal{C}_i$ is an \emph{$\Omega(\frac{1}{\eps})$-refinement} of $\mathcal{C}_{i+1}$, i.e., every cluster $C\in \mathcal{C}_{i+1}$ is obtained as the union of $\Omega(\frac{1}{\epsilon})$ clusters in $\mathcal{C}_i$ for $i\geq 1$.
	\item \textbf{(P3)~} \hypertarget{P3}{} For each cluster $C\in \mathcal{C}_i$, we have $\dm(H^{\sigma}[C]) \leq g L_{i-1}$, for a sufficiently large constant $g$ to be determined later. (Recall that $L_i$ is defined in \Cref{eq:Esigmai}.) 
\end{itemize}

Graph $H^{\sigma}$ will be constructed along with the cluster hierarchy, and at some step $s$ of the algorithm, we construct a level-$i$ cluster $C$. Let $H^{\sigma}_s$ be $H^{\sigma}$ at step $s$. We shall maintain (\hyperlink{P3}{P3}) by maintaining the invariant that $\dm(H^{\sigma}_s[C]) \leq g L_{i-1}$; indeed,  adding more edges in later steps of the algorithm does not increase the diameter of the subgraph induced by $C$. 

To bound the weight of $H^{\sigma}$, we rely on a potential function $\Phi$ that is formally defined as follows:

\begin{definition}[Potential Function $\Phi$]\label{def:Potential}  We use a potential function $\Phi: 2^{\tilde{V}}\rightarrow \mathbb{R}^+$ that maps each cluster $C$ in the hierarchy $\mathcal{H}$ to a potential value $\Phi(C)$, such that the total potential of clusters at level $1$ satisfies:
	\begin{equation}\label{eq:Phi1}
		\sum_{C\in\mathcal{C}_1}\Phi(C) ~\leq~ w(\MST)~.
	\end{equation}
	Level-$i$ potential is defined as $\Phi_i = \sum_{C\in \mathcal{C}_i} \Phi(C)$ for any $i\geq 1$. The \emph{potential change} at level $i$, denoted by $\Delta_i$ for every $i \geq 1$, is defined as:
	\begin{equation}\label{eq:PotentialReduction}
		\Delta_i ~=~ \Phi_{i-1} - \Phi_{i}~. 
	\end{equation}
\end{definition} 
We call $\Phi_i$ the potential at level $i$ and $\Delta_i$ the \emph{potential reduction} at level $i$. By definition, $\Phi_1 = O(1)w(\MST)$. The following lemma proven in~\cite{LS21} is the key in our framework.

\begin{restatable}[Lemma 4.8~\cite{LS21}]{lemma}{FrameworkTechnical}
	\label{lm:framework-technical} Let $\psi \in (0,1], t \geq 1, \eps > 0$ be  parameters such that $\eps \ll 1$   and $E^{\sigma}= \cup_{i\in \mathbb{N}^+} E^{\sigma}_i$ be the set of edges defined in Equation~\eqref{eq:Esigmai}. Let $\{a_i\}_{i \in \mathbb{N}^+}$ be a sequence of positive real numbers such that $\sum_{i \in \mathbb{N}^+} a_i \leq A\cdot w(\mst)$ for some $A\in \mathbb{R}^+$. Let $H_0 = \mst$. For any level $i\geq 1$, if we can compute all subgraphs  $H_1,\ldots,H_i\subseteq G$  as well as the cluster sets $\{\mathcal{C}_{1},\ldots,\mathcal{C}_{i},\mathcal{C}_{i+1}\}$ such that:
	\begin{enumerate}[noitemsep]
		\item[(1)] $w(H_i) \leq  \lambda \Delta_{i+1} + a_i$ for some $\lambda \geq 0$,
		\item[(2)] for every $(u,v)\in E^{\sigma}_i$, $d_{H_{<L_i}}(u,v)\leq t(1+ \rho\cdot \epsilon)w(u,v)$ when $\eps \in (1,\eps_0)$ for some constants $\rho$, $\eps_0$, where $H_{<L_i} = \cup_{j=0}^{i} H_{j}$.
	\end{enumerate}
	Then  we  can construct a $t(1+ \rho \eps)$-spanner  for $G(V,E)$ with lightness  $O(\frac{\lambda + A + 1}{\psi}\log \frac{1}{\epsilon} + \frac{1}{\eps})$ when $\eps \in (0,\eps_0)$.
\end{restatable}

We refer readers to Section 3 in \cite{LS21} for proof of \Cref{lm:framework}. While running time is stated in  \Cref{lm:framework-technical}, time is not the focus of our current paper. Furthermore, in  \cite{LS21}, we need a strong induction assumption that $H_{<L_i}$ is a spanner for edges of $G$ of length less than $L_i$, including those that are not in $E^{\sigma}$, as stated in Item (2) in \Cref{lm:framework-technical}. In this paper, a weaker assumption suffices: $H_{<L_i}$ is a spanner for edges of length less than $L_i$ in $E^{\sigma}$.


\subsection{Designing A Potential Function}\label{subsec:DesignPotential}

\paragraph{Level-$1$ clusters.~} Recall that $\msttilde$ is the minimum spanning tree of $\tilde{G}$ obtained by subdividing edges of $\mst$ of $G$. We abuse notation by also using $\msttilde$ to refer to the edge set of $\msttilde$. Level-$1$ clusters  are subgraphs of $\msttilde$ constructed by the following lemma.

\begin{lemma}[Lemma 3.8~\cite{LS21}]\label{lm:level1Const}We can construct a set of level-$1$ clusters $\mathcal{C}_1$ such that, for each cluster $C\in \mathcal{C}_1$, the subtree $\msttilde[C]$ of $\msttilde$ induced by $C$ satisfies $L_0 \leq \dm(\msttilde[C]) \leq 14L_0$. 
\end{lemma} 

The potential values of level-$1$ clusters are defined as follows:	
	\begin{equation}\label{eq:Level1Poten}
		\Phi(C) = \dm(\msttilde[C]) \qquad \forall C \in \mathcal{C}_1
	\end{equation}

Since level-$1$ clusters induces vertex-disjoint subtrees of $\msttilde$, $\Phi_1 = \sum_{C\in \mathcal{C}_1}\Phi(C) \leq w(\msttilde) = w(\mst)$. Thus, the potential function $\Phi$ satisfies \Cref{eq:Phi1} in \Cref{def:Potential}.

\paragraph{Level-$i$ clusters.~} Observe that the potential value of a level-$1$ cluster is the diameter of the subgraph of $\msttilde$ induced by the cluster. However, the potential value of a level-$i$ cluster  do not need to be its diameter. Instead, we inductively guarantee the following  {\em potential-diameter (PD) invariant}, which implies that the potential value of a level-$i$ cluster is  an {\em overestimate} of the cluster's diameter.
\hypertarget{PD}{}
\begin{quote}
	\textbf{PD Invariant:} For every cluster $C \in \mathcal{C}_{i}$ and any $i\geq 1$, $\dm(H_{< L_{i-1}}[C]) \leq \Phi(C)$. (Recall that $H_{< L_{i-1}} = \cup_{j=0}^{i-1} H_j$.)
\end{quote}

Since $H_0 = \msttilde$, level-$1$ clusters satisfy \hyperlink{PD}{PD Invariant.}

\paragraph{The cluster graph.~} The construction of level-$(i+1)$ relies on a \emph{cluster graph} formally defined below. 

\begin{definition}[Cluster Graph]\label{def:ClusterGraphNew} A cluster graph at level $i \geq 1$, denoted by $\mg_i = (\mv_i, \me'_i, \omega)$ is a \emph{simple graph}, where each node corresponds to a cluster in $\mc_i$ and each inter-cluster edge corresponds to an edge between vertices that belong to the corresponding clusters. We assign  weights to both \emph{nodes and edges} as follows:  for each node $\varphi_C \in \mv_i$ corresponding to a cluster $C \in \mathcal{C}_i$, $\omega(\varphi_C) = \Phi(C)$, and for each edge $\mbe = (\varphi_{C_u},\varphi_{C_v}) \in \me'_i$ corresponding to an edge $(u,v)$ of $\tilde{G}$, $\omega(\mbe) = w(u,v)$.   
\end{definition}

The  notion of cluster graphs in \Cref{def:ClusterGraphNew} is different from the notion of $(L,\eps,\beta)$-cluster graphs defined in \Cref{def:ClusterGraph-Param}. In particular, cluster graphs in \Cref{def:ClusterGraphNew} have weights on both edges and nodes, while $(L,\eps,\beta)$-cluster graphs in \Cref{def:ClusterGraph-Param} have weights on edges only.  The cluster graph has the following properties:

\begin{definition}[Properties of $\mg_i$]\label{def:GiProp} \begin{enumerate}[noitemsep]
		\item[(1)] The edge set $\me'_i$ of $\mg_i$ is the union $\msttilde_{i}\cup \me_i$, where $\msttilde_{i}$ is the set of edges corresponding to edges in $\msttilde$ and $\me_i$ is the set of edges corresponding to edges in $E^{\sigma}_i$.
		\item[(2)] $\msttilde_{i}$ induces a spanning tree of $\mg_i$, which is a minimum spanning tree. We abuse notation by using $\msttilde_{i}$ to denote the induced spanning tree.
	\end{enumerate}
\end{definition}
We note that $\msttilde_{i}$ is a minimum spanning tree of $\mg_i$ since any edge in $\widetilde{\mst}_i$ (of weight at most $\bar{w}$) is of strictly smaller weight than that of any edge in $E^{\sigma}_i$ (of weight at least $\frac{\bar{w}}{(1+\psi)\epsilon}$) for any $i\geq 1$ and $\epsilon \ll 1$.

We note that in~\cite{LS21}, the cluster graph $\mg_i$ is required to have no \emph{removable edges}:  an edge $(\varphi_{C_u},\varphi_{C_v}) \in \me_i$ is removable if (i) the path $\msttilde_i[\varphi_{C_u},\varphi_{C_v}]$ between $\varphi_{C_u}$ and $\varphi_{C_v}$ only contains nodes in $\msttilde_{i}$ of degree at most $2$ and (ii) $\omega(\msttilde_i[\varphi_{C_u},\varphi_{C_v}]) \leq t(1 + 6g\eps)\omega(\varphi_{C_u},\varphi_{C_v})$. Eliminating removable edges from $\mg_i$ can be seen as applying a preprocessing step to $\mg_i$ before the construction of level-$(i+1)$ clusters. In this work, we require a more careful preprocessing step. When $t = 1+\eps$, we preprocess $\mg_i$ in such a way that the output is minimal: removing any edge from the cluster graph will make the stretch of the edge larger. When $t \geq 2$, we need an even more delicate construction where identifying removable edges intertwines with the construction of level-$(i+1)$ clusters.  The details are delayed to \Cref{subsec:LeveIplus1Construction}.

\paragraph{Structure of level-$(i+1)$ clusters.~} Similar to~\cite{LS21}, level-$(i+1)$ clusters correspond to subgraph of $\mg_i$. Specifically, we shall construct collection of subgraphs $\mathbb{X}$ of $\mg_i$, and each subgraph $\mx\in \mathbb{X}$ is mapped to a cluster $C_{\mx} \in \mathcal{C}_{i+1}$ by taking the union of all level-$i$ clusters corresponding nodes in $\mathcal{X}$. More formally:
\begin{equation}\label{eq:XtoCluster}
	C_{\mathcal{X}}  = \cup_{\varphi_C\in \mv(\mx)} C~.
\end{equation}
We use $\mv(\mx)$ and $\me(\mx)$ to denote the vertex set and edge set of a subgraph $\mx$ of $\mg_i$, respectively. The set of subgraphs $\mathbb{X}$ we construct satisfies the the following properties:

\begin{itemize}[noitemsep]
	\item \textbf{(P1').~} \hypertarget{P1'}{}  $\{\mv(\mx)\}_{\mx \in \mathbb{X}}$ is a partition of $\mv_i$.
	\item \textbf{(P2').~} \hypertarget{P2'}{} $|\mv(\mx)| = \Omega(\frac{1}{\eps})$.
	\item \textbf{(P3').~} \hypertarget{P3'}{} $\zeta L_i \leq \adm(\mx) \leq gL_{i}$ for $\zeta = 1/250$.
\end{itemize}

Here $\md(\mx)$ is the augmented diameter of $\mx$ defined in \Cref{sec:prelim}, which is at least the diameter of the corresponding cluster $C_{\mx}$. The potential of  $C_{\mx}$ is then defined as:
\begin{equation}\label{eq:SetPotential-i}
	\Phi(C_{\mx}) = \adm(\mx).
\end{equation}

This implies that $\Phi(C_{\mx})\geq \dm(C_{\mx})$. Recall that in the definition of the cluster graph (\Cref{def:ClusterGraphNew}), each node in the cluster graph for level-$(i+1)$ cluster graph $\mg_{i+1}$ has a weight to be its potential. Thus, the node $\varphi_{C_{\mx}}\in \mg_{i+1}$ has a weight $\omega(\varphi_{C_{\mx}}) = \Phi(C_{\mx}) = \adm(\mx)$ by \Cref{eq:SetPotential-i}.  The following lemma relates properties (P1')-P(3') with properties \hyperlink{P1}{(P3)}-\hyperlink{P1}{(P3)}.

\begin{lemma}[Lemma 4.4~\cite{LS21}]\label{lm:PropEquiv} Let $\mx \in \mathbb{X}$ be a subgraph of $\mg_i$ satisfying properties (\hyperlink{P1'}{P1'})-(\hyperlink{P3'}{P3'}). Suppose that for every edge $(\varphi_{C_u},\varphi_{C_v})\in \me(\mx)$, $(u,v) \in H_{<L_i}$.  By setting the potential value of $C_{\mx}$ to be $\Phi(C_{\mx}) = \adm(\mx)$ for every $\mx \in \mathbb{X}$, the \hyperlink{PD}{PD Invariant} is satisfied, and that  $C_{\mx}$ satisfies all properties (\hyperlink{P1}{P1})-(\hyperlink{P3}{P3}).
\end{lemma}

\paragraph{Local potential change.~}  The notion of local potential change is central in the framework laid out in~\cite{LS21}. For each subgraph  $\mx \in \mathbb{X}$, the local potential change of $\mx$, denoted by $\Delta_{i+1}(\mx)$, is defined as follows:
\begin{equation}\label{eq:LocalPotential}
	\Delta_{i+1}(\mx) \stackrel{\mbox{\tiny{def.}}}{=}   \left(\sum_{\varphi_C\in \mv(\mx)} \Phi(C) \right) -  \Phi(C_{\mx}) = \left(\sum_{\varphi_C\in \mv(\mx)} \omega(\varphi_C) \right) - \adm(\mx). 
\end{equation} 

It was observed in~\cite{LS21} that the (global) potential change at level $(i+1)$ (\Cref{def:Potential}) is equal to the sum of local potential changes over all subgraphs in $\mx$.

\begin{claim}[Claim 4.5~\cite{LS21}]\label{clm:localPotenDecomps}$\Delta_{i+1} = \sum_{\mx \in \mathbb{X}}\Delta_{i+1}(\mx)$.
\end{claim}

The decomposition of the global potential change makes the tasks of bounding the weight of $H_i$ (defined in \Cref{lm:framework}) easier, as we could  bound the number of edges added to $H_i$ incident to nodes in a subgraph $\mx\in \mathbb{X}$ by the local potential change of $\mx$. By summing up over all $\mx$, we obtain a bound on $w(H_i)$ in terms of the (global) potential change $\Delta_{i+1}$.

\subsection{Constructing Level-$(i+1)$ Clusters and $H_i$: Proof of \Cref{lm:framework}}\label{subsec:LeveIplus1Construction}

Our goal is to construct a cluster graph $\mg_i$ and a collection  $\mathbb{X}$ of subgraphs of $\mg_i$ satisfying properties  \hyperlink{P1'}{(P1')}-\hyperlink{P3'}{(P3')}. By \Cref{lm:PropEquiv}, the set of level-$(i+1)$ obtained from subgraphs in $\mathbb{X}$ obtained by applying the transformation in \Cref{eq:XtoCluster} will satisfy properties \hyperlink{P1}{(P1)}-\hyperlink{P3}{(P3)}. To be able to bound the set of edges  in $H_i$ (constructed in \Cref{sec:stretch2} and \Cref{sec:stretch1E}), we need to guarantee that subgraphs in $\mathbb{X}$ have sufficiently large potential changes. This indeed is the crux of our construction. We assume that  $\eps > 0$ is a sufficiently small constant, i.e., $\eps \ll 1, \eps = \Omega(1)$.

\paragraph{Constructing $\mathcal{G}_i$.~}  We shall assume inductively on $i, i \ge 1$ that:
\begin{itemize}[noitemsep]
	\item The set of edges $\widetilde{\mst}_i$ is given by the construction of the previous level $i$ in the hierarchy; for the base case $i = 1$ (see \Cref{subsec:DesignPotential}), $\widetilde{\mst}_1$ is simply a set of edges of $\widetilde{\mst}$ that are not in any level-$1$ cluster. 
	\item The weight $\omega(\varphi_C )$ on each node $\varphi_C \in \mv_i$ is the potential value of cluster $C \in \mathcal{C}_i$; for the base case $i = 1$, the  potential values of level-$1$ clusters were set in \Cref{eq:Level1Poten}.
\end{itemize}

After completing the construction of $\mathbb{X}$, we can compute the weight of each node of $\mg_{i+1}$ by computing the augmented diameter of each subgraph in $\mathcal{X}$; the running time is clearly polynomial. By the \hyperlink{MSTiPlus1}{end of this section}, we show to compute the spanning tree  $\msttilde_{i+1}$ for $\mg_{i+1}$ for the construction of the next level.

\paragraph{Realization of a path.~}  Let $\mp = ( \varphi_0, (\varphi_0,\varphi_1), \varphi_1, (\varphi_1,\varphi_2), \ldots, \varphi_{p})$ be a path of $\mg_i$, written as an alternating sequence of vertices and edges. Let $C_i$ be the cluster corresponding to $\varphi_i$, $0\leq i \leq p$. Let $u$ and $v$ be two vertices such that $u$ is in the cluster corresponding to $\varphi_0$ and $v$ is in the cluster corresponding to $\varphi_p$. 
\begin{figure}[!h]
	\begin{center}
		\includegraphics[width=0.9\textwidth]{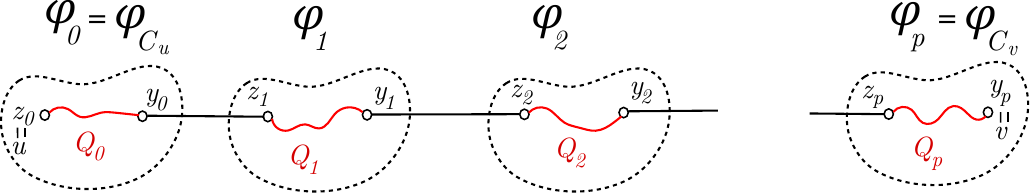}
	\end{center}
	\caption{A realization of a path $\mathcal{P}$ w.r.t $u$ and $v$.}
	\label{fig:path}
\end{figure}
Let $\{y_i\}_{i=0}^p$ and $\{z_i\}_{i=0}^p$ be sequences of vertices of $G$ such that (a) $z_0 = u$ and $y_p = v$ and (b) $(y_{i-1}, z_i)$ is the edge on $G$ corresponds to edge $(\varphi_{i-1},\varphi_i)$ in $\mathcal{P}$ for $1\leq i\leq p$. Let $Q_i$, $0\leq i \leq p$, be a shortest path in $H_{< L_{i-1}}[C_i]$ between $z_i$ and $y_i$ where $C_i$ is the cluster corresponding to $\varphi_i$. See \Cref{fig:path} for an illustration. Let $P = Q_0\circ (y_0,z_1)\circ \ldots\circ Q_p$ be a (possibly non-simple) path from $u$ to $v$. We call $P$ a \emph{realization of $\mathcal{P}$ with respect to $u$ and $v$}. The following observation follows directly from the definition of the weight function of $\mg_i$.

\begin{observation}\label{obs:realization} Let $P$ be a realization of $\mp$ w.r.t two vertices $u$ and $v$. Then $w(P) \leq \omega(\mp)$.
\end{observation}

Next, we show that to construct $H_i$, it suffices to focus on the edges of $E^{\sigma}_i$ that correspond to edges in $\me_i$ of $\mg_i$. 

\begin{lemma}\label{lm:G_i-construction}Let $\psi = 1/250$. We can construct a cluster graph 
	$\mathcal{G}_i = (\mathcal{V}_i,\mathcal{E}_i\cup \widetilde{\mst}_i,\omega)$ in  polynomial time 
	such that $\mathcal{G}_i$ satisfies all properties in \Cref{def:GiProp}. Furthermore, let $F^{\sigma}_i$ be the set of edges in $E^{\sigma}_i$ that correspond to $\mathcal{E}_i$. If every edge in $F^{\sigma}_i$ has a stretch $t(1+s\cdot \eps)$ in $H_{<L_i}$ for some constant $s\geq 1$, then every edge in $E^{\sigma}_i$ has stretch $t(1+ (2s+16g+1)\eps)$ when $\eps < \frac{1}{2(12g+1)}$.
\end{lemma}
\begin{proof} Since $\msttilde_{i}$ is given  at the  outset of the construction of $\mg_i$, we only focus on constructing $\me_i$. For each edge $e = (u,v) \in E^{\sigma}_i$, we add an edge $(\varphi_{C_u}, \varphi_{C_v})$ to $\mg_i$. Next, we remove edges from $\mg_i$. (Step 1) we remove self-loops and parallel edges from $\mg_i$; we only keep the edge of minimum weight in $\mg_i$ among parallel edges. (Step 2) If $t\geq 2$, we remove every edge  $(\varphi_{C_u},\varphi_{C_v})$ from $\mg_i$ such that $\omega(\msttilde_i[\varphi_{C_u},\varphi_{C_v}]) \leq t(1 + 6g\eps)\omega(\varphi_{C_u},\varphi_{C_v})$; the remaining edges of  $\mg_i$ not in $\msttilde_{i+1}$ are $\me_i$. If $t = 1+\eps$, we  apply the  $\pathg$ algorithm  to $\mg_i$ with stretch $t(1 + 6g\eps)$ to obtain $\ms_i$. (Note that we use augmented distances rather than normal distances when apply the greedy algorithm.) It was shown~\cite{ADDJS93} that the $\pathg$ algorithm contains the minimum spanning tree of the input graph. Thus, $\ms_i$ contains $\msttilde_{i}$ as a subgraph. We then set $\me_i = \me(\ms_i)\setminus \msttilde_{i}$; this completes the construction of $\mg_i$.
	
	We now show  the second claim: the stretch of $E^{\sigma}_i$ in $H_{<L_i}$ is  $t(1+ \max\{s+4g,10g\}\eps)$. Let $(u',v')$ be any edge in $E^{\sigma}_i\setminus F^{\sigma}_i$.  Recall that  $(u',v')$ is not in $ F^{\sigma}_i$ because (a) both $u'$ and $v'$ are in the same level-$i$ cluster in the construction of the cluster graph in \Cref{lm:G_i-construction}, or (b) $(u',v')$ is parallel with another edge $(u,v)$, or (c) the edge $(\varphi_{C_{u'}},\varphi_{C_{v'}})$ corresponding to $(u',v')$ is removed from $\mg_i$ in Step 2. 
	
	We argue that case (a) does not happen.  Observe that the level-$i$ cluster containing both $u'$  and $v'$ has diameter at most $gL_{i-1}$ by property (\hyperlink{P3}{P3}), and thus we have a path from $u'$ to $v'$ in $H_{<L_i}$ of  weight at most $gL_{i-1} ~=~ g\eps L_i ~\leq \frac{L_i}{1+\psi}~\leq~w(u',v')$ when $\eps < \frac{1}{(1+\psi)g}$, contradicting that every edge is a shortest path between its endpoints. 
	
	For case (c), we will show that:
	\begin{equation}\label{eq:stretch-case-c}
		d_{H_{i}}(u',v') \leq t(1+(2s+12g+1)\eps)w(u',v')~,
	\end{equation}
 by considering two cases:

	 \noindent \textbf{Subcase 1: $t \geq 2$.~} By construction, 	$d_{H_{< L_{i-1}}}(u',v') \leq t(1+ 6g\eps)w(u',v')$; \Cref{eq:stretch-case-c} holds.
	 
	 \noindent \textbf{Subcase 2: $t = 1+\eps$.~} Le $\mp'$ be the shortest path between $\varphi_{C_{u'}}$ and $\varphi_{C_{v'}}$ in $\ms_i$. Since $\ms_i$ is a $t(1+6g\eps)$-spanner of $\mg_i$, we have: 
	\begin{equation} \label{eq:P-vs-uv}
		\begin{split}
					\omega(\mp') &\leq t(1+6g\eps)\omega(\varphi_{C_{u'}}, \varphi_{C_{v'}}) = (1+\eps)(1+6g\eps)\omega(\varphi_{C_{u'}}, \varphi_{C_{v'}})  \\ 
					& \leq (1 + (12g+1)\eps)\omega(\varphi_{C_{u'}}, \varphi_{C_{v'}})  \qquad
					 \mbox{ (since $\eps \leq 1$)}\\
					 & = (1 + (12g+1)\eps)w(u',v')
		\end{split}
	\end{equation}
	
	\begin{claim}\label{clm:P-one-edge} $\mp'$ contains  at most one edge in $\me_i$. 
	\end{claim}
	\begin{proof} 
		Suppose that  $\mp'$ contains at least two edges in $\me_i$. Since edges in $\me_i$ have weights at least $L_i/(1+\psi)$ and at most $L_i$, $\omega(\mp')\geq  \frac{2L_i}{1+\psi} > (1 + (12g+1)\eps) L_i$ when $\eps < \frac{1}{2(12g+1)}$ and $\psi = \frac{1}{250}$. Since $w(u',v') \leq L_i$, $\omega(\mp') >  (1 + (12g+1)\eps) w(u',v')$, contradicting \Cref{eq:P-vs-uv}.  \qed
	\end{proof}
	
	Let $P'$ be a realization of $\mp'$ w.r.t $u'$ and $v'$.  If $\mp'$ contains no edge in $\me_i$, then $P'$ is a path in $H_{< L_{i-1}}$. This implies that $d_{H\leq i}(u',v') \leq (1 + (12g+1)\eps)w(u',v') \leq t(1 + (12g+1)\eps)w(u',v')$ since $t\geq 1$; \Cref{eq:stretch-case-c} holds. Otherwise,  by \Cref{clm:P-one-edge}, $P'$ contains exactly one edge $(x,y) \in F^{\sigma}_i$. By the assumption of the lemma, $d_{H_{<L_i}}(x,y) \leq t(1+s\cdot \eps)w(x,y)$. Let $Q'$ be obtained from $P'$ by replacing edge $(x,y)$ by a shortest path from $x$ to $y$ in $H_{<L_i}$. Then we have:
	\begin{equation*}
		\begin{split}
				w(Q') &\leq t(1+s\cdot \eps) w(P') \\ &\leq t(1+s\cdot \eps) (1 + (12g+1)\eps)w(u',v') \qquad \mbox{(by \Cref{eq:P-vs-uv})}\\
				 &= t(1+(2s+12g+1)\cdot \eps) \qquad \mbox{(since $(12g+1)\eps \leq 1$)}
		\end{split}
	\end{equation*}
	Thus, in all cases, \Cref{eq:stretch-case-c} holds.

	We now consider case (b); that is, $(u',v')$  is not in $F^{\sigma}_i $ because  it is parallel with another edge $(u,v)$. 	Let $C_u$ and $C_v$ be two level-$i$ clusters containing $u$ and $v$, respectively. W.l.o.g, we assume that $u' \in C_u$  and $v' \in C_v$. Since we only keep the edge of minimum weight among all parallel edges, $w(u,v) \leq w(u',v')$. Since the level-$i$ clusters that contain $u$  and $v$ have diameters at most $gL_{i-1} = g\eps L_i$ by property (\hyperlink{P3}{P3}), it follows that $\dm(H_{<L_i}[C_u]),\dm(H_{<L_i}[C_v]) \le g\epsi L_i$. 	We have:
	\begin{equation*}
		\begin{split}
			d_{H_{<L_i}}(u',v') &\leq 	d_{H_{<L_i}}(u,v) + \dm(H_{<L_i}[C_u]) + \dm(H_{<L_i}[C_v])\\ &\leq t(1+(2s+12g+1)\eps)w(u,v) +   2g\epsi L_i \qquad \mbox{(by \Cref{eq:stretch-case-c})}\\
			&\leq  t(1+(2s+12g+1)\eps) w(u',v')  +  2g\epsi L_i\\
			&\leq t(1+(2s+12g+1)\eps) w(u',v') + 4g\eps w(u',v') \qquad \mbox{(since $w(u',v') \geq L_i/(1+\psi) \geq L_i/2$)}\\
			&= t(1+(2s+16g+1)\eps) w(u',v') \qquad \mbox{(since $t\geq 1$)}.
		\end{split}
	\end{equation*}
	The lemma now follows. 
	\qed
\end{proof}

To construct the set of subgraphs $\mathbb{X}$ of $\mg_i$, we distinguish between two cases: (a) $t = 1+\eps$ and (b) $t\geq 2$. Subgraphs in $\mathbb{X}$ constructed for the case $t=1+\eps$ have properties similar to those of subgraphs constructed in~\cite{LS21}; the key difference is that subgraphs constructed in our work have a larger \emph{average potential change}, which ultimately leads to an optimal dependency on $\eps$ of the lightness. When the stretch $t\geq 2$, we show that one can construct a set of subgraphs $\mathbb{X}$ of $\mg_i$ with much larger potential change, which reduces the dependency of the lightness on $\eps$  by a factor $1/\eps$ compared to the case $t = 1+\eps$. Our construction uses \hyperlink{SPHigh}{$\sso$} as a black box. The following lemma summarizes our construction.

\begin{restatable}{lemma}{HiConstruction}
	\label{lm:ConstructClusterHi} Given \hyperlink{SPHigh}{$\sso$}, we can construct in polynomial time a set of subgraphs $\mathbb{X}$ such that every subgraph $\mx \in \mathbb{X}$ satisfies the three properties (\hyperlink{P1'}{P1'})-(\hyperlink{P3'}{P3'}) with constant $g=223$, and graph $H_i$ such that:
		\begin{equation*}
			d_{H_{<L_i}}(u,v) \leq t(1+ \max\{s_{\sso}(2g),6g\}\eps)w(u,v) \quad \forall (u,v)\in F^{\sigma}_{i}
		\end{equation*}
	 where $ F^{\sigma}_{i}$ is the set of edges defined in \Cref{lm:G_i-construction}. Furthermore,  $w(H_i) \leq  \lambda \Delta_{i+1} + a_i$ such that
	\begin{enumerate}[noitemsep]
		\item \textbf{when $t \geq 2$:}  $\lambda = O(\chi \eps^{-1} )$, and $A = O(\chi \eps^{-1} )$.
		\item \textbf{when $t = 1+\eps$:} $\lambda = O(\chi \eps^{-1} + \epsilon^{-2})$, and $A = O(\chi \eps^{-1} + \epsilon^{-2})$. 
	\end{enumerate}
Here  $A \in \mathbb{R}^+$ such that $\sum_{i\in \mathbb{N}^+}a_i \leq A \cdot w(\mst)$.
\end{restatable}

The proof of \Cref{lm:ConstructClusterHi} is deferred to \Cref{sec:stretch2} for the case $t\geq 2$ and  \Cref{sec:stretch1E} for the case $t = 1+\eps$.   

\paragraph{Constructing $\msttilde_{i+1}$.~}\hypertarget{MSTiPlus1}{}  Let $\msttilde^{out}_{i} = \msttilde_{i}\setminus (\cup_{\mx \in \mathbb{X}}(\me(\mx)\cap \msttilde_{i}))$ be the set of $\msttilde_{i}$ edges that are not contained in any subgraph $\mx \in \mathbb{X}$. Let $\msttilde_{i+1}'$ be the graph with vertex set $\mv_{i+1}$ and there is an edge between two nodes $(\mx,\my)$ in $\mv_{i+1}$ of there is at least one edge in $\msttilde^{out}_{i}$ between two nodes in the two corresponding subgraphs $\mx$ and $\my$. Note that $\msttilde_{i+1}'$ could have parallel edges (but no self-loop).  Since $\msttilde_{i}$ is  a spanning tree of $\mg_i$, $\msttilde_{i+1}'$ must be connected. $\msttilde_{i+1}$ is then a spanning tree of $\msttilde_{i+1}'$.

We are now ready to prove \Cref{lm:framework}, which we restate below.

\Framework*
\begin{proof} We apply \Cref{lm:framework-technical} to construct a light spanner $H$ for $G$ where each graph $H_i$, $i \in \mathbb{N}^+$, is constructed using \Cref{lm:ConstructClusterHi}.   
	
	When $t\geq 2$, by Item (1) of \Cref{lm:ConstructClusterHi} and \Cref{lm:framework-technical}, the lightness of $H$ is $O((\frac{O(\chi \eps^{-1}) +  O(\chi \eps^{-1}) + 1}{1+\psi})\log(\frac{1}{\eps}) + \frac{1}{\eps}) = O_{\eps}(\chi \eps^{-1})$. When $t = 1+\eps$, by Item (2) of \Cref{lm:ConstructClusterHi} and \Cref{lm:framework-technical}, the lightness of $H$ is $O((\frac{O(\chi \eps^{-1}) +  O(\chi \eps^{-1}) + 1}{1+\psi})\log(\frac{1}{\eps}) + \frac{1}{\eps^2}) = O_{\eps}(\chi \eps^{-1} + \eps^{-2})$. 
	
	We now bound the stretch of $H$. By \Cref{lm:ConstructClusterHi} and \Cref{lm:G_i-construction}, the stretch of edges in $E^{\sigma}_i$ in the graph $H_{<L_i}$  is $t(1 + (2s_{\sso}(2g) +  16g+1)\eps)$ with $g  = 223$. Thus, by \Cref{lm:framework-technical}, the stretch of $H$ is $t(1 +(2s_{\sso}(2g) +  16g+1)\eps) = t(1 + (2s_{\sso}(O(1)) +  O(1))\eps)$ as claimed.
	\qed
\end{proof}

\subsection{Summary of Notation}
\renewcommand{\arraystretch}{1.3}
\begin{longtable}{| l | l|} 
	\hline
	\textbf{Notation} & \textbf{Meaning} \\ \hline
	$E^{light}$ &$ \{e \in E(G) : w(e)\le w/\varepsilon\}$\\ \hline 
	$E^{heavy}$ & $E \setminus E^{light}$ \\\hline
	$E^{\sigma} $ & $\bigcup_{i \in \mathbb{N}^{+}}E_{i}^{\sigma}$\\\hline
	$E_{i}^{\sigma} $ & $\{e \in E(G) : \frac{L_i}{1+\psi} \leq w(e) < L_i\}$\\\hline
	$g$ & constant in \hyperlink{P3}{property (P3)}, $g = 223$ \\\hline
	$\mathcal{G}_i = (V_i, \msttilde_{i} \cup \mathcal{E}_i, \omega)$ & cluster graph; see \Cref{def:ClusterGraphNew}. \\\hline
	$\me_i$ & corresponds to a subset of edges of $E^{\sigma}_i$\\\hline
	$\mathbb{X}$ & a collection of subgraphs of $\mathcal{G}_i$\\\hline
	$\mx, \mv(\mx), \me(\mx)$ & a subgraph in $\mathbb{X}$, its vertex set, and its edge set\\\hline
	$\Phi_i$ & $\sum_{c \in C_i}\Phi(c)$ \\\hline
	$\Delta_{i+1} $&$ \Phi_i - \Phi_{i+1}$\\\hline
	$\Delta_{i+1}(\mx)$ & $(\sum_{\phi_C\in \mx }\Phi(C) ) - \Phi(C_{\mx})$\\\hline
	$C_\mx$ & $\bigcup_{\phi_C \in \mx}C$ \\\hline
	$s_{\sso}$ & the stretch constant of \hyperlink{SPHigh}{$\sso$}\\\hline
	\caption{Notation introduced in \Cref{sec:framework}.}
	\label{table:notation}
\end{longtable}
\renewcommand{\arraystretch}{1}

\section{Clustering for Stretch $t\geq 2$}\label{sec:stretch2}

In this section, we prove \Cref{lm:ConstructClusterHi} when the stretch $t$ is at least 2. The general idea is to construct a set $\mathbb{X}$ of subgraphs  of $\mg_i$ such that each subgraph in $\mathbb{X}$ has a sufficiently large local potential change, and carefully choose a subset of edges of $\mg_i$, with the help from \hyperlink{SPHigh}{$\sso$}, such that the total weight could be bounded by the potential change of subgraphs in $\mathbb{X}$ and distances between endpoints of edges in $\me_i$ are preserved. (By \Cref{lm:G_i-construction}, it is sufficient to preserve distances between the endpoints of edges in $\me_i$.)  In \Cref{lm:Clustering} below, we state desirable properties of subgraphs in $\mathbb{X}$. Recall that $H_{< L_{i-1}} = \cup_{j=0}^{i-1} H_{j}$.

\begin{restatable}{lemma}{Clustering}
	\label{lm:Clustering} Let $\mg_i = (\mv_i,\me_i)$ be the cluster graph. We can construct in polynomial time  (i) a collection $\mathbb{X}$ of subgraphs of $\mg_i$ and its partition into two sets $\{\mathbb{X}^{+}, \mathbb{X}^{-}\}$ and (ii) a partition of $\me_i$ into three sets $\{\me_i^{\take}, \me_i^{\reduce}, \me_i^{\redunt}\}$ such that:
	\begin{enumerate}
		\item[(1)] For every subgraph $\mx \in \mathbb{X}$,  $\deg_{\mg^{\take}_i}(\mv(\mx)) = O(|\mv(\mx)|)$ where $\mg^{\take}_i = (\mv_i,\me_i^{\take})$, and $\me(\mx)\cap \me_i \subseteq \me^{\take}$. Furthermore, if $\mx \in \mathbb{X}^{-}$, there is no edge in $\me_i^{\reduce}$ incident to a node in $\mx$.
		
		\item[(2)] Let $H_{< L_i}^{-}$ be a subgraph obtained by adding corresponding edges of $\me_i^{\take}$ to $H_{< L_{i-1}}$.  Then for every edge $(u,v)$ that corresponds to an edge in $\me^{\redunt}$, $d_{H_{< L_i}^{-}}(u,v)\leq 2d_G(u,v)$. 
		
		\item[(3)] Let $\Delta_{i+1}^+(\mx) = \Delta(\mx) + \sum_{\mbe \in \msttilde_i\cap \me(\mx)}w(\mbe)$ be the \emph{corrected potential change} of $\mx$. Then, $\Delta_{i+1}^+(\mx) \geq 0$ for every $\mx \in \mathbb{X}$ and 
		\begin{equation}\label{eq:averagePotential-t2}
			\sum_{\mx \in \mathbb{X}^{+}} \Delta_{i+1}^+(\mx) = \sum_{\mx \in \mathbb{X}^{+}} \Omega(|\mv(\mx)|\eps L_i). 
		\end{equation}
		\item[(4)] For every edge $(\varphi_1,\varphi_2)\in \me_i$ such that $\varphi_1 \in \mx, \varphi_2 \in \my$ for some subgraphs $\mx,\my \in \mathbb{X}^{-}$, then $(\varphi_1,\varphi_2)\in \me^{\redunt}_i$, unless a \emph{degenerate case} happens, in which  $\me^{\reduce}_i = \emptyset$ and  $\me_i^{\take} = O(\frac{1}{\eps})$.
		\item[(5)] For every subgraph $\mx \in \mathbb{X}$, $\mx$ satisfies the three properties (\hyperlink{P1'}{P1'})-(\hyperlink{P3'}{P3'}) with constant $g=223$. Furthermore, if $\mx \in \mathbb{X}^{-}$, then $|\me(\mx)\cap \me_i| = 0$.
	\end{enumerate}	
\end{restatable}

Since  $\Delta_{i+1}(\mx)$ could be negative, we view $\sum_{\mbe \in \msttilde_i\cap \me(\mx)}w(\mbe)$ in the definition of $\Delta_{i+1}^+(\mx)$ as a corrective term to  $\Delta_{i+1}(\mx)$ to make it non-negative.

We now describe the intuition behind all properties stated in \Cref{lm:Clustering}. The set of edges $\me_i$ of $\mg_i$ is partitioned into three sets $\{\me_i^{\take}, \me_i^{\reduce}, \me_i^{\redunt}\}$ where  (a) edges in $\me_i^{\redunt}$ would not be considered in the construction of $H_i$ as their endpoints already have good stretch by Item (2), these are called \emph{redundant edges}; (b) edges in   $\me_i^{\take}$ is the set of edges that we must take to $H_i$, as to guarantee that edges in $\me_i^{\redunt}$ has a good stretch (Item (2) in \Cref{lm:Clustering}); and (c) edges $\me_i^{\reduce}$ are remaining edges that the clustering algorithm has not decided whether to take them to $H_i$. In the construction of $H_i$ in \Cref{subsec:ConstructHiT2}, we rely on  \hyperlink{SPHigh}{$\sso$} to construct a good spanner for edges in $\me_i^{\reduce}$. Item (1) of \Cref{lm:Clustering} guarantees that there are only a few edges in $\me^{\take}_i$ per subgraph in $\mathbb{X}$.

The set of subgraphs $\mathbb{X}$ is partitioned into two sets  $\{\mathbb{X}^{+},\mathbb{X}^{-}\}$.  Item (3) of \Cref{lm:Clustering} implies that each node in $\mx \in \mathbb{X}^{+}$ has $\Omega(\eps L_i)$ average potential change. However, \Cref{lm:Clustering} does not provide any guarantee on the corrected potential changes of subgraphs in $\mathbb{X}^{-}$, other than non-negativity. As a result, we could not bound the weight of edges in $\me^{\take}$ incident to nodes in a subgraph $\mx\in \mathbb{X}^{-}$ by the local potential change of $\mx$. Nevertheless, Item (4) of \Cref{lm:Clustering} implies that, any edge, say $\mbe$, incident to a node in  $\mx$ is also incident to a node in a subgraph $\my \in \mathbb{X}^{+}$ (unless a degenerate case happens). It follows that the weight of $\mbe$ could be bounded by the corrected potential change of $\my$, and $\mx$ do not need to bound the weight of $\mbe$. If the degenerate case happens, there is no edge in $\me_i^{\reduce}$, and there are only few  edges in $\me_i^{\take}$, which we could bound directly by the extra term $a_i$ in \Cref{lm:ConstructClusterHi}.

\Cref{lm:Clustering} is analogous to Lemma 4.20 in \cite{LS21}. Here we point out two major differences, which ultimately lead to the optimal dependency on $\eps$ of the lightness. In~\cite{LS21}, roughly $O(1/\eps)$ edges are added to $H_i$ per node of $\mv_i$. Furthermore, in the construction in \cite{LS21}, each node has $\Omega(L_i \eps^2)$ average potential change. These two facts together incur a factor of $\Omega(1/\eps^3)$ in the lightness. Another factor of $1/\eps$ is due to $\psi = \eps$ for the purpose of obtaining a fast construction. The overall lightness has a factor of $1/\eps^4$ dependency on $\eps$. Our goal is to reduce this dependency all the way down to $1/\eps$. By choosing $\psi = 1/250$, we already eliminate one factor of $1/\eps$. By carefully partitioning $\me_i$ into three set of edges  $\{\me_i^{\take}, \me_i^{\reduce}, \me_i^{\redunt}\}$, and only taking  edges of $\me_i^{\take}$ to $H_i$, we essentially reduce the number of edges we take per node in every subgraph $\mx$ from $O(1/\eps)$ to $O(1)$ (by Item (1) in \Cref{lm:Clustering}), thereby saving another factor of $1/\eps$. Finally, we show that (by Item (3) in \cref{lm:Clustering}), each node in $\mathbb{X}^+$ has $\Omega(L_i \eps)$ average potential change, which is larger than the average potential change of nodes in the construction of \cite{LS21} by a factor of $1/\eps$. We crucially use the fact that $t\geq 2$ in bounding   the average potential change here. All of these ideas together reduce the dependency on $\eps$ from $1/\eps^{4}$ to $1/\eps$ as desired.

Next we show to construct $H_i$ given that we can construct a set of subgraphs $\mathbb{X}$ as claimed in \Cref{lm:Clustering}. The proof of \Cref{lm:Clustering} is deferred to \Cref{subsec:clusteringT2}.

\subsection{Constructing $H_i$: Proof of \Cref{lm:ConstructClusterHi} for $t\geq 2$.} \label{subsec:ConstructHiT2}

In this section, we construct graph $H_i$ as described in \Cref{lm:ConstructClusterHi} in two steps. In Step 1, we take every edge  in $\me^{\take}_i$ to $H_i$. In Step 2, we use \hyperlink{SPHigh}{$\sso$} to construct a subset of edges $F$ to provide a good stretch for edges in $\me^{\reduce}_i$. Note that edges in $F$ may not correspond to edges in $\me^{\reduce}_i$.  As the implementation of  \hyperlink{SPHigh}{$\sso$} depends on the input graph, this is the only place in our framework where the structure of the input graph plays an important role in the construction of the light spanner.

\begin{tcolorbox}
	\hypertarget{HiConstT2}{}
	\textbf{Constructing $H_i$:} We construct $H_i$ in two steps; initially $H_i$ contains no edges.
	\begin{itemize}[noitemsep]
		\item \textbf{(Step 1).~} We  add to $H_i$ every edge of $E^{\sigma}_{i}$ corresponding to an edge in $\me^{\take}_i$. 

		\item \textbf{(Step 2).~} Let 
		$\mathcal{J}_i$ be a subgraph of $\mg_i$ induced by $\me_i^{\reduce}$. We show in  \Cref{clm:Ki-clustergraph} that $\mk_i$ is a $(L_i/(1+\psi),\eps,\beta)$-cluster graph (\Cref{def:ClusterGraph-Param}) w.r.t  $H_{< L_{i-1}}$. We  run \hyperlink{SPHigh}{$\sso$} on $\mathcal{J}_i$ to obtain a set of edges $F$. We then add every edge in $F$ to $H_i$.
	\end{itemize}
\end{tcolorbox}

\paragraph{Analysis.~} We first show that the input to \hyperlink{SPHigh}{Algorithm $\ma$} satisfies its requirement.

\begin{claim}\label{clm:Ki-clustergraph}$\mathcal{J}_i$  is a $(L, \eps, \beta)$-cluster graph with $L = {L_i/(1+\psi)}$, $\beta = 2g$,  and  $H_{< L} = H_{< L_{i-1}}$.
\end{claim} 
\begin{proof}
	We verify all properties in \Cref{def:ClusterGraph-Param}. Properties (1) and (2) follow directly from the definition of $\mathcal{J}_i$. Since we set $\psi = \frac{1}{250}$, every edge $(\varphi_{C_1},\varphi_{C_2}) \in \me^{\reduce}_i$  has $L_i/(1+\psi) \leq \omega(\varphi_{C_1},\varphi_{C_2})\leq L_i$. Since $L = L_i/(1+\eps)$, we have that $L \leq \omega(\varphi_{C_1},\varphi_{C_2}) \leq (1+\psi)L \leq 2L$; this implies property (3). By property \hyperlink{P3}{(P3)}, we have $\dm(H_{<L}[C]) \leq gL_{i-1} = g\eps L_i = g \eps (1+\psi) L \leq 2g \eps L = \beta \eps L$ when $\eps < 1$. Thus, $\mathcal{J}_i$ is a   $(L, \eps, \beta)$-cluster graph with the claimed values of the parameters. \qed 
\end{proof}

In  \Cref{lm:Hi-StretchT2} and \Cref{lm:Hi-WeightT2} below,  we bound the stretch of edges in $F^{\sigma}_i$ and the weight of $H_i$, respectively. Recall that $F^\sigma_{i}$ is the set of edges in $E^{\sigma}_i$ that correspond to $\mathcal{E}_i$.

\begin{lemma}\label{lm:Hi-StretchT2} For every edge $(u,v) \in F^\sigma_{i}$, $d_{H_{< L_i}}(u,v) \leq t(1+ s_{\sso}(2g)\eps)w(u,v)$.
\end{lemma}
\begin{proof}
	By construction, edges in $F^{\sigma}_i$ that correspond to $\me_i^{\take}$ are added to $H_i$ and hence have stretch $1$. By Item (2) of \Cref{lm:Clustering}, edges in $F^{\sigma}_i$ that correspond to $\me_i^{\redunt}$ have stretch $2 \leq t$ in $H_{< L_i}$. Thus, it remains to focus on edges corresponding to $\me_i^{\reduce}$. Let $(\varphi_{C_u},\varphi_{C_v}) \in \me^{\reduce}_i$ be the edge corresponding to an edge $(u,v \in F^{\sigma}_i$. 	Since we add all edges of $F$ to $H_i$, by property (2) of \hyperlink{SPHigh}{$\sso$}, the stretch of edge $(u,v)$ in $H_{< L_i}$ is at most $t(1+s_{\sso}(\beta)\eps) = t(1+s_{\sso}(2g)\eps)$ since $\beta  = 2g$ by \Cref{clm:Ki-clustergraph}.\qed
	\end{proof}

 Let $\msttilde^{in}_i(\mx) = \me(\mx)\cap \msttilde_i$ for each $\mx \in \mathbb{X}$. Let $\msttilde^{in}_i = \cup_{\mx \in \mathbb{X}}(\me(\mx)\cap \msttilde_i)$ be the set of $\msttilde_i$ edges that are contained in subgraphs in $\mathbb{X}$.  We have the following observations.

\begin{observation}\label{obs:supportingPropHiT2}
	\begin{enumerate}[noitemsep]
	\item[(1)]  $\sum_{\mx \in\mathbb{X}} \Delta^+_{i+1}(\mx) = (\Delta_{i+1} + w(\msttilde^{in}_i))$. Furthermore, $(\Delta_{i+1} + w(\msttilde^{in}_i))\geq 0$.
	\item[(2)] $\sum_{i\in \mathbb{N}^+} \msttilde^{in}_i\leq w(\mst)$.
	\end{enumerate}
\end{observation}
\begin{proof}
	We observe that $\sum_{\mx \in\mathbb{X}} \Delta^+_{i+1}(\mx) = \sum_{\mx \in \mathbb{X}} \left(\Delta_{i+1}(\mx) + w(\msttilde^{in}_i(\mx))\right) = (\Delta_{i+1} + w(\msttilde^{in}_i))$ by \Cref{clm:localPotenDecomps}. Furthermore, since $ \Delta^+_{i+1}(\mx)\geq 0$ by Item (3) of \Cref{lm:Clustering}, $(\Delta_{i+1} + w(\msttilde^{in}_i))\geq 0$. Thus, Item (1) holds.
	By the definition, the sets of corresponding edges of $\msttilde^{in}_i$ and $\msttilde^{in}_j$ are disjoint for any $i\not=j \geq 1$; this implies Item (2). 	\qed
\end{proof}

\begin{lemma}\label{lm:Hi-WeightT2}  $w(H_i) \leq \lambda \Delta_{i+1} + a_i$ for $\lambda = O(\chi \eps^{-1} )$  and $a_i =   O(\chi \eps^{-1} )w(\msttilde^{in}_i) + O(L_i/\eps)$.  
\end{lemma}
\begin{proof}  First, we consider the non-degenerate case. Note by the \hyperlink{HiConstT2}{construction of $H_i$} that we do not add any edge  corresponding to an edge in $\me^{\redunt}_i$ to $H_i$. Thus, we only need to consider edges in $\me^{\take}_i \cup \me^{\reduce}_i$. Let $\mv_i^{+} = \cup_{\mx \in \mathbb{X}^+}\mv(\mx)$ and  $\mv_i^{-} = \cup_{\mx \in \mathbb{X}^-}\mv(\mx)$. By \Cref{obs:supportingPropHiT2}, any edge in $\me^{\take}_i$ incident to a node in $\mv_i^{-}$ is also incident to a node in $\mv_i^{+}$.   Let $F^{(a)}_i$ be the set of edges added to $H_i$ in the construction in Step $a$, $a\in \{1,2\}$.
	
		By Item (3) of \Cref{obs:supportingPropHiT2}, $\me(\mx)\cap \me_i  = \emptyset$ if $\mx \in \mathbb{X}^{-}$. By the construction in Step 1, $F^{(1)}_i$ includes edges in $E^{\sigma}_{i}$ corresponding to $\me_i^{\take}$. By Item (1) in \Cref{lm:Clustering}, the total weight of the edges added to $H_i$ in Step 1 is:
	\begin{equation}\label{eq:Fi1T2}
		\begin{split}
			w(F^{(1)}_i)  &=  \sum_{\mx \in \mathbb{X}^{+}} O(|\mv(\mx)|) L_i \stackrel{\mbox{\tiny{\cref{eq:averagePotential-t2}}}}{=}  O(\frac{1}{\eps})\sum_{\mx \in \mathbb{X}^{+}} \Delta^+_{i+1}(\mx)\\
			&= O(\frac{1}{\eps})\sum_{\mx \in\mathbb{X}} \Delta^+_{i+1}(\mx)  \qquad \mbox{(since $\Delta^+_{i+1}(\mx)\geq 0 \quad \forall \mx \in \mathbb{X}$  by Item (3) in \Cref{lm:Clustering})} \\
			&= O(\frac{1}{\eps})(\Delta_{i+1} + w(\msttilde^{in}_i)) \qquad \mbox{(by Item (1) of \Cref{obs:supportingPropHiT2})}~.
		\end{split}
	\end{equation}

	Next, we bound $w(F^{(2)}_i)$.  By Item (1) of \Cref{lm:Clustering}, there is no edge in $\me^{\reduce}_i$ incident to a node in $\mv_i^{-}$. Thus, $\mv(\mathcal{J}_i) \subseteq \mv_i^{+}$.	 By property (1) of \hyperlink{SPHigh}{$\sso$}, it follows that
	\begin{equation}\label{eq:Fi3T2}
		\begin{split}
			w(F^{(2)}_i)  &~\leq~  \chi |\mv(\mathcal{J}_i)|  L_i \leq \chi|\mv_i^{+}| L_i =  \chi\sum_{\mx \in \mathbb{X}^{+}}|\mv(\mx)| L_i\\
			&\stackrel{\mbox{\tiny{\cref{eq:averagePotential-t2}}}}{=}  O(\chi/\eps)\sum_{\mx \in \mathbb{X}^{+}} \Delta^+_{i+1}(\mx)  \\   
			& =     O(\chi /\eps)\sum_{\mx \in \mathbb{X}} \left(\Delta_{i+1}(\mx) + w(\msttilde^{in}_i(\mx))\right)\\
			&= O(\chi/\eps)(\Delta_{i+1} + w(\msttilde^{in}_i)) \qquad \mbox{(by Item (1) of \Cref{obs:supportingPropHiT2})}~.
		\end{split}
	\end{equation} 
	
	By \Cref{eq:Fi1T2,eq:Fi3T2}, we conclude that:
	\begin{equation}\label{eq:Hi-nondegenT2}
		\begin{split}
			w(H_i) &=  O(\chi/\eps) (\Delta_{i+1} + w(\msttilde^{in}_i)) \leq \lambda(\Delta_{i+1} + w(\msttilde^{in}_i))
		\end{split}
	\end{equation} 
	for some $\lambda =  O(\chi/\eps) $.

	It remains to consider the degenerate case. By Item (4) of \Cref{lm:Clustering}, we only add to $H_i$ edges corresponding to $\me^{\take}_i$, and there are $O(1/\eps)$ such edges. Thus, we have:
	\begin{equation}\label{eq:Hi-degenT2}
		w(H_i) = O(\frac{L_i}{\eps}) \leq  \lambda\cdot (\Delta_{i+1} + w(\msttilde^{in}_i)) + O(\frac{L_i}{\eps}), 
	\end{equation} 
since $\Delta_{i+1} + w(\msttilde^{in}_i) \geq 0$ by Item (1) in \Cref{obs:supportingPropHiT2}. Thus, the lemma follows from \Cref{eq:Hi-degenT2,eq:Hi-nondegenT2}. \qed
\end{proof}

We are now ready to prove \Cref{lm:ConstructClusterHi} for the case $t\geq 2$, which we restate below.

\HiConstruction*
\begin{proof} The fact that subgraphs in $\mathbb{X}$ satisfy the three properties (\hyperlink{P1'}{P1'})-(\hyperlink{P3'}{P3'}) with constant $g=223$ follows from Item (5) of \Cref{lm:Clustering}.  The stretch in $H_{< L_i}$ of edges in $F^{\sigma}_{i}$ follows from \Cref{lm:Hi-StretchT2}.
	
	By  \Cref{lm:Hi-WeightT2}, $w(H_i) \leq \lambda \Delta_{i+1} + a_i$ where  $\lambda = O(\chi \eps^{-1} )$  and $a_i =   O(\chi \eps^{-1} )w(\msttilde^{in}_i) + O(L_i/\eps)$. It remains to show that $A = \sum_{i\in \mathbb{N}^+}a_i = O(\chi \eps^{-1} )$.  Observe that
	\begin{equation*}
		\sum_{i\in \mathbb{N}^+}O(\frac{L_i}{\eps}) ~=~  O(\frac{1}{\epsilon}) \sum_{i=1}^{i_{\max}} \frac{L_{i_{\max}}}{\epsilon^{i_{\max}-i}} ~=~ O(\frac{L_{i_{\max}}}{\epsilon(1-\epsilon)}) ~=~ O(\frac{1}{\epsilon}) w(\mst)~;
	\end{equation*}
	here $i_{\max}$ is the maximum level. The last equation is due to that $\eps \leq 1/2$  and every edge has weight at most $w(\mst)$ since the weight of every is the shortest distance between its endpoints. By Item (2) of \Cref{obs:supportingPropHiT2},  $\sum_{i\in \mathbb{N}^+} \msttilde^{in}_i\leq w(\mst)$.  Thus, $A = O(\chi/\eps) + O(1/\eps) = O(\chi /\eps)$ as desired.   \qed
\end{proof}

\subsection{Clustering} \label{subsec:clusteringT2}

In this section, we give a construction of the set of subgraphs $\mathbb{X}$ of the cluster graph $\mg_i$ as claimed in \Cref{lm:Clustering}. Our construction builds on the construction described in our companion work~\cite{LS21}. However, there are two specific goals we would like to achieve: the total degree of nodes in each subgraph $\mx$ in $\mg_i^{\take}$  is $O(|\mv(\mx)|)$, and the average potential change of each node (up to some edge cases) is $\Omega(\eps L_i)$ (instead of $\Omega(\eps^2 L_i)$ as achieved in \cite{LS21}),

Our construction  has 6 main steps (Steps 1-6). The first five steps are similar to the first five steps in the construction of \cite{LS21}. The major differences are in Step 2 and Step 4. In particular, in Step 2, we need to apply a clustering procedure of~\cite{LS19} to guarantee that the formed clusters have large average potential change. In Step 4, by using the fact that the stretch is at least 2, we form subgraphs in such a way that the potential change of the formed subgraphs is large. Step 6 is new in this paper. The idea is to post-process clusters formed in Steps 1-5 to form larger subgraphs that are trees, and hence, the average degree of nodes is $O(1)$. For those that are not grouped in the larger subgraphs,  the total degree of the nodes in each subgraph  is $O(1/\eps)$, which is at most the number of nodes.   In this step, we also rely on the fact that the stretch $t\geq 2$.

Now we give the details of the construction. Recall that $g$ is a constant defined in \hyperlink{P3}{property (P3)}, and that $\msttilde_{i}$ is a spanning tree of $\mg_i$ by Item (2) in \Cref{def:GiProp}. We reuse the construction in  Lemma 5.1~\cite{LS21} for Step 1 which applies to the subgraph $\mk_i$ of $\mg_i$ with edges in $\me_i$, as described by the following lemma.

\begin{lemma}[Step 1, Lemma 6.1~\cite{LS21}]\label{lm:Clustering-Step1T2} Let $\mv^{\high}_i = \{\varphi_{C} \in \mv: \varphi_{C} \mbox{ is incident to at least $\frac{2g}{\zeta\eps}$ edges in } \me_i\}$. Let $\mv_i^{\high+}$ be obtained from $\mv^{\high}_i$  by adding all neighbors that are connected to nodes in $\mv^{\high}_i$ via edges in $\me_i$. We can construct in polynomial time a collection of node-disjoint subgraphs $\mathbb{X}_1$ of $\mk_i =(\mv_i, \me_i)$ such that:
	\begin{enumerate}[noitemsep]
		\item[(1)] Each subgraph $\mx \in \mathbb{X}_1$ is a tree.
		\item[(2)] $\cup_{\mx \in \mathbb{X}_1}\mv(\mx) = \mv^{\high+}_i$.
		\item[(3)] $L_i \leq \adm(\mx) \leq (6+7\eta)L_i$, assuming that every node of $\mv_i$ has weight at most $\eta L_i$.
		\item[(4)] $\mx$ contains a node in $\mv^{\high}_i$ and all of its neighbors in $\mk_i$. In particular,  this implies $|\mv(\mx)|\geq \frac{2g}{\zeta\eps}$.
	\end{enumerate}
\end{lemma}

We note \Cref{lm:Clustering-Step1T2} is slightly more general than Lemma 6.1~\cite{LS21} in that we parameterize the weights of nodes in $\mv_i$ by $\eta L_i$. Clearly, we can choose $\eta = g\eps \leq 1$ when $\eps \leq 1/g$ since every node in $\mv_i$ has a weight at most $g\eps L_i$ by property \hyperlink{P3'}{(P3')} for level $i-1$. By parameterizing the weights, it will be more convenient for us to use the same construction again in Step 6 below.

Given a tree $T$, we say that a node $x\in T$ is \emph{$T$-branching} if it has degree at least 3 in $T$.  For brevity, we shall omit the prefix $T$ in ``$T$-branching'' whenever this does not lead to  confusion.  Given a forest $F$, we say that $x$ is \emph{$F$-branching} if it is $T$-branching for some tree $T\subseteq F$. Our construction of Step 2 uses the following lemma by \cite{LS19}.

\begin{lemma}[Lemma 6.12, full version~\cite{LS19}]\label{lm:tree-clustering} Let $\mt$ be a tree with vertex weights and edge weights. Let $L, \eta, \gamma,\beta$  be parameters where $\eta \ll \gamma \ll 1$ and $\beta\geq 1$. Suppose that for any vertex $v\in \mt$ and any edge $e\in \mt$, $w(e) \leq w(v) \leq \eta L$ and $w(v)\geq \eta L/\beta$. There is a polynomial-time algorithm that finds a collection of vertex-disjoint subtrees $\mathbb{U} = \{\mt_1,\ldots, \mt_k\}$ of $\mt$ such that:
	\begin{enumerate}[noitemsep]
		\item[(1)] $\adm(\mt_i) \leq 190\gamma L$ for any $1\leq i \leq k$.
		\item[(2)] Every branching node is contained in some tree in $\mathbb{U}$. 
		\item[(3)] Each tree $\mt_i$ contains a $\mt_i$-branching node $b_i$ and three internally node-disjoint paths $\mp_1,\mp_2, \mp_3$ sharing $b_i$ as the same endpoint, such that $\adm(\mp_1\cup \mp_2) = \adm(\mt_i)$ and $\adm(\mp_3 \setminus \{b_i\})= \Omega(\adm(\mt_i)/\beta)$. We call $b_i$ the \emph{center} of $\mt_i$.
		\item[(4)] Let $\overline{\mt}$ be obtained by contracting each subtree of $\mathbb{U}$ into a single node. Then each $\overline{\mt}$-branching node corresponds to a sub-tree  of augmented diameter at least $\gamma L$.
	\end{enumerate}
\end{lemma}

\begin{figure}[!h]
	\begin{center}
		\includegraphics[width=0.8\textwidth]{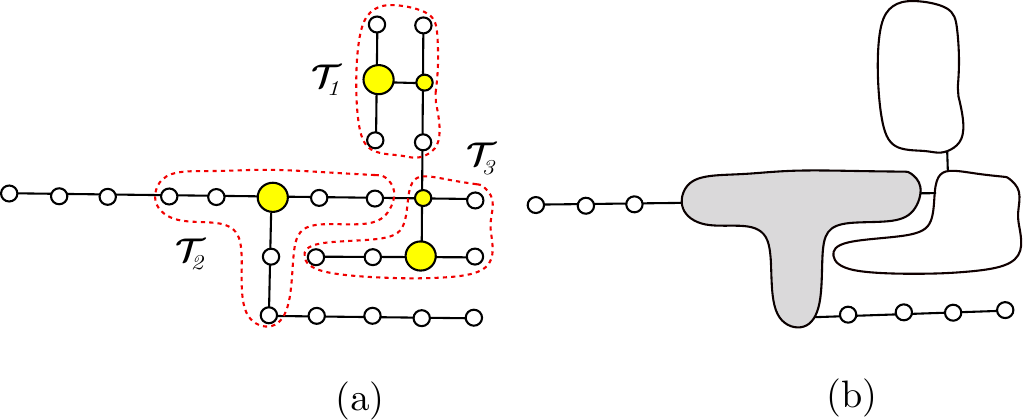}
	\end{center}
	\caption{(a) A collection $\mathbb{U} = \{\mt_1,\mt_2,\mt_3\}$ of a tree $\mt$ as in Lemma~\ref{lm:tree-clustering}. Yellow nodes are $\mt$-branching nodes. Big yellow nodes are the centers of their corresponding subtrees in $\mathbb{U}$. (b)  The shaded node in $\overline{\mt}$ is a $\overline{\mt}$-branching node and has an augmented diameter of at least $\gamma L$.}
	\label{fig:T-clustering}
\end{figure}

Let $\treeClustering(\mt,L,\eta,\gamma,\beta)$ be the output of \Cref{lm:tree-clustering} for input $\mt$ and parameters $L,\eta,\gamma,\beta$. 

\noindent See an illustration of Lemma~\ref{lm:tree-clustering}  in Figure~\ref{fig:T-clustering}. We are now ready to describe Step 2. Recall that $\zeta = 1/250$ is the constant in property \hyperlink{P3'}{(P3')}

\begin{lemma}[Step 2]\label{lm:Clustering-Step2T2} Let $\Ftilde^{(2)}_i$ be the forest obtained from $\msttilde_{i}$ by removing every node in $\mv^{\high+}_i$ (defined in \Cref{lm:Clustering-Step1T2}). Let $\mathcal{U} = \cup_{\tilde{T}\in \Ftilde^{(2)}_i} \treeClustering(\tilde{T},L_i,g\eps,\zeta,g/\zeta)$ and  $\mathbb{X}_2 = \{\tilde{T} \in \mathcal{U}: \adm(\tilde{T})\geq \zeta L_i\}$. Then, for every $\mx \in \mathbb{X}_2$, 
	\begin{enumerate}[noitemsep]
		\item[(1)] $\mx$ is a subtree of $\msttilde_{i}$.
		\item[(2)] $\zeta L_i \leq \adm(\mx)\leq L_i$.
		\item[(3)] $|\mv(\mx)| = \Omega(\frac{1}{\epsilon})$  when $\epsilon \leq 2/g$. 
		\item[(4)] $\Delta^+_{i+1}(\mx)  = \Omega(L_i)$.
 	\end{enumerate}

Furthermore, let $\Fbar^{(3)}_i$ be obtained from $\Ftilde^{(2)}_i$ by removing every tree in $\mathbb{U}$ that is added to $\mathbb{X}_2$, and contracting each remaining tree in $\mathbb{U}$ into a single node.  Then every tree $\Tbar \subseteq \Fbar^{(3)}_i$ is a path.
\end{lemma}
\begin{proof} We observe that Item (1) follows directly from the construction. $\adm(\mx) \geq \zeta L_i $ follows from the definition of $\mathbb{X}_2$. By Item (1) in \Cref{lm:tree-clustering}, $\adm(\mx)\leq 190\zeta  = \frac{190}{250} L_i \leq L_i$. Thus, Item (2) follows.  
	
	To show Item (3), let $\md$ be an augmented diameter path of $\mx$. Note that  $\adm(\md)\geq \zeta L_i$ by Item (2). Furthermore, every edge has a weight of at most $\bar{w} \leq L_{i-1}$ and every node has a weight of in $[L_{i-1},gL_{i-1}]$ by \hyperlink{P3'}{property (P3')}. Thus,  $\md$ has at least $\frac{\adm(\md)}{2gL_{i-1}} ~\geq~ \frac{\zeta L_i}{2g\eps L_i} = \Omega(\frac{1}{\epsilon})$ nodes; this implies Item (3).
	
	Finally, we show Item (4). Let $\varphi_b$ be the center node of $\mx$. By Item (3) in \Cref{lm:tree-clustering}, there are three internally node-disjoint paths $\mp_1,\mp_2, \mp_3$ sharing $\varphi_b$ as the same endpoint. There must be an least one path, say $\mp_1$, such that $\mp_1\cap \md \subseteq \{\varphi_b\}$. That is, $\mp_1$ is internally disjoint from the diameter path $\md$. Also by Item (3) in \Cref{lm:tree-clustering}, $\adm(\mp_1\setminus \{\varphi_{b}\}) = \Omega(\adm(\mx)/\beta) = \Omega(\zeta L_i/(g/\zeta)) = \Omega(L_i)$. Observe that:
	
	\begin{equation*}
		\begin{split}
			\Delta^+_{i+1}(\mx) =  \left(\sum_{\varphi \in \mx}\omega(\varphi) + \sum_{ e\in \me(\mx)} \omega(e)\right) - \adm(\mx) \geq \adm(\mp_1\setminus \{\varphi_{b}\})  = \Omega(L_i)~,
		\end{split}
	\end{equation*}
	as claimed.	\qed
\end{proof}

By Item (4) of \Cref{lm:Clustering-Step2T2}, the amount of potential change of subgraphs in $\mathbb{X}_2$ is $\Omega(L_i)$, while in subgraphs in $\mathbb{X}_2$ in the construction of our companion work~\cite{LS21} only have $\Omega(\eps L_i)$ potential change. 

We note that there might be isolated nodes in $\Fbar^{(3)}_i$, which we still consider as paths.  We refer to nodes in $\Fbar^{(3)}_i$ that are contracted from $\mathcal{U}$ as \emph{contracted nodes},  and nodes that correspond to original nodes of $\Ftilde^{(2)}_i$ as \emph{uncontracted nodes}. For each node $\bar{\varphi} \in \Fbar^{(3)}_i$, we abuse notation by denoting $\bar{\varphi}$ the subtree of $\Ftilde^{(2)}_i$ corresponding to the node $\bar{\varphi}$; $\bar{\varphi}$  could be a single node in $\Ftilde^{(2)}_i$ for the uncontracted case. We then define the weight function of $\bar{\varphi}$ as follows:
\begin{equation}\label{eq:weightContractedNode}
	\omega(\bar{\varphi}) = \adm(\bar{\varphi})
\end{equation}

Where in the RHS of \Cref{eq:weightContractedNode}, we interpret $\bar{\varphi}$ as a subtree of $\Ftilde^{(2)}_i$ with weights on nodes an edges. 
\begin{observation}\label{obs::weightContractedNode} $\omega(\bar{\varphi}) \leq  \zeta L_i $ for every node $\bar{\varphi} \in \Fbar^{(3)}$.
\end{observation}
\begin{proof}
	This is because any node of weight at least $\zeta L_i$ is processed in Step 2. \qed
\end{proof}
For each subpath $\Pbar \subseteq \Fbar^{(3)}_i$, let $\tilde{P}^{\uncontract}$ be the subtree of $\msttilde_{i}$ obtained by uncontracting the contracted nodes in $\Pbar$. We have:

\begin{observation}\label{obs:contrct-uncontract}   For every path $\Pbar \subseteq \Fbar^{(3)}_i$,  $\adm(\tilde{P}^{\uncontract}) \leq \adm(\Pbar)$.
\end{observation}
\begin{proof} The observation follows directly from the definition of the weights of nodes in $\Pbar$ in \Cref{eq:weightContractedNode}. \qed
\end{proof}

We say that a node $\bar{\varphi} \in \Fbar^{(3)}_i$ is \emph{incident to an edge} $\mbe \in \msttilde_{i}\cup \me_i$ if one endpoint of $\mbe$ belongs to $\bar{\varphi}$.

\paragraph{Step 3: Augmenting $\mathbb{X}_1\cup \mathbb{X}_2$.~}\hypertarget{S3T2}{}   Let $\Fbar^{(3)}_i$ be the forest obtained in Item (4b) in \Cref{lm:Clustering-Step2T2}. Let $\bar{A}$ be the set of all nodes $\bar{\varphi}$ in $\Fbar^{(3)}_i$ such that there is (at least one)  $\msttilde_i$ edge $\mbe = (\varphi_1,\varphi_2)$ between a node $\varphi_1 \in \bar{\varphi}$, and a node $\varphi_2 \in \mx$ for some subgraph $\mx \in\mathbb{X}_1\cup \mathbb{X}_2$. Then, for each node  $\bar{\varphi}\in \bar{A}$, we augment $\mx$ by adding $\bar{\varphi}$ and $\mbe$ to $\mx$.

\begin{lemma}\label{lm:Clustering-Step3} The augmentation in Step 3 increases the augmented diameter of each subgraph in  $\mathbb{X}_1\cup \mathbb{X}_2$ by at most $4L_i$ when $\eps \leq 1/g$. \\
	Furthermore, let $\Fbar^{(4)}_i$ be the forest obtained from $\Fbar^{(3)}_i$ by removing every node in $\bar{A}$. Then, for every path $\Pbar \subseteq \Fbar^{(4)}_i$,  at least one endpoint $\bar{\varphi} \in \Pbar$ has an $\msttilde_{i}$ edge to a subgraph of $\mathbb{X}_1\cup \mathbb{X}_2$, unless $\mathbb{X}_1\cup \mathbb{X}_2 = \emptyset$. 
\end{lemma}
\begin{proof}
	Since every $\msttilde_{i}$ edge has a weight of at most $\bar{w}\leq L_i$ and every node has a weight of at most $\zeta L_i \leq L_i$ by \Cref{obs::weightContractedNode}, the augmentation in Step 3 increases the augmented diameter of each subgraph in  $\mathbb{X}_1\cup \mathbb{X}_2$ by at most $2(\bar{w} + 2L_i) ~\leq~ 4L_i$. The second claim about the property of  $\Fbar^{(4)}_i$ follows directly from the construction of \hyperlink{S3T2}{Step 3}.  \qed
\end{proof}

\paragraph{Required definitions/preparations for Step 4.~} Let $\Fbar^{(4)}_i$ be the forest obtained from $\Fbar^{(3)}_i$ as described in \Cref{lm:Clustering-Step3}. We call every path of augmented diameter at least $6L_i$ of $\Fbar^{(4)}_i$ a \emph{long path}.

\begin{quote}
	\textbf{Red/Blue Coloring.~}\hypertarget{RBColoring}{}  Given a path $\Pbar\subseteq \Fbar^{(4)}_i$, we color their nodes red or blue. If a node has augmented distance at most $L_i$ from at least one of the path's endpoints, we color it red; otherwise, we color it blue. Observe that each red node belongs to the suffix or prefix of $\Pbar$; the other nodes are colored blue. 
\end{quote}

For each blue node $\bnu$ in a long path $\Pbar$, we denote by $\Ibar(\bnu)$ the subpath of $\Pbar$ containing every node  within an augmented distance (in $\Pbar$) at most $(1-\psi)L_i$ from $\bnu$. We call $\Ibar(\bnu)$ the \emph{interval} of $\bnu$. Recall that $\psi = 1/250$ is the constant defined in \Cref{eq:Esigmai}. 

We define the following set of edges between nodes of $\Fbar^{(4)}_i$.
	\begin{equation}\label{eq:Ebar-i}
		\bar{\me}_i = \{(\bmu,\bnu) | \exists \mu \in \bar{\mu}, \nu \in \bar{\nu} \mbox{ and }(\mu,\nu)\in \me_i\}.
	\end{equation}
We note that there is no edge in $\me_i$ whose nodes belong to the same tree, say $\bar{\mu}$, that corresponds to a node in $\Fbar^{(4)}_i$, because such an edge, say $\mbe$, will have weight at most  $\omega(\bar{\mu}) \leq \zeta L_i < L_i/2 < \omega(\mbe)$, a contradiction. 

Next, we define the weight: 
\begin{equation}\label{eq:Ebar-weight}
	\omega(\bmu,\bnu) = \min_{\substack{\mu\in\bar{\mu}, \nu\in\bar{\nu}\\(\mu,\nu)\in \me_i}}\omega(\mu,\nu)
\end{equation}

That is, the weight of edges $(\bmu,\bnu)$ is the minimum weight over all edges between two trees $\bar{\mu}$ and $\bar{\nu}$. We then denote $(\mu,\nu)$ the edge in $\me_i$ corresponding to an edge $(\bmu,\bnu) \in \bar{\me}_i$. Next, we define:

\begin{equation}\label{eq:Ebar-farclose}
	\begin{split}	
		\bar{\me}^{far}_i(\Fbar^{(4)}) &= \{(\bnu,\bmu) \in \bar{\me}_i | color(\bnu) = color({\bmu}) = blue  \mbox{ and }\Ibar(\bnu)\cap \Ibar(\bmu) = \emptyset\}\\
		\bar{\me}^{close}_i(\Fbar^{(4)})  &= \{(\bnu,\bmu) \in \bar{\me}_i | color(\bnu) = color({\bmu}) = blue  \mbox{ and }\Ibar(\bnu)\cap \Ibar(\bmu)\not= \emptyset\}
	\end{split}
\end{equation}
We note that the definition of $\bar{\me}^{far}_i(\Fbar^{(4)})$ and $\bar{\me}^{close}_i(\Fbar^{(4)})$ depends on the underlying forest $\Fbar^{(4)}$.

\begin{lemma}[Step 4]\label{lm:Clustering-Step4} Let $\Fbar^{(4)}_i$ be the forest obtained from $\Fbar^{(3)}_i$ as described in \Cref{lm:Clustering-Step3}. We can construct a collection $\mathbb{X}_4$ of subgraphs of $\mg_i$ such that every $\mx\in \mathbb{X}_4$:
	\begin{enumerate}[noitemsep]
		\item[(1)] $\mx$ is a tree and  contains a single edge in $\me_i$.
		\item[(2)] $L_i \leq \adm(\mx)\leq 5L_i$.
		\item[(3)]  $|\mv(\mx)| = \Omega(1/\eps)$ when $\epsilon \leq 1/8$. 
		\item[(4)] $\Delta_{i+1}^{+}(\mx) = \Omega(L_i)$.
	\end{enumerate}
 Let $\Fbar^{(5)}_i$ be obtained from $\Fbar^{(4)}_i$ by removing every node whose corresponding tree is contained in subgraphs of $\mathbb{X}_4$. If we apply \hyperlink{RBColoring}{Red/Blue Coloring} to each path of augmented diameter at least $6L_i$ in $\Fbar^{(5)}_i$, then $\bar{\me}^{far}_i(\Fbar^{(5)}) = \emptyset$. Furthermore,  for every path $\Pbar \subseteq \Fbar^{(5)}_i$,  at least one endpoint of $\Pbar$ has an $\msttilde_{i}$ edge to a subgraph of $\mathbb{X}_1\cup \mathbb{X}_2\cup \mathbb{X}_4$, unless $\mathbb{X}_1\cup \mathbb{X}_2 \cup \mathbb{X}_4 = \emptyset$. 
\end{lemma}
\begin{proof} We only apply the construction to long paths of $\Fbar^{(4)}_i$; those that have  augmented diameter at least $6L_i$. We use the following claim from~\cite{LS21}; the only difference is that nodes have weights at most $\zeta L_i$ instead of $g\eps L_i$ and the interval $\bar{\mi}(\bnu)$ contains nodes within augmented distance $(1-\psi)L_i$ from $\bnu$ instead of within augmented distance $L_i$, which ultimately leads to changes in the upper bound and the lower bound of the augmented diameter of the interval.
	
	\begin{claim}[Claim 6.5 in~\cite{LS21}, adapted]\label{clm:Interval-node}
		For any blue node $\nu$, it holds that
		\begin{itemize}[noitemsep]
			\item[(a)] $ (2 - 3\zeta - 2 \epsilon-2\psi)L_i \leq  \adm(\overline{\mathcal{I}}(\bar{\nu}))\leq 2(1-\psi)L_i $.
			\item[(b)]   	Denote by  $\overline{\mi}_1$ and $\overline{\mi}_2$  the two subpaths obtained by removing $\bnu$ from the path $\overline{\mathcal{I}}(\bar{\nu})$. 
			Each of these subpaths has augmented diameter at least $(1-2\zeta - \epsilon -\psi)L_i$.
		\end{itemize}
	\end{claim}
	
	We now construct $\mathbb{X}_4$, which initially is empty. 
	
	\begin{itemize}
		\item  Pick an edge $(\bnu,\bmu)$ with both blue endpoints and  form a subgraph $\bmx = \{(\bnu,\bmu)\cup \overline{\mi}(\bnu) \cup \overline{\mi}(\bmu)\}$. We remove  all nodes in  $\overline{\mi}(\bnu) \cup \overline{\mi}(\bmu) $ from the path or two paths containing $\bnu$ and $\bmu$, update the color of nodes in the new paths to satisfy \hyperlink{RBColoring}{Red/Blue Coloring}. We then uncontract nodes in $\bmx$ to obtain a subgraph $\mx$ of $\mg_i$, add $\mx$ to $\mathbb{X}_4$, and  repeat this step until it no longer applies.
	\end{itemize}
	
	We observe that Item (1) and  the last claim about  $\Fbar^{(5)}_i$ follow directly from the construction. For Item (2), we observe by Claim~\ref{clm:Interval-node}  that $\overline{\mathcal{I}}(\bnu)$ has augmented diameter at most $2L_i$ and at least $L_i$ when $\epsilon  \leq 1/8$, and  the weight of the edge $(\bmu,\bnu)$ is at most $L_i$. Thus, $L_i \leq \adm(\mx)\leq L_i + 2\cdot 2L_i = 5L_i$, as claimed.

	Let $\mi$ be the subtree of $\mst_i$ obtained by uncontracting nodes in $\overline{\mi}(\bnu)$.  By \Cref{clm:Interval-node}, $\adm(\overline{\mi}(\bnu)) \geq L_i = \Omega(L_i)$ when $\eps \leq 1/8$.  By \Cref{lm:size-MSTsubree}, which we show below,  $|\mv(\mi)| =  \Omega(1/\eps)$; this implies Item (3).
	
	We now focus on Item (4). Let $\overline{\mi}_1, \overline{\mi}_2, \overline{\mi}_3, \overline{\mi}_4$ be four paths obtained from $\overline{\mi}(\bmu)$ and $\overline{\mi}(\bnu)$ by removing $\bmu$ and $\bnu$. Let $\overline{\md}$ be the diameter path of $\bar{\mx}$. Then $\overline{\md}$ contains at most 2 paths among the four paths, and possibly contains edge $(\bnu,\bmu)$ as well. Since each path has augmented diameter at most $2L_i$ and $\omega(\bnu,\bmu) \leq L_i$, we have that:
	
	\begin{equation*}
		\begin{split}
			 \left(\sum_{\bar{\varphi} \in \bar{\mx}}\omega(\bar{\varphi}) + \sum_{ \mbe\in \me(\bar{\mx})\cap \msttilde_{i}} \omega(\mbe)\right) - \adm(\bar{\md}) &\geq  4(1-2\zeta - \epsilon -\psi)L_i - 3L_i \qquad \mbox{(by \Cref{clm:Interval-node})}\\
			 &\geq (1-8\zeta -4\eps - 4\psi)L_i  = \Omega(L_i)~,
		\end{split}
	\end{equation*}
when $\eps \leq 1/8$. Note that $\psi = \zeta = 1/250$. Furthermore, since $\adm(\mx) \leq \adm(\bar{\mx}) = \adm(\bar{\md})$, we have:

	\begin{equation*}
	\begin{split}
		\Delta^+_{i+1}(\mx) &=  \left(\sum_{\varphi \in \mx}\omega(\varphi) + \sum_{ \mbe\in \me(\mx)\cap \msttilde_{i}} \omega(\mbe)\right) - \adm(\mx) \\ &\geq  \left(\sum_{\varphi \in \mx}\omega(\varphi) + \sum_{ \mbe\in \me(\mx)\cap \msttilde_{i}} \omega(\mbe)\right)  - \adm(\bar{\md}) =  \Omega(L_i)~,
	\end{split}
\end{equation*}
as claimed. \qed		
\end{proof}

\begin{lemma}\label{lm:size-MSTsubree} Let $\overline{P} \subseteq \Fbar^{(3)}_i$ be a path of augmented diameter $\Omega(L_i)$. Then $|\mv(\tilde{P}^{\uncontract})| =  \Omega(1/\eps)$.
\end{lemma}
\begin{proof}
	Recall that $\tilde{P}^{\uncontract}$ is obtained by uncontracting every contracted node of  $\overline{P}$, and that is a subtree of $\mst_i$. 	
	Since the weight of each node is at least the weigh of each edge in $\tilde{P}^{\uncontract}$, we have
	\begin{equation*}
		\begin{split}
			\sum_{\varphi \in \tilde{P}^{\uncontract}}\omega(\varphi) \geq \left(\sum_{\varphi \in \tilde{P}^{\uncontract}}\omega(\varphi) + \sum_{ e\in \me(\tilde{P}^{\uncontract})} \omega(e)\right)/2 \geq \adm(\overline{P})/2 = \Omega(L_i).
		\end{split}
	\end{equation*}
	
	Furthermore, $\omega(\varphi) \leq g\eps L_i$ by property  \hyperlink{P3'}{(P3')} for level $i-1$. It follows that $|\mv(\tilde{P}^{\uncontract})| =\Omega( \frac{L_i}{g\eps L_i}) = \Omega(1/\eps)$ as desired. \qed
\end{proof}

\begin{remark}\label{remark:Clustering-Step4} Item (5) of \Cref{lm:Clustering-Step4} implies that for every edge $(\bmu,\bnu)\in \bar{\me}_i$ with both endpoints in $\Fbar^{(5)}_i$, either (i) the edge is in $ \bar{\me}^{close}_i(\Fbar^{(5)}_i)$, or (ii) at least one of the endpoints must belong to a low-diameter tree of $\Fbar^{(5)}_i$ or (iii) in a (red) suffix of a long path in $\Fbar^{(5)}_i$ of augmented diameter at most $L_i$.
\end{remark}

\paragraph{Step 5.~} Let $\Pbar$ be  a path in  $\Fbar^{(5)}_i$ obtained by Item (5) of \Cref{lm:Clustering-Step4}. We construct two sets of subgraphs, denoted by $\mathbb{X}^{\internal}_5$ and $\mathbb{X}^{\prefix}_5$, of $\mg_i$. The construction is broken into two steps. Step 5A is only applicable when $\mathbb{X}_1 \cup \mathbb{X}_2\cup \mathbb{X}_4 \not= \emptyset$.

\begin{itemize}
	\item (Step 5A)\hypertarget{5A}{}  If $\Pbar$ has augmented diameter at most $6L_i$, let $\mbe$ be an $\widetilde{\mst}_i$ edge connecting $\Ptilde^{\uncontract}$  and a node in some subgraph $\mx \in \mathbb{X}_1\cup \mathbb{X}_2 \cup \mathbb{X}_4$; $\mbe$ exists by \Cref{lm:Clustering-Step4}. We add both $\mbe$ and $\Ptilde^{\uncontract}$ to $\mx$.
	\item (Step 5B)\hypertarget{5B}{} 	Otherwise,  the augmented diameter of $\Pbar$ is at least $6L_i$. In this case, we greedily break $\Pbar$ into subpaths $\{\Qbar_1,\ldots, \Qbar_k\}$ such that for each $j\in [1,k]$, $\tilde{Q}^{\uncontract}_j$ has augmented diameter at least $L_i$ and at most $2L_i$.  If $\Qbar_j$ is connected to a node in a subgraph $\mx \in \mathbb{X}_1 \cup \mathbb{X}_2\cup \mathbb{X}_4$ via an  edge $e\in \msttilde_{i}$, we add $\tilde{Q}^{\uncontract}_j$ and $e$ to $\mx$.	If $\Qbar_j$ contains an endpoint of $\Pbar$, we add $\Qtilde_j^{\uncontract}$ to $\mathbb{X}^{\prefix}_5$; otherwise, we add  $\Qtilde_j^{\uncontract}$ to $\mathbb{X}^{\internal}_5$. 
\end{itemize}

In Step 5B, we want $\tilde{Q}_j^{\uncontract}$ to have augmented diameter at least $L_i$ (to satisfy property\hyperlink{P3'}{(P3')})  instead of  requiring $\adm(\Qbar_j) \geq L_i$ because a lower bound on the augmented diameter of $\Qbar_j$ does not translate to a lower bound on the augmented diameter of $\tilde{Q}^{\uncontract}_j$.

\begin{lemma}\label{lm:Clustering-Step5}  Every subgraph $\mx \in \mathbb{X}_5^{\internal} \cup \mathbb{X}_5^{\prefix}$ satisfies:
	\begin{enumerate}[noitemsep]
		\item[(1)] $\mx$ is a subtree of $\msttilde_{i}$.
		\item[(2)] $L_i \leq \adm(\mx)\leq 2 L_i$.
		\item[(3)] $|\mv(\mx)| = \Omega(1/\eps)$.
	\end{enumerate}
Furthermore, if $\mx \in \mathbb{X}_5^{\prefix}$, then $\mx$ the uncontraction of a prefix subpath $\Qbar$ of a long path $\Pbar$, and additionally, the (uncontraction of) other suffix $\Qbar'$ of   $\Pbar$ is augmented to a subgraph in $\mathbb{X}_1 \cup \mathbb{X}_2\cup \mathbb{X}_4$, unless $\mathbb{X}_1 \cup \mathbb{X}_2\cup \mathbb{X}_4 = \emptyset$.
\end{lemma}
\begin{proof} 
	Items (1) and (2) follow directly from the construction. Item (3) follows  from \Cref{lm:size-MSTsubree}.  The last claim about subgraphs in  $\mathbb{X}_5^{\prefix}$ follows from \Cref{lm:Clustering-Step4}.
	 \qed
\end{proof}

\begin{lemma}\label{lm:Adm-Xprime}Let $\mathbb{X}' = \mathbb{X}_1 \cup \mathbb{X}_2\cup \mathbb{X}_4 \cup \mathbb{X}_5^{\internal} \cup \mathbb{X}_5^{\prefix}$. Every node of $\mv_i$ is grouped to some subgraph in $\mathbb{X}'$. Furthermore, for every $\mx \in \mathbb{X}'$,
		\begin{enumerate}[noitemsep]
		\item[(1)] $\mx$ is a tree. Furthermore, if $\mx \not\in \mathbb{X}_4$, it is a subtree of $\msttilde_{i}$. 
		\item[(2)]  $\zeta L_i \leq \adm(\mx) \leq 31 L_i$ when $\eps \leq 1/g$.
		\item[(3)] $|\mv(\mx)| = \Omega(1/\eps)$.
	\end{enumerate}
\end{lemma}
\begin{proof}
 The fact that  every node of $\mv_i$ is grouped to some subgraph in $\mathbb{X}'$ follows directly from the construction. Observe that only subgraphs  in $\mathbb{X}'$ formed in Step 4 contain edges in $\me_i$, and such subgraphs are trees by Item (1)  of \Cref{lm:Clustering-Step4};  this implies Item (1). Item 3 follows directly from \Cref{lm:Clustering-Step1T2,lm:Clustering-Step2T2,lm:Clustering-Step4,lm:Clustering-Step5}. 
 
 We now focus on bounding $\adm(\mx)$. The lower bound on $\adm(\mx)$ follows directly from Item (3) of \Cref{lm:Clustering-Step1T2}, Items (2) of \Cref{lm:Clustering-Step2T2,lm:Clustering-Step4,lm:Clustering-Step5}. For the upper bound, we observe that if $\mx$ is formed in Step 1, it could be augmented further in Step 3, and hence, by Item (3) of \Cref{lm:Clustering-Step1T2} (here $\eta = g\eps$), and \Cref{lm:Clustering-Step3}, $\adm(\mx) \leq (6 + 7g\eps)L_i + 4L_i \leq 17L_i$ since $\eps \leq 1/g$. By Items (2) of \Cref{lm:Clustering-Step2T2,lm:Clustering-Step4,lm:Clustering-Step5}, $\adm(\mx) \leq 5L_i$ if $\mx$ is not initially formed in Step 1. Furthermore,  the augmentation in Step 5A and 5B increases $\adm(\mx)$ by at most $2(\bar{w} + 6L_i)\leq 14L_i$. This implies that, in any case, $\adm(\mx)\leq \max\{17L_i, 5L_i\} + 14L_i = 31L_i$. \qed
\end{proof}

Except for subgraphs in $\mathbb{X}_5^{\internal} \cup \mathbb{X}_5^{\prefix}$, we can show every subgraph $\mx \in \mathbb{X}_1 \cup \mathbb{X}_2\cup \mathbb{X}_4$ has large potential change: $\Delta_{i+1}(\mx) = \Omega(L_i)$. The last property that we need to complete the proof of \Cref{lm:Clustering} is to guarantee that the total degree of vertices in  $\mx \in  \mathbb{X}_2\cup \mathbb{X}_4\cup \mathbb{X}_5^{\internal} \cup \mathbb{X}_5^{\prefix}$ in $\mg^{\reduce}$ is $O(1/\eps)$ (we have not defined $\mg^{\reduce}$ yet). To this end, we need Step 6 (which is not required in our companion work~\cite{LS21}). The basic idea is that if any subgraph has many out-going edges in $\bar{\me}_i$ (defined in \Cref{eq:Ebar-i}), then we apply the clustering procedure in Step 1 to group it to a larger subgraph.

\paragraph{Required definitions/preparations for Step 6.~} We construct a graph $\doverline{\mk}_i(\doverline{\mv}_i, \doverline{\me}_i, \doverline{\omega})$ as follows. Each node $\doverline{\varphi}_{\mx} \in \doverline{\mv_i}$ corresponds to a subgraph $\mx \in \mathbb{X}'$.  We then set $\doverline{\omega}(\doverline{\varphi}_{\mx}) = \adm(\mx)$.  There is an edge $(\doverline{\varphi}_{\mx},\doverline{\varphi}_{\my}) \in \doverline{\me}_i$ between two \emph{different nodes} $\doverline{\varphi}_{\mx},\doverline{\varphi}_{\my}$  if there exists an edge $(\varphi_1,\varphi_2) \in \me_i$ between a node $\varphi_1 \in \mx$ and a node $\varphi_2 \in \my$. We set the weight $\doverline{\omega}(\doverline{\varphi}_{\mx},\doverline{\varphi}_{\my})$ to be the minimum weight over all edges in $\me_i$ between $\mx$ and $\my$. We call nodes of $\doverline{\mk}_i$ \emph{supernodes}.

We call $\doverline{\varphi}_{\mx}$ a \emph{heavy} supernode if $|\mv(\mx)|\geq \frac{2g}{\zeta\eps}$ or $\doverline{\varphi}_{\mx}$ is incident to at least $\frac{2g}{\zeta\eps}$ edges in $\doverline{\mk}_i$. Otherwise, we call  $\doverline{\varphi}_{\mx}$ a \emph{light} supernode. By definition of a heavy supernode and by Item (4) in \Cref{lm:Clustering-Step1T2}, if $\mx$ is formed in Step 1, then  $\doverline{\varphi}_{\mx}$  is a heavy supernode. We then do the following.

\begin{quote}
	We apply the construction in \Cref{lm:Clustering-Step1T2} to graph $\doverline{\mk}_i(\doverline{\mv}_i, \doverline{\me}_i, \doverline{\omega})$, where $\doverline{\mv}_i^{\high}$ is the set of heavy supernodes in $\doverline{\mk}$ and  $\doverline{\mv}_i^{\highp}$ is obtained from  $\doverline{\mv}_i^{\high}$ by adding neighbors in $\doverline{\mk_i}$. Let $\doverline{\mathbb{X}}_6$ be the set of subgraphs of $\doverline{\mk}_i(\doverline{\mv}_i, \doverline{\me}_i, \doverline{\omega})$ obtained by the construction. Every subgraph $\doverline{\mx} \in \doverline{\mathbb{X}}_6$ satisfies all properties in \Cref{lm:Clustering-Step1T2} with $\eta = 31$.
\end{quote}

 Let $\mathbb{X}_6$ be obtained from $\doverline{\mathbb{X}}_6$ by uncontracting supernodes. This completes our Step 6.
 
 \begin{lemma}\label{lm:Step6-T2-Prop} Every subgraph $\mx \in \mathbb{X}_6$ has $\zeta L_i \leq \adm(\mx) \leq 223L_i$.
 \end{lemma}
 \begin{proof}
 	Let $\doverline{\mx}$ be the subgraph in $\doverline{\mathbb{X}}_6$ that corresponds to $\mx$. By \Cref{lm:Adm-Xprime}, every node $\doverline{\varphi} \in \doverline{\mx}$ has weight $\doverline{\omega}(\doverline{\varphi}) \leq 31L_i$. Thus, by Item (3) in \Cref{lm:Adm-Xprime}, $\adm(\doverline{\mx})\leq (6 + 7\cdot  31)L_i = 223L_i$.	  \qed
 \end{proof}

In \Cref{subsec:X-T2} we construct the set of subgraphs $\mathbb{X}$, and show several properties of subgraphs in $\mathbb{X}$. In \Cref{subsec:E-T2}, we construct a partition of $\me_i$ into three sets $\me^{\take}_i, \me^{\redunt}_i$ and $\me_i^{\reduce}$, and prove \Cref{lm:Clustering}.

\subsubsection{Constructing $\mathbb{X}$}\label{subsec:X-T2}
 
 For each $i\in \{2,4,5\}$ let $\mathbb{X}_i^{-}$ be obtained from $\mathbb{X}_i$ by removing subgraphs corresponding to nodes in  $\doverline{\mv}_i^{\highp}$ (which then form subgraphs in $\mathbb{X}_6$).  We now define $\mathbb{X}$ and a partition of $\mathbb{X}$ into two sets $\mathbb{X}^{+}$ and 	$\mathbb{X}^{-}$ $\mathbb{X}^{\lowm}$ 
 as claimed in \Cref{lm:Clustering}. We distinguish two cases:
 
\paragraph{Degenerate Case.~} The degenerate case is the case where   $\mathbb{X}^{-}_1\cup \mathbb{X}^{-}_2\cup \mathbb{X}^{-}_4 = \mathbb{X}^{\internal}_5 =  \emptyset$. In this case, we set $\mathbb{X} = \mathbb{X}^{-} =  \mathbb{X}_5^{\internal} \cup \mathbb{X}_5^{\prefix}$, and $	\mathbb{X}^{+} = 	 \emptyset$. 

\paragraph{Non-degenerate case.~} If $\mathbb{X}^{-}_1\cup \mathbb{X}^{-}_2\cup \mathbb{X}^{-}_4 = \mathbb{X}_6 \not=  \emptyset$, we call this the non-degenerate case. In this case, we define.
\begin{equation}\label{eq:MathbbX}
	\begin{split}
		\mathbb{X}^{+} &=    \mathbb{X}^{-}_2\cup \mathbb{X}^{-}_4 \cup \mathbb{X}_5^{\prefix-} \cup \mathbb{X}_6, \quad
		\mathbb{X}^{-} = \mathbb{X}_5^{\internal -} \\
		\mathbb{X} &= \mathbb{X}^{+}\cup \mathbb{X}^{-}
	\end{split}
\end{equation}

We note that every subgraph in $\mathbb{X}_1$ corresponds to a heavy supernode in $\doverline{\mk_i}$ and hence, it will be grouped in some subgraph in $\mathbb{X}_6$.

In the analysis below, we only explicitly  distinguish the degenerate case from the non-degenerate case when it is necessary, i.e, in the proof Item (4) of \Cref{lm:Clustering}. Otherwise, which case we are in is either implicit from the context, or does not matter.

\begin{lemma}\label{lm:XProp} Let $\mathbb{X}$ be the subgraph as defined in \Cref{eq:MathbbX}. For every subgraph $\mx \in \mathbb{X}$, $\mx$ is a tree and satisfies the three properties (\hyperlink{P1'}{P1'})-(\hyperlink{P3'}{P3'}) with $g = 223$. Consequently, Item (5) of \Cref{lm:Clustering} holds.
\end{lemma}
\begin{proof} We observe that property \hyperlink{P1'}{(P1')} follows directly from the construction.  Property \hyperlink{P2'}{(P2')} follows from Item (3) of \Cref{lm:Adm-Xprime}.  Property \hyperlink{P3'}{(P3')} follows from \Cref{lm:Step6-T2-Prop}. 
	
	By Item (1) of \Cref{lm:Adm-Xprime}, every subgraph $\mx \in \mathbb{X}'$ is a tree. Since subgraphs in $\doverline{\mx}_6$ in the construction of Step 6 are trees, subgraphs in $\mathbb{X}$ are also trees. Thus, $|\me(\mx) \cap \me_i|  = O(|\mv(\mx)|)$. Furthermore, if $\mx \in \mathbb{X}^{-}$, by the definition $\mathbb{X}^{-}$, $\mx \not\in \mathbb{X}_4$. Thus, $\mx$ is a subtree of $\msttilde_{i}$ by Item (1) of \Cref{lm:Adm-Xprime}. That implies $\me(\mx)\cap \me_i =  \emptyset$,  which implies Item (5) of \Cref{lm:Clustering}. \qed
\end{proof}

Our next goal is to show Item (3) of \Cref{lm:Clustering}. \Cref{lm:manynodes} below  implies that if $\mx \in \mathbb{X}$ is formed in Steps 2,4, and 6, then $\Delta^+_{i+1}(\mx) = \Omega(\eps L_i |\mv(\mx)|)$.

\begin{lemma}\label{lm:manynodes} For any subgraph $\mx \in \mathbb{X}$ such that $|\mv(\mx)|\geq \frac{2g}{\zeta\eps}$ or $\Delta^+_{i+1}(\mx) = \Omega(L_i)$, then $\Delta^+_{i+1}(\mx) = \Omega(\eps L_i |\mv(\mx)|)$.
\end{lemma}
\begin{proof} We fist consider the case where $|\mv(\mx)|\geq \frac{2g}{\zeta\eps}$.	By definition of corrected potential change in Item (3) of \Cref{lm:Clustering}, we have:
	\begin{equation*}
		\begin{split}
			\Delta^+_{i+1}(\mx) &\geq 	\Delta_{i+1}(\mx) = \sum_{\varphi \in \mv(\mx)}\omega(\varphi) -  \adm(\mx) \qquad \mbox{(by \Cref{eq:LocalPotential})} \\
			&\geq (\zeta \eps L_i |\mv(\mx)|)  - \adm(\mx) \qquad \mbox{(by  property \hyperlink{P3'}{(P3')})} \\
			&\geq  (\zeta \eps L_i |\mv(\mx)|)/2   - gL_i + (\zeta \eps L_i |\mv(\mx)|)/2   \qquad \mbox{(by  property \hyperlink{P3'}{(P3')})} \\
			&\geq \frac{\zeta \eps L_i}{2}\cdot \frac{2g}{\zeta \eps }  - gL_i + (\zeta \eps L_i |\mv(\mx)|)/2  = (\zeta \eps L_i |\mv(\mx)|)/2  = \Omega(\eps L_i |\mv(\mx)|)~.
		  	\end{split}
	\end{equation*}
Next, we consider the case where $\Delta^+_{i+1}(\mx) = \Omega(L_i)$. If $|\mv(\mx)|\geq \frac{2g}{\zeta \eps}$, then  $\Delta^+_{i+1}(\mx) = \Omega(\eps L_i |\mv(\mx)|)$ as we have just shown. Otherwise, we have:
\begin{equation*}
	\Delta^+_{i+1}(\mx) = \Omega(L_i) = \Omega(\eps L_i \frac{2g}{\zeta \eps}) = \Omega(\epsilon L_i |\mv(\mx)|),
\end{equation*}
as claimed.	\qed
\end{proof}

\begin{lemma}\label{lm:Item3Clustering}   $\Delta_{i+1}^+(\mx) \geq 0$ for every $\mx \in \mathbb{X}$ and 
	\begin{equation*}
		\sum_{\mx \in \mathbb{X}^{+}} \Delta_{i+1}^+(\mx) = \sum_{\mx \in\mathbb{X}^{+}} \Omega(|\mv(\mx)|\eps L_i). 
	\end{equation*}
Consequently, Item (3) of \Cref{lm:Clustering} holds.
\end{lemma}
\begin{proof}
If $\mx \in \mathbb{X}_6$ then $|\mv(\mx)|\geq \frac{2g}{\zeta \eps}$ by the definition of heavy nodes. If $\mx \in \mathbb{X}^{-}_2\cup \mathbb{X}^{-}_4$, then $\Delta^+_{i+1}(\mx) = \Omega(L_i)$. Thus, by \Cref{lm:manynodes} for every $\mx \in \mathbb{X}^{-}_2\cup \mathbb{X}^{-}_4 \cup \mathbb{X}_6$, it holds that 	
	\begin{equation} \label{eq:Delta-246}
		\Delta^+_{i+1}(\mx) = \Omega(\epsilon L_i |\mv(\mx)|)~,
	\end{equation}
which also implies that  $	\Delta^+_{i+1}(\mx) \geq 0$. 

If $\mx \in \mathbb{X}\setminus (\mathbb{X}^{-}_2\cup \mathbb{X}^{-}_4 \cup \mathbb{X}_6)$, then $\mx \in \mathbb{X}_5^{\prefix-}\cup \mathbb{X}_5^{\internal-}$. Thus, $\mx \subseteq \msttilde_{i}$, and hence by the definition, we have that:
	\begin{equation*}
	\begin{split}
		\Delta^+_{i+1}(\mx) &=  \sum_{\varphi \in \mv(\mx)}\omega(\varphi) + \sum_{\mbe \in \me(\mx)\cap \msttilde_{i}}\omega(\mbe) -  \adm(\mx) \geq 0~.
	\end{split}
\end{equation*}
In all cases, we have $\Delta^+_{i+1}(\mx)\geq 0$.

By the definition of $\mathbb{X}^{+}$ in \Cref{eq:MathbbX}, the only case where $\Delta^+_{i+1}(\mx)$ could be $0$ is $\mx \in  \mathbb{X}_5^{\prefix-}$.  Next, we  use an averaging argument to assign potential change to $\mx$. Observe that $\mx$ is an uncontraction of some prefix $\overline{Q}$ of some path $\Pbar \in \Fbar^{(5)}$. By \Cref{lm:Clustering-Step5}, the uncontraction of the other suffix  $\Qbar'$ of $\Pbar$, say $\Qtilde'$, is augmented to a subgraph in $\mathbb{X}_1\cup  \mathbb{X}_2\cup  \mathbb{X}_4$. It follows that $\Qtilde'$ is a subgraph of some graph $\my \in \mathbb{X}^{-}_2\cup \mathbb{X}^{-}_4 \cup \mathbb{X}_6$.  If we distribute the corrected potential change $\Delta^+_{i+1}(\my)$ to nodes in $\my$, each node gets $\Omega(\eps L_i)$ potential change. Thus, the total potential change of nodes in $\Qtilde'$  is $\Omega(\eps L_i|\mv(\Qtilde')|)$. By Item (3) of \Cref{lm:Clustering-Step5}, $|\mv(\Qtilde')| = \Omega(1/\eps)$. Thus the potential change of nodes in $\Qtilde'$  is $\Omega( L_i|)$. We distribute \emph{half} of the potential change to $\mx$. Thus,  $\mx$ has $\Omega(L_i)$ potential change, and by \Cref{lm:manynodes}, the potential change of $\mx$ is $\Omega(\eps L_i |\mv(\mx)|)$. This, with \Cref{eq:Delta-246}, implies that:
	\begin{equation*}
	\sum_{\mx \in \mathbb{X}^{+}} \Delta_{i+1}^+(\mx) = \sum_{\mx \in\mathbb{X}^{+}} \Omega(|\mv(\mx)|\eps L_i), 
\end{equation*}
as desired.\qed
\end{proof}

\subsubsection{Constructing the partition of of $\me_i$: Proof of \Cref{lm:Clustering}}\label{subsec:E-T2}

In this section, we construct a partition of $\me$ and prove \Cref{lm:Clustering}. Items (3) and (5) of \Cref{lm:Clustering} were proved in \Cref{lm:Item3Clustering} and \Cref{lm:XProp}, respectively. In the following, we prove Items (1), (2) and (4). Indeed, Item (2) follows directly from the construction (\Cref{obs:Item2Clustering}). Item (1) is proved in \Cref{lm:Item1Clustering} and Item (4) is proved in \Cref{lm:Item4-Nonde} and \Cref{lm:degenerate}. 

Recall that we define $\mathbb{X}' = \mathbb{X}_1\cup \mathbb{X}_2\cup \mathbb{X}_4\cup \mathbb{X}^{\prefix}_5 \cup \mathbb{X}^{\internal}_5$ in \Cref{lm:Adm-Xprime}. We say that a subgraph $\mx\in \mathbb{X}'$ is \emph{light} if it corresponds to a light supernode in $\doverline{\mk_i}$ (defined in Step 6); otherwise, we say that $\mx$ is \emph{heavy}.

We construct $\me_i^{\take}$ and $\me^{\redunt}_i$ in two steps below;  $\me_i^{\reduce} = \me_i \setminus (\me_i^{\take}\cup \me_i^{\redunt})$. Initially, both sets are empty.

\begin{tcolorbox}
	\hypertarget{EiPartition}{}
	\textbf{Constructing $\me_i^{\take}$ and $\me^{\redunt}_i$:} Let $\mathbb{X}^{\light}$ be the set of light subgraphs in $\mathbb{X}'$.
	\begin{itemize}
		\item \textbf{Step 1:} For each subgraph $\mx\in \mathbb{X}$,  we add all edges of $\me_i$ in $\mx$ to $\me_i^{\take}$. That is, $$\me_i^{\take} \leftarrow \me_i^{\take} \cup (\me_i\cap \me(\mx)).$$
		\item \textbf{Step 2:} We construct a graph $\mh_i = (\mv_i, \msttilde_{i}\cup \me_i^{\take}, \omega)$. We then consider every edge $\mbe = (\nu\cup \mu) \in \me_i$, where both endpoints are in subgraphs in  $\mathbb{X}^{\light}$, in the non-decreasing order of the weight. If:
		\begin{equation}\label{eq:greedy-Hi}
			d_{\mh_i }(\nu,\mu) > 2 \omega(\mbe)~,
		\end{equation}
		then we add $\mbe$ to $\me_i^{\take}$ (and hence, also to $\mh_i$). Otherwise, we add $\mbe$ to $\me_i^{\redunt}$. Note that the distance in $\mh_i$ in \Cref{eq:greedy-Hi} is the augmented distance. 
	\end{itemize}
\end{tcolorbox}

The construction in Step 2 is the $\pathg$ algorithm. We observe that:

\begin{observation}\label{obs:Ereduce} For every edge $\mbe \in \me^{\reduce}_i$, at least one endpoint of $\mbe$ is in a heavy subgraph.
\end{observation}

\begin{observation}\label{obs:Item2Clustering}  Let $H_{< L_i}^{-}$ be a subgraph obtained by adding corresponding edges of $\me_i^{\take}$ to $H_{< L_{i-1}}$.  Then for every edge $(u,v)$ that corresponds to an edge in $\me^{\redunt}$, $d_{H_{< L_i}^{-}}(u,v)\leq 2d_G(u,v)$. 
\end{observation}
\begin{proof}
	The observation follows directly from the construction in Step 2.\qed 
\end{proof}
We now focus on proving Item (1) of \Cref{lm:Clustering}. The key idea is the following lemma.

\begin{lemma}\label{lm:partitionX} Any subgraph $\mx \in \mathbb{X}'\setminus \mathbb{X}_1$ can be partitioned into $ k = O(1/\zeta)$ subgraphs $\{\my_1,\ldots, \my_k\}$ such that $\adm(\my_j)\leq 9 \zeta L_i$ for any $1\leq j\leq k$ when $\eps \leq \frac{\zeta}{g}$.
\end{lemma}
\begin{proof}
	Let $\varphi$ be a branching node in $\Ftilde^{(2)}$, the tree in \Cref{lm:Clustering-Step2T2}. We say that $\varphi$ is \emph{special} if there exists three internally node disjoint paths $\Ptilde_1,\Ptilde_2,\Ptilde_3$ of  $\Ftilde^{(2)}$ sharing the same node $\varphi$ such that $\adm(\Ptilde_j \setminus \{\varphi\})\geq \zeta L_i$.
	\begin{claim}\label{clm:special}
		Any special node $\varphi$ of $\Ftilde^{(2)}$ is contained in a subgraph  in $\mathbb{X}_2$.
	\end{claim}
	\begin{proof} Let $\Ttilde$ be the tree of  $\Ftilde^{(2)}$ containing $\varphi$. 	Recall that in \Cref{lm:Clustering-Step2T2}, we apply tree clustering in \Cref{lm:tree-clustering} to $\Ttilde$ to obtain a set of subtrees $\mathbb{U}$. Let $\Tbar$ be obtained from $\Ttilde$ by contracting every tree in $\mathbb{U}$ into a single node. We show that $\varphi$ must be in a tree  $\tilde{A} \in \mathbb{U}$ such that $\adm(\tilde{A})\geq \zeta L_i$. This will imply the claim, since by the definition of $\mathbb{X}_2$,  $\tilde{A}\in \mathbb{X}_2$.

		To show that $\adm(\tilde{A})\geq\zeta L_i$, we show that $\tilde{A}$  corresponds to a $\Tbar$-branching node. Since $\adm(\Ptilde_j \setminus \{\varphi\})\geq \zeta L_i$ for every $j \in \{1,2,3\}$,  each path $\Ptilde_j$ must contain at least one node of a different subtree in $\mathbb{U}$. But this means, the contracted node corresponding to $\tilde{A}$ will be a $\Ttilde$-branching node, as claimed. By Item (4) of \Cref{lm:tree-clustering},  $\adm(\tilde{A})\geq \gamma L_i$ and since  $\gamma = \zeta$ (see \Cref{lm:Clustering-Step2T2}), we have that $\adm(\tilde{A})\geq\zeta L_i$. \qed
	\end{proof}
	
	By \Cref{lm:Adm-Xprime}, $\mx$ is a tree. Let $\mx'$ be a maximal subtree of $\mx$ such that $\mx'$ is a subtree of $\msttilde_{i}$. If $\mx$ is in $\mathbb{X}_2 \cup \mathbb{X}_5^{\prefix}\cup \mathbb{X}_5^{\internal}$ then $\mx' = \mx$. Otherwise, $\mx \in \mathbb{X}_4$, and thus it has a single edge in $\me_i$ by Item (1) of \Cref{lm:Clustering-Step4}. That is, $\mx$ has exactly two such maximal subtrees $\mx'$. Thus, to complete the lemma, we show that $\mx'$ can be partitioned into $O(1/\zeta)$ subtrees as claimed in the lemma. 
	
	Let $\md$ be the path in $\mx'$ of maximum augmented diameter. Let $\mathcal{J}$ be the forest obtained from  $\mx'$ by removing nodes of $\md$. 
	
	\begin{claim}\label{clm:diameter-J} $\adm(\mt) \leq 2\zeta L_i \quad \forall \mbox{ tree } \mt \in \mathcal{J} $
	\end{claim}
	\begin{proof}
		Let $\mu$ be the node in $\mt$ that is incident to a node, say $\varphi$, in $\md$. Then, for any node $\nu \in \mt$, $\adm(\mt[\mu,\nu])$ must be at most $\zeta L_i$, since otherwise, there are three internally node disjoint paths $\mp_1,\mp_2,\mp_3$  sharing $\varphi$ as an endpoint, two of them are paths in $\md$, such that $\adm(\mp_j\setminus \{\varphi\})\geq \zeta L_i$. That is $\varphi$ is a special node, and hence is grouped to a subgraph in $\mathbb{X}_2$ by \Cref{clm:special}; this is a contradiction. Since  $\adm(\mt[\mu,\nu]) \leq \zeta L_i$ for any $\nu\in \mt$, $\adm(\mt)\leq 2L_i$.\qed
	\end{proof}
	
	Now we greedily partition $\md$ into $k = O(1/\zeta)$ subpaths $\{\mp_1,\ldots, \mp_k\}$, each of augmented diameter at least $\zeta L_i$ and at most $3\zeta L_i$. This is possible because each node/edge has a weight at most $\max\{g\eps L_i,\bar{w}\} \leq \max\{g\eps L_i,\eps L_i\} \leq \zeta L_i$ when $\eps \leq \zeta/g$. Next, for every tree $\mt \in \mathcal{J}$, if $\mt$ is connected to a node $\varphi \in \mp_j$ via some $\msttilde_{i}$ edge $\mbe$ for some $j \in [1,k]$, we augment $\mbe$ and $\mt$ to $\mp_j$.   By  \Cref{clm:diameter-J}, the augmentation increases the diameter of $\mp$ by at most $2(\bar{w} + 2\zeta L_i)\leq 6\zeta L_i$ additively; this implies the lemma.\qed
\end{proof}

\begin{lemma}\label{lm:Const-Edge}Let $\mx, \my$ be two (not necessarily distinct) subgraphs in $ \mathbb{X}^{\light}$. Then there are $O(1)$ edges  in $\me_i^{\take}$ between nodes in $\mx$ and nodes in $\my$.
\end{lemma}
\begin{proof} Let  $\{\ma_1,\ldots, \ma_x\}$ ($\{\mb_1,\ldots, \mb_{y}\}$) be a partition of $\mx$ ($\my$) into $x = O(1/\zeta)$ ($y = O(1/\zeta)$) subgraphs of augmented diameter at most $9\zeta L_i$ as guarantee by \Cref{lm:partitionX}. 
	
	We claim that there is at most one edge in $\me^{\take}$ between $\ma_j$ and $\mb_k$ for any $1\leq j\leq x, 1\leq k \leq y$. Suppose otherwise, then let $(\nu,\mu)$ and $(\nu',\mu')$ be two such edges, where $\{\nu,\nu'\}\subseteq \mv(\ma_j)$ and  $\{\mu,\mu'\} \subseteq \mv(\mb_k)$. W.l.o.g, we assume that $\omega(\nu,\mu) \leq \omega(\nu',\mu')$. Note that $\omega(\nu',\mu')\geq L_i/(1+\psi)\geq L_i/2$. When $(\nu',\mu')$ is considered in Step 2, by the triangle inequality, 
	\begin{equation*}
		\begin{split}
					d_{\mh_i }(\nu',\mu') &\leq \adm(\ma_j) + \omega(\nu,\mu) + \adm(\mb_k) \leq 18\zeta L_i + \omega(\nu',\mu')\\
					&\leq (1+36\zeta)\omega(\nu',\mu') \qquad \mbox{(since $\omega(\nu',\mu')\geq L_i/2$)} \\
					& < 2 \omega(\nu',\mu') \qquad\mbox{(since $\zeta = 1/250$)}, 
		\end{split}
	\end{equation*}
	which contradicts \Cref{eq:greedy-Hi}. 
	
	Since there is at most one edge in $\me^{\take}$ between $\ma_j$ and $\mb_k$, the number of edges in $\me^{\take}$ between $\mx$ and $\my$ is at most $x\cdot y = O(1/\zeta^2) = O(1)$.\qed 
\end{proof}

We obtain the following corollary of \Cref{lm:Const-Edge}.

\begin{corollary}\label{cor:bounded-DegXprime}For any subgraph $\mx \in  \mathbb{X}^{\light}$,  $\deg_{\mg^{\take}_i}(\mv(\mx)) = O(1/\eps) = O(|\mv(\mx)|)$ where $\mg^{\take}_i = (\mv_i,\me_i^{\take})$. 
\end{corollary}
\begin{proof} Let  $\doverline{\varphi}_{\mx}$  be the corresponding supernode of $\mx$ in $\doverline{\mk_i}$. Since $\doverline{\varphi}_{\mx}$ is a light supernode, it has at most $\frac{2g}{\zeta \eps} = O(1/\eps)$ neighbors in $\doverline{\mk}_i$. That means there are $O(1/\eps)$ subgraphs in  $ \mathbb{X}'\setminus \mathbb{X}_1$ to which  $\mx$ has edges. By \Cref{lm:Const-Edge}, there are $O(1)$ edges for each such subgraph. Thus,  $\deg_{\mg^{\take}_i}(\mv(\mx)) = O(1/\eps) =O(|\mv(\mx)|)$ since $|\mv(\mx)| = \Omega(1/\eps)$ by \Cref{lm:Adm-Xprime}.\qed
\end{proof}

We now prove Item (1) of \Cref{lm:Clustering}.

\begin{lemma}\label{lm:Item1Clustering}For every subgraph $\mx \in \mathbb{X}$,  $\deg_{\mg^{\take}_i}(\mx) = O(|\mv(\mx)|)$ where $\mg^{\take}_i = (\mv_i,\me_i^{\take})$, and $\me(\mx)\cap \me_i \subseteq \me^{\take}$.  Furthermore, if $\mx \in \mathbb{X}^{-}$, there is no edge in $\me_i^{\reduce}$ incident to a node in $\mx$.
\end{lemma}
\begin{proof} Let $\mx$ be a subgraph in $\mathbb{X}$. Observe by the construction of $\me^{\take}_i$ in Step 1, $\me\cap \me(\mx)\subseteq \me^{\take}_i$.  Clearly, the number of edges incident to nodes in $\mx$ added in Step 1 is $O(|\mv(\mx)|)$ since  every subgraph in $\mathbb{X}$ is a tree by \Cref{lm:Adm-Xprime}.  Thus, it remains to bound the number of edges added in Step 2.	
	
	If $\mx \in \mathbb{X}^{-}_2\cup \mathbb{X}^{-}_4 \cup \mathbb{X}^{\prefix-}_5\cup \mathbb{X}^{\internal-}_5$, then $\mx$ corresponds to a light supernode in $\mk_i$. Thus, $\deg_{\mg^{\take}_i}(\mv(\mx)) = O(|\mv(\mx)|)$  by \Cref{cor:bounded-DegXprime}. Otherwise, $\mx \in \mathbb{X}_6$. By construction in Step 6, $\mx$ is the union  heavy subgraphs and light subgraphs  (and some edges in $\me_i$). By construction of $\me_i^{\take}$, only light subgraphs have nodes incident to edges in $\me_i^{\take}$. Let $\{\my_1,\ldots, \my_p\}$ be the set of light subgraphs constituting $\mx$. Then, by \Cref{cor:bounded-DegXprime}, we have that:
	\begin{equation*}
		\deg_{\mg^{\take}_i}(\mx)  \leq \sum_{k=1}^{p}\deg_{\mg^{\take}_i}(\my) =  \sum_{k=1}^{p} (|\mv(\my_k)|) = O(|\mv(\mx)|)~.
	\end{equation*}

We now show that there is no edge in $\me_i^{\reduce}$ incident to a node in $\mx \in \mathbb{X}^-$. Suppose otherwise, let $\mbe$ be such an edge.  
By \Cref{obs:Ereduce}, $\mbe$ is incident to a node in a heavy subgraph, say $\my$. That is, $\doverline{\varphi}_{\my}\in \doverline{\mv}^{\high}_i$.  By the construction in Step 6, $\doverline{\varphi}_{\mx} \in \doverline{\mv}^{\highp}_i$ and hence $\mx$ is grouped to a larger subgraph in $\mathbb{X}_6$, contradicting that  $\mx \in \mathbb{X}^-$. 
 \qed
\end{proof}

We now focus on proving Item (4) of \Cref{lm:Clustering}. In \Cref{lm:Item4-Nonde}, we consider the non-degenerate case, and in \Cref{lm:degenerate} we consider the degenerate case. 

\begin{lemma}\label{lm:Item4-Nonde}Let $(\varphi_1,\varphi_2)$ be any edge in $\me_i$ between nodes of two light subgraphs $\mx ,\my$ in $\mathbb{X}_5^{\internal}$. Then,  $(\varphi_1,\varphi_2) \in \me_i^{\redunt}$.
\end{lemma}
\begin{proof} By the construction of Step 5, $\mx$ and $\my$ correspond to two subpaths $\bar{{\mx}}$ and $\bar{\my}$ of two paths $\Pbar$ and $\Qbar$ in $\Fbar^{(5)}$. Note that all nodes in $\bar{{\mx}}$ and $\bar{\my}$  have a blue color since the suffix/prefix of  $\Pbar$ and $\Qbar$ are either in  $\mathbb{X}_5^{\prefix}$ or are augmented to existing subgraphs in Step 5B.

	Since there is an edge in $\me_i$ between $\mx$ and $\my$, there must be an edge in $\bar{\me}_i$, say $(\bar{\mu},\bar{\nu})$ between a node of $\bar{\mu} \in \bar{{\mx}}$ and a node of $\bar{\nu} \in \bar{\my}$ by the definition of $\bar{\me}_i$ (in \Cref{eq:Ebar-i}) such that $\varphi_1 \in \bmu, \varphi_2 \in \bnu$.
	
	As $\bar{\mu}$ and $\bar{\nu}$ both have a blue color, either $(\bar{\mu},\bar{\nu}) \in \me_i^{far}(\Fbar^{(5)})$ or $(\bar{\mu},\bar{\nu}) \in \me_i^{close}(\Fbar^{(5)})$ by the definition in \Cref{eq:Ebar-farclose}. By \Cref{lm:Clustering-Step4}, $\me_i^{far}(\Fbar^{(5)}) = \emptyset$. Thus,  $(\bar{\mu},\bar{\nu}) \in \me_i^{close}(\Fbar^{(5)})$. This implies $\Ibar(\bnu)\cap \Ibar(\bmu)\not= \emptyset$, and hence, $\bar{{\mx}}$ and $\bar{\my}$ are broken from the same path, say $\bar{P} \in \Fbar^{(5)}$, in Step 5B.  
	
	Furthermore,  by the definition of $\Ibar(\bnu)$, every node $\bar{\varphi} \in \Ibar(\bnu)$ is within an augmented distance (along $\Pbar$) of at most $(1-\psi)L_i$ from $\bnu$. This means, $\adm(\bar{P}[\bnu,\bmu]) \leq 2(1-\psi)L_i$. Note that the uncontraction of $\bar{P}[\bnu,\bmu]$ is a subtree of $\msttilde_{i}$. Thus, $d_{\msttilde_{i}}(\varphi_1,\varphi_2) \leq \adm(\bar{P}[\bnu,\bmu]) \leq 2(1-\psi)L_i \leq \frac{2L_i}{1+\psi} \leq 2 \omega(\varphi_1,\varphi_2)$. As $\msttilde_{i}$ is a subgraph of $\mh_i$, $(\varphi_1,\varphi_2)$ will be added to $\me_i^{\redunt}$ in Step 2, \Cref{eq:greedy-Hi}.	\qed
\end{proof}

\begin{lemma}[Structure of Degenerate Case]\label{lm:degenerate}
	If the degenerate case happens, then
	$\Fbar^{(5)}_i = \Fbar^{(4)}_i = \Fbar^{(3)}_i$, and $\Fbar^{(5)}_i$  is a single (long) path. Moreover, $|\me^{\take}_i| = O(1/\epsilon)$.
\end{lemma}
\begin{proof} Recall that the degenerate case happens when $\mathbb{X}^{-}_1\cup \mathbb{X}^{-}_2\cup \mathbb{X}^{-}_4 = \mathbb{X}_6 =  \emptyset$. This implies $\mathbb{X}_1\cup \mathbb{X}_2\cup \mathbb{X}_4 = \emptyset$. Thus, $\Fbar^{(5)}_i = \Fbar^{(4)}_i = \Fbar^{(3)}_i$. Furthermore, $\Fbar^{(5)}_i$  is a single (long) path since $\Fbar^{(3)}_i$ is a path by \Cref{lm:Clustering-Step2T2}. This gives  $|\mathbb{X}_5^{\prefix}| = 2$. By \Cref{lm:Item4-Nonde}, there is no edge in $\me_i^{\take}$ between two subgraphs in $\mathbb{X}_5^{\internal}$. Thus, any edge in $\me_i^{\take}$ must be incident to a node in a subgraph of $\mx \in\mathbb{X}_5^{\prefix}$. By \Cref{cor:bounded-DegXprime}, there are $O(1/\eps)$ such edges.	\qed
\end{proof}

\section{Clustering for Stretch $t = 1+ \eps$}\label{sec:stretch1E}

In this section, we prove \Cref{lm:ConstructClusterHi} when the stretch $t = 1+\eps$. The key technical idea is the following clustering lemma, which is analogous to \Cref{lm:Clustering} in \Cref{sec:stretch2}; the highlighted texts below are the major differences.  Recall that $H_{< L_{i-1}} = \cup_{j=0}^{i-1} H_{j}$.

\begin{restatable}{lemma}{ClusteringE}
	\label{lm:ClusteringE} Let $\mg_i = (\mv_i,\me_i)$ be the cluster graph. We can construct in polynomial time  (i) a collection $\mathbb{X}$ of subgraphs of $\mg_i$ and its partition into  two sets $\{\mathbb{X}^{+}, \mathbb{X}^{-}\}$ and (ii) a partition of $\me_i$ into three sets $\{\me_i^{\take}, \me_i^{\reduce}, \me_i^{\redunt}\}$ such that:
	\begin{enumerate}
		\item[(1)] For every subgraph $\mx \in \mathbb{X}$,  \hl{$\deg_{\mg^{\take}_i}(\mv(\mx)) = O(|\mv(\mx)|/\eps)$}  where $\mg^{\take}_i = (\mv_i,\me_i^{\take})$, and $\me(\mx)\cap \me_i \subseteq \me^{\take}$. Furthermore, if $\mx \in \mathbb{X}^{-}$, there is no edge in $\me_i^{\reduce}$ incident to a node in $\mx$.
		
		\item[(2)] Let $H_{< L_i}^{-}$ be a subgraph obtained by adding corresponding edges of $\me_i^{\take}$ to $H_{< L_{i-1}}$.  Then for every edge $(u,v)$ that corresponds to an edge in $\me^{\redunt}$, $d_{H_{< L_i}^{-}}(u,v)\leq(1+6g\eps)2d_G(u,v)$. 
		
		\item[(3)] Let $\Delta_{i+1}^+(\mx) = \Delta(\mx) + \sum_{\mbe \in \msttilde_i\cap \me(\mx)}w(\mbe)$ be the \emph{corrected potential change} of $\mx$. Then, $\Delta_{i+1}^+(\mx) \geq 0$ for every $\mx \in \mathbb{X}$ and 
		\begin{equation}\label{eq:averagePotential-t1E}
			\sum_{\mx \in \mathbb{X}^{+}} \Delta_{i+1}^+(\mx) = \sum_{\mx \in \mathbb{X}^{+}} \Omega(|\mv(\mx)|\eps L_i). 
		\end{equation}
		\item[(4)] \hl{There exists an orientation of edges in $\me_i^{\take}$ such that for every subgraph $\mx \in \mathbb{X}^{-}$, if $\mx$ has $t$ out-going edges for some $t\geq 0$, then $\Delta^+_{i+1}(\mx) =\Omega(|\mv(\mx)|t\eps^2 L_i)$}, unless a \emph{degenerate case} happens, in which  $\me^{\reduce}_i = \emptyset$ and  
		\begin{center}
			\hl{$\omega(\me_i^{\take}) = O(\frac{1}{\eps^2})(\sum_{\mx \in \mathbb{X}} \Delta_{i+1}^+(\mx) + L_i).$}
		\end{center}
 		\item[(5)] For every subgraph $\mx \in \mathbb{X}$, $\mx$ satisfies the three properties (\hyperlink{P1'}{P1'})-(\hyperlink{P3'}{P3'}) with constant $g=31$. 
	\end{enumerate}	
\end{restatable}

The total node degree of $\mx$ in $\mg_i^{\take}$ in \Cref{lm:ClusteringE} is worst than the total node degree of $\mx$  in \Cref{lm:Clustering} by a factor of $1/\eps$.  Furthermore, Item (4) of \Cref{lm:ClusteringE} is qualitatively different from  Item (4) of \Cref{lm:Clustering} and we no longer can bound the size of $\me_i^{\take}$ in the degenerate case. All of these are due to the fact that the stretch $t = 1+\eps < 2$ when $\eps < 1$.

Next we show to construct $H_i$ given that we can construct a set of subgraphs $\mathbb{X}$ as claimed in \Cref{lm:ClusteringE}. The proof of \Cref{lm:ClusteringE} is deferred to \Cref{subsec:clusteringTE}.

\subsection{Constructing $H_i$: Proof of \Cref{lm:ConstructClusterHi} for $t = 1+\eps$.} \label{subsec:ConstructHiTE}

Let $\msttilde^{in}_i(\mx) = \me(\mx)\cap \msttilde_i$ for each $\mx \in \mathbb{X}$. Let $\msttilde^{in}_i = \cup_{\mx \in \mathbb{X}}(\me(\mx)\cap \msttilde_i)$ be the set of $\msttilde_i$ edges that are contained in subgraphs in $\mathbb{X}$.   The construction of $H_i$ is exactly the same as the construction of $H_i$ in \Cref{subsec:ConstructHiT2}: first, add every edge of $\me_i^{\take}$ to $H_i$, and then apply \hyperlink{SPHigh}{$\sso$} on the subgraph of $\mg_i$ induced by $\me_i^{\reduce}$. Furthermore, \Cref{clm:Ki-clustergraph} and \Cref{obs:supportingPropHiT2} hold here.  

We now bound the stretch of $F^\sigma_{i}$. Recall that $F^\sigma_{i}$ is the set of edges in $E^{\sigma}_i$ that correspond to $\mathcal{E}_i$.

\begin{lemma}\label{lm:Hi-StretchT1E} For every edge $(u,v) \in F^\sigma_{i}$, $d_{H_{< L_i}}(u,v) \leq t(1+ \max\{s_{\sso}(2g),6g\}\eps)w(u,v)$.
\end{lemma}
\begin{proof}
	By construction, edges in $F^{\sigma}_i$ that correspond to $\me_i^{\take}$ are added to $H_i$ and hence have stretch $1$. By Item (2) of \Cref{lm:ClusteringE}, edges in $F^{\sigma}_i$ that correspond to $\me_i^{\redunt}$ have stretch $ (1+6g\eps)\leq t(1+6g\eps)$ in $H_{< L_i}$ since $t\geq 1$. Thus, it remains to focus on edges corresponding to $\me_i^{\reduce}$. Let $(\varphi_{C_u},\varphi_{C_v}) \in \me^{\reduce}_i$ be the edge corresponding to an edge $(u,v \in F^{\sigma}_i$.  	Since we add all edges of $F$ to $H_i$, by property (2) of \hyperlink{SPHigh}{$\sso$}, the stretch of edge $(u,v)$ in $H_{< L_i}$ is at most $t(1+s_{\sso}(\beta)\eps) = t(1+s_{\sso}(2g)\eps)$ since $\beta  = 2g$ by \Cref{clm:Ki-clustergraph}.\qed
\end{proof}

Next, we bound the total weight of $H_i$.

\begin{lemma}\label{lm:Hi-WeightTE}  $w(H_i) \leq \lambda \Delta_{i+1} + a_i$ for $\lambda = O(\chi \eps^{-1}  + \eps^{-2})$  and $a_i =   O(\chi \eps^{-1} +\eps^{-2})w(\msttilde^{in}_i) + O(L_i/\eps^2)$.  
\end{lemma}
\begin{proof}  First, we consider the non-degenerate case. Note that edges in $\me^{\redunt}_i$ are not added to $H_i$. Let $\mv_i^{+} = \cup_{\mx \in \mathbb{X}^{+}}\mv(\mx)$ and  $\mv_i^{-} = \cup_{\mx \in \mathbb{X}^{-}}\mv(\mx)$. 
	Let $F^{(a)}_i$ be the set of edges added to $H_i$ in the construction in Step $a$, $a\in \{1,2\}$. 
	
By the construction in Step 1, $F^{(1)}_i$ includes edges in $\me_i^{\take}$. 	Let $A^{(1)}\subseteq F^{(1)}_i$  be the set of edges incident to at least one node in $\mv_i^{+}$ and $A^{(2)}=F^{(1)}_i\setminus A^{(1)}$.  By Item (1) in \Cref{lm:ClusteringE}, the total weight of the edges added to $H_i$ in Step 1 is:
	\begin{equation}\label{eq:A1TE}
		\begin{split}
			w(A^{(1)}_i)  &=  \sum_{\mx \in \mathbb{X}^{+}} O(|\mv(\mx)|/\eps) L_i \stackrel{\mbox{\tiny{\cref{eq:averagePotential-t1E}}}}{=}  O(\frac{1}{\eps^2})\sum_{\mx \in \mathbb{X}^{+}} \Delta^+_{i+1}(\mx)\\
			&= O(\frac{1}{\eps^2})\sum_{\mx \in\mathbb{X}} \Delta^+_{i+1}(\mx)  \qquad \mbox{(since $\Delta^+_{i+1}(\mx)\geq 0 \quad \forall \mx \in \mathbb{X}$  by Item (3) in \Cref{lm:ClusteringE})} \\
			&= O(\frac{1}{\eps^2})(\Delta_{i+1} + w(\msttilde^{in}_i)) \qquad \mbox{(by Item (1) of \Cref{obs:supportingPropHiT2})}~.
		\end{split}
	\end{equation}   
By definition 	$A^{(2)}$ is the set of edges with both endpoints in subgraphs of $\mv_i^{-}$. Consider the orientation of $\me^{\take}_i$ as in \Cref{lm:ClusteringE}. Then, every edge of $A^{(2)}$ is an out-going edge from some node in a  graph in $\mathbb{X}^{-}$. For each graph $\mx \in \mathbb{X}^{-}$, by Item (4) of \Cref{lm:ClusteringE}, the total weight of incoming edges of $\mx$ is $O(t L_i) = O(1/\eps^2)\Delta^{+}_{i+1}(\mx)$. Thus,  we have:
\begin{equation}\label{eq:A2TE}
	\begin{split}
		w(A^{(2)}_i)  &= O(\frac{1}{\eps^2})\sum_{\mx \in\mathbb{X}} \Delta^+_{i+1}(\mx)  \qquad \mbox{(since $\Delta^+_{i+1}(\mx)\geq 0 \quad \forall \mx \in \mathbb{X}$  by Item (3) in \Cref{lm:ClusteringE})} \\
		&= O(\frac{1}{\eps^2})(\Delta_{i+1} + w(\msttilde^{in}_i)) \qquad \mbox{(by Item (1) of \Cref{obs:supportingPropHiT2})}~.
	\end{split}
\end{equation}  
 Thus, by \Cref{eq:A1TE,eq:A2TE}, we have $w(F^{(1)}) = O(\frac{1}{\eps^2})(\Delta_{i+1} + w(\msttilde^{in}_i))$. By the exactly the same argument in \Cref{lm:Hi-WeightTE}, we have that $w(F^{(2)}) =O(\chi/\eps)(\Delta_{i+1} + w(\msttilde^{in}_i)) $. This gives: 
\begin{equation}\label{eq:Hi-nondegenT1E}
		\begin{split}
			w(H_i) &=  O(\chi/\eps + 1/\eps^2) (\Delta_{i+1} + w(\msttilde^{in}_i)) \leq \lambda(\Delta_{i+1} + w(\msttilde^{in}_i))
		\end{split}
	\end{equation} 
	for some $\lambda =  O(\chi/\eps + 1/\eps^2) $.

	It remains to consider the degenerate case, and in which case, we only add to $H_i$ edges corresponding to $\me^{\take}_i$. Thus, by Item (4) of \Cref{lm:ClusteringE}, we have:
	\begin{equation}\label{eq:Hi-degenT1E}
		w(H_i) = O(\frac{L_i}{\eps^2}) \leq  \lambda\cdot (\Delta_{i+1} + w(\msttilde^{in}_i)) + O(\frac{L_i}{\eps^2}), 
	\end{equation} 
	since $\Delta_{i+1} + w(\msttilde^{in}_i) = \sum_{\mx\in \mathbb{X}}\Delta^+_{i+1}(\mx)$ by Item (1) in \Cref{obs:supportingPropHiT2}. Thus, the lemma follows from \Cref{eq:Hi-degenT1E,eq:Hi-nondegenT1E}. \qed
\end{proof}

We are now ready to prove \Cref{lm:ConstructClusterHi} for the case $t = 1+\eps$, which we restate below.

\HiConstruction*
\begin{proof} The fact that subgraphs in $\mathbb{X}$ satisfy the three properties (\hyperlink{P1'}{P1'})-(\hyperlink{P3'}{P3'}) with constant $g=223$ follows from Item (5) of \Cref{lm:ClusteringE}.  The stretch in $H_{< L_i}$ of edges in $F^{\sigma}_{i}$ follows from \Cref{lm:Hi-StretchT1E}.

	By  \Cref{lm:Hi-WeightTE}, $w(H_i) \leq \lambda \Delta_{i+1} + a_i$ where  $\lambda = O(\chi \eps^{-1} +\eps^{-2})$  and $a_i =   O(\chi \eps^{-1} + \eps^{-2})w(\msttilde^{in}_i) + O(L_i/\eps^2)$. It remains to show that $A = \sum_{i\in \mathbb{N}^+}a_i = O(\chi \eps^{-1} + \eps^{-2})$.  Observe that
	\begin{equation*}
		\sum_{i\in \mathbb{N}^+}O(\frac{L_i}{\eps^2}) ~=~  O(\frac{1}{\epsilon^2}) \sum_{i=1}^{i_{\max}} \frac{L_{i_{\max}}}{\epsilon^{i_{\max}-i}} ~=~ O(\frac{L_{i_{\max}}}{\epsilon^2(1-\epsilon)}) ~=~ O(\frac{1}{\epsilon^2}) w(\mst)~;
	\end{equation*}
	here $i_{\max}$ is the maximum level. The last equation is due to that $\eps \leq 1/2$  and every edge has weight at most $w(\mst)$ since the weight of every is the shortest distance between its endpoints. By Item (2) of \Cref{obs:supportingPropHiT2},  $\sum_{i\in \mathbb{N}^+} \msttilde^{in}_i\leq w(\mst)$.  Thus, $A = O(\chi/\eps^2) + O(1/\eps^2)$ as desired.   \qed
\end{proof}

\subsection{Clustering} \label{subsec:clusteringTE}

In this section, we prove \Cref{lm:ClusteringE}. The construction of $\mathbb{X}$ has 5 steps. The first four steps are exactly the same as the first four steps in the construction in \Cref{sec:stretch2}. In Step 5, we construct $\mathbb{X}^{\internal}_5$ differently, taking into account of edges in $\bar{\me}_i^{close}$ in \Cref{eq:Ebar-farclose}. Recall that when the stretch parameter $t\geq 2$, we show that edges in $\me_i$ corresponding to $\bar{\me}_i^{close}$ are added to $\me_i^{\redunt}$ (implicitly in \Cref{lm:Item4-Nonde}). However, when $t = 1+\eps$, we could not afford to do so, and the construction in Step 5 will take care of these edges.

\paragraph{Steps 1-4.~} The construction of Steps 1 to 4 are exactly the same as Steps 1-4 in \Cref{subsec:clusteringT2} to obtain three sets of clusters $\mathbb{X}_1,\mathbb{X}_2$ and $\mathbb{X}_4$ whose properties are described in \Cref{lm:Clustering-Step1T2,lm:Clustering-Step2T2,lm:Clustering-Step3,lm:Clustering-Step4}. After the four steps, we obtain the forest $\Fbar^{(5)}_i$, where every tree is a path. In particular, \Cref{remark:Clustering-Step4} applies: for every edge $(\bmu,\bnu)\in \bar{\me}_i$ with both endpoints in $\Fbar^{(5)}_i$, either (i) the edge is in $ \bar{\me}^{close}_i(\Fbar^{(5)}_i)$, or (ii) at least one of the endpoints must belong to a low-diameter tree of $\Fbar^{(5)}_i$ or (iii) in a (red) suffix of a long path in $\Fbar^{(5)}_i$  of augmented diameter at most $L_i$.

Before moving on to Step 5, we need a preprocessing step in which we find all edges in $\me^{\redunt}_i$. The construction of Step 5 relies on edges that are not in  $\me^{\redunt}_i$. 

\paragraph{Constructing $\me_i^{\redunt}$ and $\me_i^{\take-}$.~} Let $\Ftilde^{(5)}_i$ be obtained from $\Fbar^{(5)}_i$ by uncontracting the contracted nodes.  We apply the greedy algorithm. Initially, both $\me_i^{\redunt}$ and $\me_i^{\take-}$ are empty sets.  We construct a graph $\mh_i = (\mv_i, \msttilde_{i}\cup \me_i^{\take-}, \omega)$, which initially only include edges in $\msttilde_{i}$. We then consider every edge $\mbe = (\nu\cup \mu) \in \me_i$, where both endpoints are in $\mv(\Ftilde^{(5)}_i)$, in the non-decreasing order of the weight. If:
\begin{equation}\label{eq:greedy-HiE}
	d_{\mh_i }(\nu,\mu) \leq (1+6g\eps) \omega(\mbe)~,
\end{equation}
then we add $\mbe$ to $\me_i^{\redunt}$. Otherwise, we add $\mbe$ to $\me_i^{\take-}$ (and hence to $\mh_i$).  Note that the distance in $\mh_i$ in \Cref{eq:greedy-HiE} is the augmented distance.  We have the following observation which follows directly from the greedy algorithm.

\begin{observation}\label{obs:greedy-HiE} For every edge $\mbe = (\nu,\mu)  \in \me_i^{\take-}$, $d_{\mh_i }(\nu,\mu) \geq (1+6g\eps) \omega(\mbe)$.
\end{observation}

\paragraph{Step 5.~}  
Let $\Pbar$ be  a path in  $\Fbar^{(5)}_i$ obtained by Item (5) of \Cref{lm:Clustering-Step4}. We construct two sets of subgraphs, denoted by $\mathbb{X}^{\internal}_5$ and $\mathbb{X}^{\prefix}_5$, of $\mg_i$. The construction is broken into two steps. Step 5A is only applicable when $\mathbb{X}_1 \cup \mathbb{X}_2\cup \mathbb{X}_4 \not= \emptyset$. In Step 5B, we need a more involved construction by ~\cite{LS19}, as described in \Cref{lm:Clustering-Step5B1E}. 

\begin{itemize}
	\item (Step 5A)\hypertarget{5A}{}  If $\Pbar$ has augmented diameter at most $6L_i$, let $\mbe$ be an $\widetilde{\mst}_i$ edge connecting $\Ptilde^{\uncontract}$  and a node in some subgraph $\mx \in \mathbb{X}_1\cup \mathbb{X}_2 \cup \mathbb{X}_4$; $\mbe$ exists by \Cref{lm:Clustering-Step4}. We add both $\mbe$ and $\Ptilde^{\uncontract}$ to $\mx$.
	\item (Step 5B)\hypertarget{5B}{} 	Otherwise,  the augmented diameter of $\Pbar$ is at least $6L_i$.  Let $\{\Qbar_1, \Qbar_2\}$ be the suffix and prefix of $\Pbar$ such that $\Qbar^{\uncontract}_1$ and $\Qbar^{\uncontract}_2$ have augmented diameter at least $L_i$ and at most $2L_i$. If $\Qbar_j$, $j\in \{1,2\}$ is connected to a node in a subgraph $\mx \in \mathbb{X}_1 \cup \mathbb{X}_2\cup \mathbb{X}_4$ via an  edge $e\in \msttilde_{i}$, we add $\tilde{Q}^{\uncontract}_j$ and $e$ to $\mx$. 	If $\Qbar_j$ contains an endpoint of $\Pbar$, we add $\Qtilde_j^{\uncontract}$ to $\mathbb{X}^{\prefix}_5$. 
	
	Next, denote by $\Pbar^{'}$ the path obtained by removing $\Qbar_1, \Qbar_2$  from $\Pbar$. We then apply the construction in \Cref{lm:Clustering-Step5B1E} to $\Pbar^{'}$ to obtain a set of subgraphs $\mathbb{C}_5(\Pbar^{'})$ and an orientation of edges in  $\me^{\take-}_i(\Pbar')$, the set edges of $\me^{\take-}_i$	with both endpoints in  the uncontraction of $\Pbar^{'}$. We add all edges of  $\me^{\take-}_i(\Pbar')$ to a set $\me^{(5B)}_{i}$ (which is initially empty).  We then add all subgraphs in $\mathbb{C}_5(\Pbar^{'})$ to  $\mathbb{X}^{\internal}_5$.
\end{itemize}

The construction of Step 5B is described in the following lemma, which is a slight adaption of Lemma 6.17 in~\cite{LS19}. See \Cref{fig:5B} for an illustration. The construction crucially exploit the fact that  	$d_{\mh_i }(\nu,\mu) \leq (1+6g\eps) \omega(\mbe)$.  For completeness, we include the proof in \Cref{app:Clustering-5BProofE}.

\begin{lemma}[Step 5B]\label{lm:Clustering-Step5B1E} Let $\Pbar$ be a path in $\Fbar^{(5)}_i$.  
	Let  $\me^{\take-}_i(\Pbar)$ be the edges of $\me^{\take-}_i$
	with both endpoints in $\Ptilde^{\uncontract}$. We can construct a set of subgraphs $\mathbb{C}_5(\Pbar)$ such that: 	
	\begin{enumerate}
		\item[(1)]  Subgraphs in $\mathbb{C}_5(\Pbar)$ contain every node in $\Ptilde^{\uncontract}$.
		\item[(2)] For every subgraph $\mx\in \mathbb{C}_4(\Pbar)$, $\zeta L_i \leq \adm(\mx)\leq 5L_i$. Furthermore, $\mx$ is a subtree of $\Ptilde^{\uncontract}$ and some edges in  $\me^{\take-}_i(\Pbar)$  whose both endpoints are in $\mx$.
		\item[(3)]  There is an orientation of edges in $\me^{\take-}_i(\Pbar)$  
		such that, for any subgraph $\mx \in  \mathbb{C}_5(\Pbar)$,  if the total number of out-going edges incident to nodes in $\mx$ is $t$ for any $t\geq 0$, then:
		\begin{equation}\label{eq:Step5Bpotential}
			\Delta^+_{i+1}(\mx) =  \Omega(t\epsilon^2) L_i
		\end{equation}
	\end{enumerate}
\end{lemma}

\begin{figure}[!htb]
	\center{\includegraphics[width=0.9\textwidth]{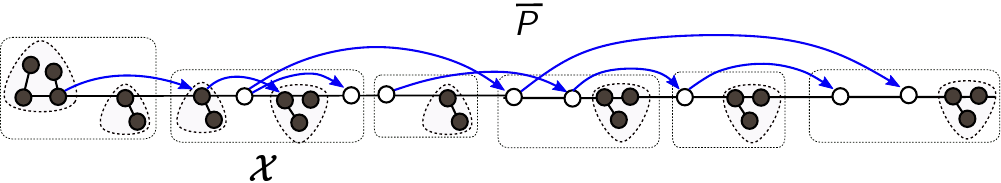}}
	\caption{A path $\Pbar$, a cluster $\mx$, and a set of (blue) edges in $\me^{\take-}_i(\Pbar)$.
		 White nodes are uncontracted nodes and black nodes are those in contracted nodes (triangular shapes). $\Delta^+_{i+1}(\mx) $ is proportional to the number of out-going edges from nodes in $\mx$, which is 3 in this case; there could be edges with both endpoints in $\mx$.}
	\label{fig:5B}
\end{figure} 

We observe the following from the construction.
\begin{observation}\label{obs:MEiTakeMinus} For every edge $\mbe \in \mathcal{E}_i^{\take-}$, either at least one endpoint of $\mbe$ is in a subgraph in $\mathbb{X}_5^{\prefix}$, or both endpoints of $\mbe$ are in $\me_i^{(5B)}$.
\end{observation}

The following lemma, which is analogous to \Cref{lm:Clustering-Step5}, describes properties of subgraphs in Step 5.

\begin{lemma}\label{lm:Clustering-Step51E}  Every subgraph $\mx \in \mathbb{X}_5^{\internal} \cup \mathbb{X}_5^{\prefix}$ satisfies:
	\begin{enumerate}[noitemsep]
		\item[(1)] $\mx$ is a subtree of $\msttilde_{i}$ if $\mx\in \mathbb{X}_5^{\prefix}$.
		\item[(2)] $\zeta L_i \leq \adm(\mx)\leq 20 L_i$.
		\item[(3)] $|\mv(\mx)| = \Omega(1/\eps)$.
	\end{enumerate}
	Furthermore, if $\mx \in \mathbb{X}_5^{\prefix}$, then $\mx$ the uncontraction of a prefix/suffix subpath $\Qbar$ of a long path $\Pbar$, and additionally, the (uncontraction of) other suffix $\Qbar'$ of   $\Pbar$ is augmented to a subgraph in $\mathbb{X}_1 \cup \mathbb{X}_2\cup \mathbb{X}_4$, unless $\mathbb{X}_1 \cup \mathbb{X}_2\cup \mathbb{X}_4 = \emptyset$.
\end{lemma}
\begin{proof} Observe that if $\mx\in \mathbb{X}_5^{\prefix}$, then $\mx$ is a subtree of  $\msttilde_{i}$. Item (2) follows from the construction and Item (2) of \Cref{lm:Clustering-Step5B1E}.    Item (3) follows  from \Cref{lm:size-MSTsubree}, and the last claim about subgraphs in  $\mathbb{X}_5^{\prefix}$ follows from \Cref{lm:Clustering-Step4}.
	\qed
\end{proof}

In the next section, we prove \Cref{lm:ClusteringE}.

\subsubsection{Constructing $\mathbb{X}$ and the partition of $\me_i$:  Proof of \Cref{lm:ClusteringE}}\label{subsec:X-T1E}
  We distinguish two cases:
 
 \paragraph{Degenerate Case.~} The degenerate case is the case where   $\mathbb{X}_1\cup \mathbb{X}_2\cup \mathbb{X}_4 =  \emptyset$. In this case, we set $\mathbb{X} = \mathbb{X}^{-} =  \mathbb{X}_5^{\internal} \cup \mathbb{X}_5^{\prefix}$, and $	\mathbb{X}^{+} = 	 \emptyset$. 
 
 \paragraph{Non-degenerate case.~} We define:
 \begin{equation}\label{eq:MathbbX1E}
 	\begin{split}
 		\mathbb{X}^{+} &=    \mathbb{X}_1\cup \mathbb{X}_2\cup \mathbb{X}_4 \cup \mathbb{X}_5^{\prefix}, \quad
 		\mathbb{X}^{-} = \mathbb{X}_5^{\internal}\\
 		\mathbb{X} &= 	\mathbb{X}^{+}\cup \mathbb{X}^{-} 
 	\end{split}
 \end{equation}

 Next, we construct the partition of $\{\me^{\take}_i, \me^{\redunt}_i, \me_i^{\reduce}\}$ of $\me_i$. Recall that we constructed two edge sets $\me_i^{\redunt}$  and $\me^{\take-}_i$ above (\Cref{eq:greedy-HiE}). We then construct $\me_i^{\take}$ as described below. It follows that $\me^{\reduce}_i = \me_i\setminus (\me_i^{\take}\cup \me_i^{\redunt})$.
 
 \begin{tcolorbox}
 	\hypertarget{EiPartition1E}{}
 	\textbf{Constructing $\me_i^{\take}$:} Let $\mv_i^{+} = \cup_{\mx\in \mathbb{X}^{+}}\mv(\mx)$ and $\mv_i^{-} = \cup_{\mx\in \mathbb{X}^{-}}\mv(\mx)$. 	First, we add all edges in  $\me^{\take-}_i$ to $\me^{\take}_i$. Next, we add $ (\cup_{\mx\in \mathbb{X}}\me(\mx)\cap \me_i)$ to $\me_i^{\take}$. Finally, for every edge $\mbe \in \me_i\setminus \me_i^{\redunt}$ such that $\mbe$ is incident to at least one node in $\mv_i^{-}$, we add $\mbe$ to   $\me_i^{\take}$.
 \end{tcolorbox}

In the analysis below, we only explicitly  distinguish the degenerate case from the non-degenerate case when it is necessary, i.e, in the proof Item (4) of \Cref{lm:ClusteringE}. Otherwise, which case we are in is either implicit from the context, or does not matter.

We observe that Item (2) in \Cref{lm:ClusteringE} follows directly from the construction of $\me_i^{\redunt}$. Henceforth, we focus on proving other items of \Cref{lm:ClusteringE}.  We first show Item (5).

\begin{lemma}\label{lm:XProp1E} Let $\mathbb{X}$ be the subgraph as defined in \Cref{eq:MathbbX1E}. For every subgraph $\mx \in \mathbb{X}$, $\mx$ satisfies the three properties (\hyperlink{P1'}{P1'})-(\hyperlink{P3'}{P3'}) with $g = 31$. Consequently, Item (5) of \Cref{lm:ClusteringE} holds.
\end{lemma}
\begin{proof} 	We observe that property \hyperlink{P1'}{(P1')} follows directly from the construction.  Property \hyperlink{P2'}{(P2')} follows directly from \Cref{lm:Clustering-Step1T2,lm:Clustering-Step2T2,lm:Clustering-Step4,lm:Clustering-Step51E}. We now bound  $\adm(\mx)$. The lower bound on $\adm(\mx)$ follows directly from Item (3) of \Cref{lm:Clustering-Step1T2}, Items (2) of \Cref{lm:Clustering-Step2T2,lm:Clustering-Step4,lm:Clustering-Step51E}. For the upper bound, by the same argument in \Cref{lm:XProp}, if $\mx$ is initially formed in Steps 1-4, then $\adm(\mx)\leq 31L_i$. Otherwise, by \Cref{lm:Clustering-Step51E}, $\adm(\mx) \leq 5L_i$, which implies  property \hyperlink{P3'}{(P3')} with $g= 31$. \qed 
\end{proof}

 We observe that \Cref{lm:manynodes} and  \Cref{lm:Item3Clustering} holds for $\mathbb{X}^{+}$, which we restate below in \Cref{lm:manynodes1E} and \Cref{lm:Item3Clustering1EHigh}, respectively. In particular,  \Cref{lm:Item3Clustering1EHigh} implies Item (3) of \Cref{lm:ClusteringE}.

\begin{lemma}\label{lm:manynodes1E} For any subgraph $\mx \in \mathbb{X}$ such that $|\mv(\mx)|\geq \frac{2g}{\zeta\eps}$ or $\Delta^+_{i+1}(\mx) = \Omega(L_i)$, then $\Delta^+_{i+1}(\mx) = \Omega(\eps L_i |\mv(\mx)|)$.
\end{lemma}

\begin{lemma}\label{lm:Item3Clustering1EHigh}   $\Delta_{i+1}^+(\mx) \geq 0$ for every $\mx \in \mathbb{X}$ and 
	\begin{equation*}
		\sum_{\mx \in \mathbb{X}^{+}} \Delta_{i+1}^+(\mx) = \sum_{\mx \in\mathbb{X}^{+}} \Omega(|\mv(\mx)|\eps L_i). 
	\end{equation*}
\end{lemma}

We now prove Item (1) of \Cref{lm:ClusteringE}, which we restate here for convenience.

\begin{lemma}\label{lm:Item1Clustering1E} For every subgraph $\mx \in \mathbb{X}$,  $\deg_{\mg^{\take}_i}(\mv(\mx)) = O(|\mv(\mx)|/\eps)$  where $\mg^{\take}_i = (\mv_i,\me_i^{\take})$, and $\me(\mx)\cap \me_i \subseteq \me^{\take}$. Furthermore, if $\mx \in \mathbb{X}^{-}$, there is no edge in $\me_i^{\reduce}$ incident to a node in $\mx$.
\end{lemma}
\begin{proof} 	Let $\mv^{\highp}_i = \cup_{\mx\in \mathbb{X}_1}\mx$. Note by the construction in Step 1 (\Cref{lm:Clustering-Step1T2}), nodes in $\mv_i\setminus \mv^{\highp}_i$ have degree $O(\frac{1}{\eps})$. Let $\me_i^{(1)}$  be the set of edges in $\me_i^{\take}$ with both endpoints in $\mv^{\highp}_i$ and $\me_i^{2} =\me_i^{\take}\setminus \me_i^{(1)}$. Also by the construction in Step 1 (\Cref{lm:Clustering-Step1T2}), both endpoints of every edge in $\me_i^{2}$ have degree $O(1/\eps)$. Thus, for any $\mx\in \mathbb{X}$, the number of edges in $\me_i^{(2)}$ incident to nodes in $\mx$ is $O(|\mv(\mx)|/\eps)$. 
	
	Next, we consider  $\me_i^{(1)}$. Observe by the construction of $\me_i^{\take}$ that there is no edge in  $\me_i^{(1)}$  with two endpoints in two \emph{different graphs} of $\mathbb{X}_1$. Furthermore, since $\mx$ is a tree for every subgraph $\mx\in \mathbb{X}_1$, the number of edges in  $\me_i^{(1)}$   incident to nodes in $\mx$ is $O(|\mv(\mx)|)$. This bound also holds for every subgraph $\mx$ not in $\mathbb{X}_1$ since the number of incident edges in $\me_i^{(1)}$ is 0; this implies the claimed bound on $\deg_{\mg^{\take}_i}(\mv(\mx))$.
	
	For the last claim, we observe that nodes in subgraphs of  $\mathbb{X}^{-}$ are in $\mv_i^{-}$. Thus, by the construction of $\me_i^{\take}$, every edge incident to a node of  $\mx \in \mathbb{X}^{-}$ is either in $\me_i^{\take}$ or $\me_i^{\redunt}$.\qed
\end{proof}

We now focus on proving Item (4) of \Cref{lm:ClusteringE} which we restate below.

\begin{lemma}\label{lm:Item4Clustering1E}  There exists an orientation of edges in $\me_i^{\take}$ such that for every subgraph $\mx \in \mathbb{X}^{-}$, if $\mx$ has $t$ out-going edges for some $t\geq 0$, then $\Delta^+_{i+1}(\mx) =\Omega(|\mv(\mx)|t\eps^2 L_i)$, unless a \emph{degenerate case} happens, in which  $\me^{\reduce}_i = \emptyset$ and  $$\omega(\me_i^{\take}) = O(\frac{1}{\eps^2})(\sum_{\mx \in \mathbb{X}} \Delta_{i+1}^+(\mx) + L_i).$$
\end{lemma}
\begin{proof} First, we consider the non-degenerate case. Recall that $\{\mv^+_i,\mv^-_i\}$ is a partition of $\mv_i$ in the \hyperlink{EiPartition1E}{construction of $\me_i^{\take}$}. We orient edges of $\me_i^{\take}$ as follows.
	
	First, for any $\mbe = (\mu,\nu)\in \me_i^{\take}$ such that at least one endpoint, say $\mu \in \mv_i^+$, we orient $\mbe$ as out-going from $\mu$. (If both $\mu,\nu$ are in $\mv_i^+$, we orient $\mbe$ arbitrarily).  Remaining  edges are subsets  of $\me_i^{(5B)}$ by \Cref{obs:MEiTakeMinus}. We orient edges in  $\me_i^{(5B)}$ as in the construction of Step 5B.	For every subgraph  $\mx \in \mathbb{X}^{-}$, by construction, out-going edges incident to nodes in $\mx$ are in $\me_i^{(5B)}$. By Item (3) of \Cref{lm:Clustering-Step5B1E}, $\Delta^+_{i+1}(\mx) =\Omega(|\mv(\mx)|t\eps^2 L_i)$.

	It remains to consider the degenerate case. In this case,  by the same argument in \Cref{lm:degenerate}, 	$\Fbar^{(5)}_i = \Fbar^{(4)}_i = \Fbar^{(3)}_i$, and $\Fbar^{(5)}_i$  is a single (long) path. Furthermore, $\me_i^{\take} = \me_i^{\take-}$, and $|\mathbb{X}_5^{\prefix}| = 2$.  We orient edges in  $\me_i^{(5B)}$ as in the construction of Step 5B, and other edges of $\me_i^{\take}$, which must be incident to nodes in subgraphs of $\mathbb{X}_5^{\prefix}$, are oriented as out-going from subgraphs in  $\mathbb{X}_5^{\prefix}$. 	By Item (3) of \Cref{lm:ClusteringE}, for any subgraph $\mx \in \mathbb{X}^{\internal}_5$ that has $t$ out-going edges, the total weight of the out-going edges is at  most $tL_i = O(1/\eps) \Delta^+_{i+1}(\mx)$. Thus, $\omega(\me_i^{(5B)}) =  O(1/\eps) \sum_{\mx \in \mathbb{X}^{\internal}_5}\Delta^+_{i+1}(\mx) = O(1/\eps^2)\sum_{\mx \in \mathbb{X}}\Delta^+_{i+1}(\mx) $.
	
	It remains to consider edges incident to at leas one node in a subgraph in $\mathbb{X}_5^{\prefix}$ by \Cref{obs:MEiTakeMinus}. Let $\mx \in \mathbb{X}_5^{\prefix}$. Observe that if $|\mv(\mx)|\geq \frac{2g}{\zeta\eps}$, then
	\begin{equation*}
		\begin{split}
			\Delta_{i+1}^+(\mx) &= O(|\mv(\mx)|\eps L_i)  \qquad \mbox{(by \Cref{lm:manynodes1E})}\\
			& =  O(|\deg_{\mg_i^{\take}}(\mx)|\eps^2 L_i)   \qquad \mbox{(by Item (1) of \Cref{lm:ClusteringE})}
		\end{split}
	\end{equation*}
	Otherwise, $|\mv(\mx)| \leq \frac{2g}{\zeta\eps}$, and hence $|\deg_{\mg_i^{\take}}(\mx)| = O(1/\eps^2)$ by Item (1) of \Cref{lm:ClusteringE}. This implies that the total weight of edges incident to $\mx$ is at most $|\deg_{\mg_i^{\take}}(\mx)| L_i = O(1/\eps^2) (\Delta_{i+1}^+(\mx) + L_i)$. Since $|\mathbb{X}_5^{\prefix}| = 2$ and $\Delta^+_{i+1}(\mx)\geq 0$ for every $\mx \in \mathbb{X}$ by Item (3) of \Cref{lm:ClusteringE}, we have that the total weight of edges ncident to at leas one node in a subgraph in $\mathbb{X}_5^{\prefix}$ is $O(\frac{1}{\eps^2})(\sum_{\mx \in \mathbb{X}} \Delta_{i+1}^+(\mx) + L_i)$. The lemma now follows.  \qed
\end{proof}

\paragraph{Acknowledgement.~} Hung Le is supported by a start up funding of University of Massachusetts at Amherst  and by the National Science Foundation under Grant No. CCF-2121952. Shay Solomon is partially supported by the Israel Science Foundation grant No.1991/19 and by Len Blavatnik and the Blavatnik Family foundation.
We thank Oded Goldreich for his suggestions concerning the presentation of this work and we thank Lazar Milenković for his support.
\bibliographystyle{plain}
\bibliography{spanner}
\pagebreak
\appendix

\section{Proof of \Cref{lm:Clustering-Step5B1E}}\label{app:Clustering-5BProofE}

We follow the same construction in~\cite{LS19}.  We color the contracted nodes of  $\Pbar$ by black color, and other nodes by white color. Let $\bar{W}$ be the set of black nodes of $\Pbar$. Let $\me_{white}$ be the set of edges between two white nodes of $\Pbar$, called white edges. We will orient edges in $\me_i^{\take -}(\Pbar)\setminus  \me_{white}$ as out-going from their black endpoints. Thus, in what follows, we only focus on  orienting white edges.

Let $\Ptilde^{\uncontract}$ be obtained from $\Pbar$ by uncontracting contracted nodes. Let $\tilde{W}$ be the set of nodes in $\Ptilde^{\uncontract}$ that are in contracted nodes of $\Pbar$. We color nodes in $\tilde{W}$ by black color, and other nodes by white color.  Due to a special structure of contracted nodes (as shown in \Cref{lm:Clustering-Step2T2}), we can show that the potential change of a subgraph $\mx$ to be constructed is proportional to the number black nodes it has. Thus, the remaining challenge is to relate the potential change with the number of white edges incident to $\mx$. This is basically accomplished in the construction of the authors~\cite{LS19}; for completeness, we include a proof below.

\begin{lemma}[Lemma 6.17~\cite{LS19}]\label{lm:Clustering-White} We can construct a set of subgraphs $\mathbb{C}_5(\Pbar)$ such that: 	
	\begin{enumerate}
		\item[(1)]  Subgraphs in $\mathbb{C}_5(\Pbar)$ contain every node in $\Ptilde^{\uncontract}$.
		\item[(2)] For every subgraph $\mx\in \mathbb{C}_5(\Pbar)$, $\zeta L_i \leq \adm(\mx)\leq 5L_i$. Furthermore, $\mx$ is a subtree of $\Ptilde^{\uncontract}$ and some edges in  $\me_{white}$  whose both endpoints are in $\mx$.
		\item[(3)]  There is an orientation of edges in $\me_{white}$  
		such that, for any subgraph $\mx \in  \mathbb{C}_5(\Pbar)$,  if the total number of out-going edges incident to nodes in $\mx$ is $t$ for any $t\geq 0$, then:
		\begin{equation}\label{eq:Step5BpotentialWhite}
			\Delta^+_{i+1}(\mx) =  \Omega(t\eps + |\mv(\mx)\cap \tilde{W}|)\eps L_i
		\end{equation}
	\end{enumerate}
\end{lemma}

We now show that \Cref{lm:Clustering-White} implies \Cref{lm:Clustering-Step5B1E}.

\begin{proof}[Proof of \Cref{lm:Clustering-Step5B1E}] Let  $\mathbb{C}_5(\Pbar)$  be the set of subgraphs constructed by \Cref{lm:Clustering-White}. Items (1) and (2) of \Cref{lm:Clustering-Step5B1E} follows directly from Items (1) and (2) of \Cref{lm:Clustering-White}. Thus, it remains to show Item (3).
	
	We orient edges of $\me_i^{\take -}(\Pbar)$ as follows. Edges in $  \me_{white}$ are oriented following \Cref{lm:Clustering-White}. For each edge in  $\me_i^{\take -}(\Pbar)\setminus  \me_{white}$, we orient it as out-going from (arbitrary) one of its black endpoint(s); there must exist at least one black endpoint by the definition of $\me_{white}$. 
	
	Let $t$ be the number of out-going edges in $\me_i^{\take -}(\Pbar)$ incident  to nodes in $\mx$. Let $t_1$ be the number of out-going edges in $\me_{white}$ incident to $\mx$, and $t_2 =  t-t_1$. By the construction in Step 1 (\Cref{lm:Clustering-Step1T2}), every node in $\mx$ is incident to at most $\frac{2g}{\zeta \eps} = O(1/\eps)$ edges in  $\me_i$. Thus, $t_2\leq O(1/\eps) |\mv(\mx)\cap \tilde{W}|$. That is, $ |\mv(\mx)\cap \tilde{W}| = \Omega(t_2\eps)$. Thus, by \Cref{eq:Step5BpotentialWhite}, we have:
	 	\begin{equation*}
	 	\Delta^+_{i+1}(\mx) =  \Omega(t_1\eps + t_2\eps)\eps L_i =  \Omega(t)\eps^2L_i
	 \end{equation*}
 as claimed.
	\qed
\end{proof}

We now focus on proving \Cref{lm:Clustering-White}. There are two main steps in the construction. 

\paragraph{Step 1: tiny clusters.~} We greedily break $\Pbar$ into subpaths of augmented diameter at least $\zeta L_i$ and at most $3\zeta L_i$. This is possible since each node has weight at most $\zeta L_i$ due to the construction in Step 2 (\Cref{lm:Clustering-Step2T2}) and each edge has weight at most $\eps L_i \leq \zeta L_i$ when $\eps \leq \zeta$. Each subpath of $\Pbar$ is called a tiny cluster.  Let $\mathbb{C}_{tiny}$ be the set of all tiny clusters.

Since $3\zeta L_i < L_i/(2)$ (as $\zeta = 1/250$), there is no edge in $\me_i^{\take -}(\Pbar)$ whose both endpoints are in the same tiny cluster.  However, there may be parallel edges: two edges of $\me_{white}$ connecting the same two tiny clusters.

Let $\widehat{P}$ be obtained from $\Pbar$ by contracting each tiny cluster in a single \emph{supernode}. We abuse notation here by viewing  $\me_{white}$ as edges between contracted nodes in $\widehat{P}$.  Given a supernode $\hat{\nu}\in \widehat{P}$, we say that an edge  $\hat{\mbe}\in \me_{white}$ \emph{shadows} $\hat{\nu}$ if  $\hat{\nu}$ lies on the subpath of  $\widehat{P}$  between two endpoints of $\hat{\mbe}$. Note that edges incident to $\hat{\nu}$ shadow  $\hat{\nu}$  by definition.

\paragraph{Step 2: constructing $ \mathbb{C}_5(\Pbar)$ and orient edges.~}  We iteratively construct a set of subgraphs $ \widehat{\mathbb{C}}_5$ from $\widehat{P}$ and $\me_{white}$, orient edges of $\me_{white}$ and mark contracted nodes of $\widehat{P}$ along the way.  Let  $\hat{\nu} \in \widehat{P}$ be an unmarked supernode  incident to a maximum number of \emph{unoriented edges} in $\me_{white}$. Let $\me_{white}(\hat{\nu})\subseteq \me_{white}$ be the set of unoriented edges shadowing $\hat{\nu}$; $\me_{white}(\hat{\nu})$ could be empty. Let  $\widehat{Q}$ be the minimal subpath of  $\widehat{P}$ that contains $\hat{\nu}$ and the endpoints of every edge in $\me_{white}(\hat{\nu})$. (If $\me_{white}(\hat{\nu}) = \emp$ then $\widehat{Q}$ contains a single node $\hat{\nu}$.) We refer to $\hat{\nu}$ as\emph{ the center} of  $\widehat{Q}$. We add the subgraph $\widehat{Q} \cup \me_{white}(\hat{\nu})$ to $ \widehat{\mathbb{C}}_5$. (See \Cref{fig:clustering-tiny}.) Next, we orient unoriented edges of  $\me_{white}$ incident to contracted nodes in $\widehat{Q}$ as out-going from $\widehat{Q}$. We then repeat this step until there is no unmarked supernode in $\widehat{P}$. Finally, we construct  $\mathbb{C}_5(\Pbar)$ by uncontracting each supernode  in $ \widehat{\mathbb{C}}_5$ twice; the first uncontraction gives a subbgraph of $\Pbar\cup \me_{white}$ and the second uncontraction gives a subgraph of $\Ptilde^{\uncontract}\cup \me_{white}$.  This completes the construction of  $\mathbb{C}_5(\Pbar)$ in \Cref{lm:Clustering-White}.

\begin{claim}\label{clm:hatP-structure} During the construction of Step 2,  if  we remove marked supernodes from $\widehat{P}$, then there is no edge in $\me_{white}$ connecting two nodes in two different connected components. 
\end{claim}
\begin{proof} We assume inductively that the claim is true before the removal of  $\widehat{Q}$ in Step 2. Thus,   $\widehat{Q}$  is a subpath of a connected component, say $\widehat{P}_1$,  of $\widehat{P}$ induced by unmakred contracted nodes. Marking contracted nodes of $\widehat{Q}$  could break  $\widehat{P}_1$ into two subpaths say  $\widehat{P}_2$ and  $\widehat{P}_3$, of unmarked contracted nodes. However, any edge of $\me_{white}$ between $\widehat{P}_2$ and  $\widehat{P}_3$ must shadow the center $\hat{\nu}$ of $\widehat{Q}$. That is, there is an edge in $\me_{white}(\hat{\nu})$ whose endpoints are not in $\widehat{Q}$, contradicting the construction in Step 2. \qed
\end{proof}

\begin{figure}[htb]
	\center{\includegraphics[width=0.9\textwidth]{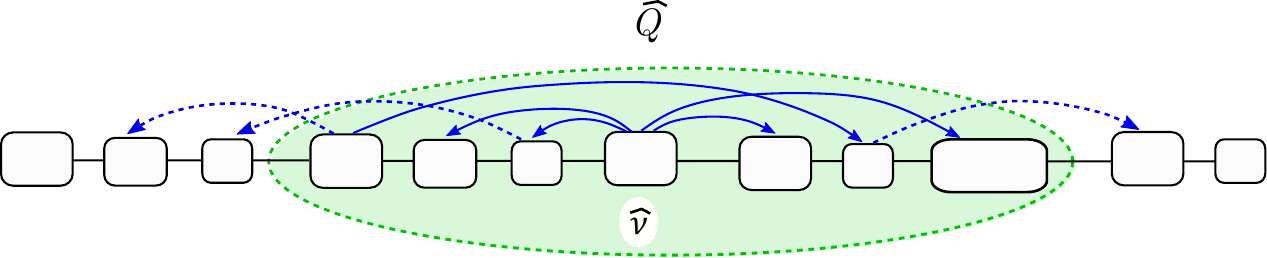}}
	\caption{Rectangular nodes are tiny clusters. Every unoriented edge incident to contracted nodes of $\widehat{Q}$ will be oriented as out-going from $\widehat{Q}$ (except those with both endpoints in  $\widehat{Q}$, which will be oriented arbitrarily). Solid edges are in $\me_{white}(\hat{\nu})$ whose both endpoints are white nodes.}
	\label{fig:clustering-tiny}
\end{figure}

We now focus on proving the properties of $\mathbb{C}_5(\Pbar)$  claimed in \Cref{lm:Clustering-White}. Item (1) follows directly from the construction. Item (2) follows from the following claim.
\begin{claim} \label{clm:edm-IV}
	$\zeta L_i \leq \adm(\widehat{Q}) \leq 5L_i$.
\end{claim}
\begin{proof}
	The lower bound of $ \adm(\widehat{Q}) $ follows directly from the construction. For each edge $\hat{\mbe} = (\dbar{\alpha},\dbar{\beta})$  with both endpoints in $ \adm(\widehat{Q})$, we claim that \begin{equation}\label{eq:diam-shadow}
		\adm(\widehat{Q}[\dbar{\alpha},\dbar{\beta}]) \leq 2(1+3\zeta)L_i
	\end{equation}
	Let $\Qbar$ be obtained from $\widehat{Q}$ by uncontracting tiny clusters; $\Qbar$ is a path.  Let  $\bar{\alpha}$ and $\bar{\beta}$ be endpoints of $\hat{\mbe}$ on $\Qbar$; $\bar{\alpha}$ and $\bar{\beta}$  are nodes in the subpaths constituting tiny clusters  $\dbar{\alpha}$ and $\dbar{\beta}$, respectively. By definition  of $\bar{\me}_i^{close}$ in \Cref{eq:Ebar-farclose},  two intervals $\overline{\mathcal{I}}(\bar{\alpha})$ and $\overline{\mathcal{I}}(\bar{\beta})$ has $\overline{\mathcal{I}}(\bar{\alpha})\cap \overline{\mathcal{I}}(\bar{\beta}) \not=\emptyset$. By definition, each interval, say $\overline{\mathcal{I}}(\bar{\alpha})$, includes all contracted nodes within augmented distance $(1-\psi)L_i \leq L_i$ from $\bar{\alpha}$. This implies $\widehat{Q}[\bar{\alpha}, \bar{\beta}]\leq 2L_i$; thus Equation~\eqref{eq:diam-shadow} holds.  (An extra term $6\zeta L_i$ in \Cref{eq:diam-shadow}  is the upper bound on the sum of augmented diameters of $\dbar{\alpha}$ and $ \dbar{\beta}$.)

	Let $\dbar\nu_0, \dbar\mu_0$ be the two tiny clusters that are endpoints of $\widehat{Q}$. Let $\hat{\mbe}_1 = (\dbar{\nu}_0, \dbar{\nu}_1)$ and $\hat{\mbe}_2 = (\dbar{\mu}_0, \dbar{\mu}_1)$	be two edges shadowing $\dbar{\nu}$, the center of $\widehat{Q}$. Two edges $\hat{\mbe}_1$ and $\hat{\mbe}_2$ exist by the minimality of $\widehat{Q}$. Then:
	\begin{equation*}
		\adm(\widehat{Q}) \leq \adm(\widehat{Q}[\dbar{\nu}_0, \dbar{\nu}_1]) + \adm(\widehat{Q}[\dbar{\mu}_0, \dbar{\mu}_1]) \stackrel{\text{Eq.~\eqref{eq:diam-shadow}}}{\leq} 4(1+3\zeta)L_i  < 5L_i
	\end{equation*}
	as $\zeta = \frac{1}{250}$.\QED
\end{proof}

We now show Item (3) of \Cref{lm:Clustering-White}.  Observe that by construction, $\widehat{Q}$ consists of a path of at most $\frac{5}{\zeta} = O(1)$ tiny clusters since $\adm(\hat{Q})\leq 5L_i$ while each tiny cluster has an augmented diameter at least $\zeta L_i$ . Let $\Qbar$ be obtained from $\widehat{Q}$  by uncontracting tiny clusters, and $\mathcal{Q}$ be obtained from $\Qbar$ by uncontracting contracted nodes. (Note by construction that $\Qbar$  is a path.) Let $\md$ be the diameter path of $\mq$; it could be that $\md$ contains edges in $\me_{white}(\hat{\nu})$.

Let $\widehat{Q}^-$ be obtained from $\widehat{Q}$ by removing all edges in $\me_{white}(\hat{\nu})$; $\hat{\nu}$ is the center of $\widehat{Q}$. Let $\Qbar^{-}$ and $\mq^-$ be obtained from $\Qbar$ and $\mq$ by removing all edges in $\me_{white}(\hat{\nu})$, respectively.  Let $\md^-$ be the shortest path in $\mq^-$ (w.r.t both edge and node weights) between $\md$'s endpoints. (See Figure~\ref{fig:uncontract-tiny}).  Observe that:

\begin{observation}\label{obs:Dminus-vs-AdmX}
	$\adm(\md^-)\geq \adm(\mq)$.
\end{observation}
\begin{proof}
	We have $\adm(\md^-) \geq \adm(\md) = \adm(\mq)$.\qed.
\end{proof}
\begin{figure}[!ht]
	\centering
	\includegraphics[scale = 1.5]{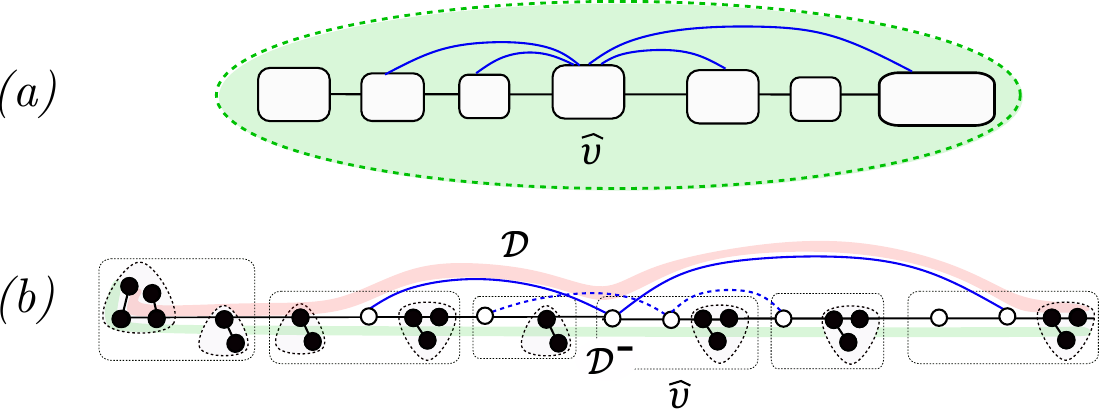}
	\caption{\footnotesize{(a) A subgraph $\widehat{Q}$ formed in Step 5A; triangular blocks are tiny clusters. Solid blue edges are edges in $\me_{white}(\hat{\nu})$. (b) $\mq$ obtained by uncontracting tiny clusters and  contracted nodes. Diameter path $\md$ of $\mq$ is highlighted red; the path $\md^{-}$ is highlighted blue. Solid blue edges are in $\md$.}}
	\label{fig:uncontract-tiny}
\end{figure}

Let 	$\me_{white}({\widehat{Q}})$ be the set of edges in $\me_{white}$ incident to tiny clusters in $\widehat{Q}$.  By construction, $\dbar{\nu}$ is incident to the maximum number of edges in $ \me_{white}$ over every node in $\widehat{Q}$. (Note that $\widehat{Q}$ only has $O(1)$ tiny clusters.) This implies that:

\begin{observation}\label{obs:Incident-tiny}
	$|\me_{white}({\widehat{Q}})| = O(| \me_{white}(\dbar{\nu})|)$.
\end{observation}

We now define:
\begin{equation}\label{eq:PhiX}
	\Phi(\mq) = \sum_{\varphi \in \mv(\mq)}\omega(\varphi) + \sum_{\mbe \in \msttilde_i\cap \me(\mq)}w(\mbe)
\end{equation}
Observe by the definition that:
\begin{equation}\label{eq:Potential-vs-Phi}
	\begin{split}
		 \Delta^+_{i+1}(\mq) &= \Delta_{i+1}(\mq) +  \sum_{\mbe \in \msttilde_i\cap \me(\mq)}w(\mbe)\\
		&= \sum_{\varphi \in \mv(\mq)}\omega(\varphi) - \adm(\mq) +  \sum_{\mbe \in \msttilde_i\cap \me(\mq)}w(\mbe)\\
		&= \Phi(\mq) - 		\adm(\mq) 
	\end{split}
\end{equation}

\begin{lemma}\label{lm:leftover-LS19} Let $\mu$ be node in $\md^-$ that is incident to $t \geq 1$ edges in  $\me_{white}(\dbar{\nu})$. Then  $\Delta^+_{i+1}(\mq)= \Omega(t\epsilon L_i)$.
\end{lemma}
\begin{proof} Let $\mz$ be the set other $t$ endpoints  of $t$ edges incident to $\mu$. If $|\mz \cap \md|\leq t/2$, then:
	\begin{equation*}
		\begin{split}
			\Delta^+_{i+1}(\mq) &= \Phi(\mq) - \adm(\mq) \stackrel{\mbox{Obs.~\ref{obs:Dminus-vs-AdmX}}}{\geq} \Phi(\mq) - \adm(\md^-)\\
			&\geq \adm(\mz\setminus \md) \stackrel{\mbox{Eq.~\eqref{eq:PhiX}}}{\geq} O(|\mz\setminus \md|\epsilon L_i) = \Omega(t\epsilon L_i),
		\end{split}
	\end{equation*}
	as claimed.

\begin{wrapfigure}{r}{0.45\textwidth}
	\vspace{-35pt}
	\begin{center}
		\includegraphics[width=0.45\textwidth]{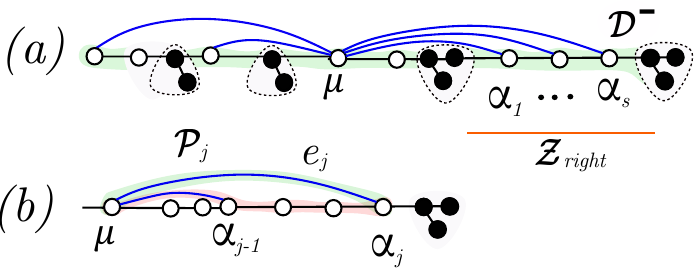}
	\end{center}
	\caption{\footnotesize{(a) blue edges are level-$i$ edges incident to $\mu$. (b) $\md_j$ obtained by replacing $\md_{j-1}[\mu,\alpha_{j}]$ by $\mp_j = (\mu,e,\alpha_{j})$.}}
	\vspace{-5pt}
	\label{fig:z-right}
\end{wrapfigure}

Herein, we assume that  $|\mz \cap \md|\geq t/2$. We can also assume~w.l.o.g. that at least $t/4$ nodes in $\mz$ that are to the right of $\mu$ on $\md^-$. (Note that $\md^-$ contains only $\widetilde{\mst}_i$ edges.) Let $\mz_{right} = \{\alpha_1,\ldots, \alpha_s\}$, $s \geq t/4$,  be the set of nodes to the right of $\mu$, and such that $\alpha_{j-1}\in \md^{-}[\mu, \alpha_{j}]$ for any $j\in [2,s]$ (see Figure~\ref{fig:z-right}(a)). Let $\mbe_j = (\mu,\alpha_j)$ be the edge in $\me_{white}(\dbar{\nu})$ between $\mu$ and $\alpha_j$, $j \in [2]$. By construction, $\Qbar$ is  a path. Note that edges in $\me_{white}$ have white nodes as endpoints, which by definition, are uncontracted nodes in $\Qbar$, and hence $\mu$ and $\alpha_j$ are uncontracted nodes. That is, $\mu = \bar{\mu}$ and $\alpha_j = \bar{\alpha}_j$.  

Let $\md_0 = \md^-$ and we define $\md_j$ for each $j \in [1,s]$ as follows: $\md_{j}$ is obtained from $\md_{j-1}$ by replacing the subpath $\md_{j-1}[\mu, \alpha_{j}]$ by path $\mp_{j} \stackrel{\mbox{def.}}{=} (\mu, \mbe_j, \alpha_{j})$ which has only one edge $\mbe_j$ (see Figure~\ref{fig:z-right}(b)). 

\begin{claim}\label{clm:dj-vs-dj1}
	$\adm(\md_j) \leq \adm(\md_{j-1}) + \epsilon g L_i$.
\end{claim}
\begin{proof}
	We have: 
	
	\begin{equation*}
		\begin{split}
			\adm(\md_{j-1}) - \adm(\md_j) &= w(\md_{j-1}[\nu,\mu]) - \omega(\mp_j) \geq  d_{\mathcal{H}_i}(\nu,\mu) - \omega(\mp_j) \\
			&\geq 6g\epsilon  \cdot \omega(\mbe_j)  - \omega(\nu) - \omega(\mu) \qquad \mbox{(by \Cref{obs:greedy-HiE})} \\
			&\geq 6g\epsilon L_i/2 - 2g\epsilon L_i = g\epsilon L_i,
		\end{split}
	\end{equation*}
	as claimed.\qed
\end{proof}

By Claim~\ref{clm:dj-vs-dj1}, we have:
\begin{equation}\label{eq:mds-vs-mdM}
	\adm(\md_s) ~\leq~ \adm(\md_0) + s\epsilon g L_i ~=~ \adm(\md^-) + s\epsilon g L_i.
\end{equation}
Since $\md_s$ and $\md$ has the same endpoint and $\md$ is a shortest path, $\adm(\md_s)\geq \adm(\md)$.  This implies:
\begin{equation*}
	\begin{split}
		\Delta^+_{i+1}(\mq) &= \Phi(\mq) - \adm(\mq) =  \Phi(\mq) - \adm(\md)\\
		&\geq \Phi(\mq)  - \adm(\md_s) \geq \adm(\md^-) - \adm(\md_s)\\
		&\stackrel{\mbox{Eq.~\eqref{eq:mds-vs-mdM}}}{\geq} s\epsilon g L_i = \Omega(t\epsilon L_i),\qquad \mbox{since }s\geq t/4 
	\end{split}
\end{equation*}
as desired.\qed		
\end{proof}

\begin{lemma}\label{lm:H5}
	$\Delta^+_{i+1}(\mq) = \Omega(\eps^2)|\me_{white}(\widehat{Q})|L_i$.
\end{lemma}
\begin{proof}
	Suppose that $|\me_{white}(\dbar{\nu})| = \frac{t}{\epsilon}$ for some $t > 0$. By Observation~\ref{obs:Incident-tiny}, it holds that $|\me_{white}(\widehat{Q})| = O( \frac{t}{\epsilon})$. This implies:
	\begin{equation}\label{eq:Size-Etiny-X}
		t = \Omega(\eps)|\me_{white}(\widehat{Q})|
	\end{equation}
	
	Let $\overline{\mp}$ be the path of contracted nodes corresponding to tiny cluster $\dbar{\nu}$. Let $\mz$ be the set of uncontracted nodes of $\overline{\mp}$ that are incident to at least $\frac{t \zeta}{4g}$ edges in $\me_{white}(\hat{\nu})$. We claim that:
	\begin{claim}\label{clm:Size-Z}
		$|\mz|\geq \frac{t\zeta}{4g}$.
	\end{claim}
	\begin{proof}
		Let $\ma$ be the set of remaining uncontracted nodes in  $\overline{\mp}\setminus \mz$.  Then $|\ma|\leq \frac{\adm(\overline{\mp})}{\zeta L_{i-1}} = \frac{2g}{\zeta \epsilon}$ since $\adm(\overline{\mp})\leq gL_i$. Recall that each in $\mz$ is incident to at most $\frac{2g}{\epsilon}$ edges by the construction in Step 1 (\Cref{lm:Clustering-Step1T2}).  Thus, we have:
		\begin{equation*}
			|\me_{white}(\dbar{\nu})| \leq |\mz|\frac{2g}{\zeta\epsilon} + \frac{t \zeta}{4g} |\ma|  < \frac{t\zeta}{4g} \cdot \frac{2g}{\zeta \eps} + \frac{t \zeta}{4g} \cdot \frac{2g}{\zeta \epsilon} = \frac{t}{\epsilon}
		\end{equation*}	
		This is a contradiction since $|\me_{white}(\dbar{\nu})| = \frac{t}{\epsilon}$.\qed
	\end{proof}
	We consider two cases:
	\begin{itemize}
		\item \textbf{Case 1:} $\mathcal{D}^- \cap \mz \not= \emptyset$.  By Lemma~\ref{lm:leftover-LS19}, $\Delta^+_{i+1}(\mq) = \Omega(\frac{t \zeta}{4g} \epsilon L_i) = \Omega(t \epsilon L_i)$ since every node in $\mz$ is incident to at least $\frac{t \zeta}{4g}$ edges in $\me_{white}(\hat{\nu})$. By Equation~\eqref{eq:Size-Etiny-X}, $\Delta^+_{i+1}(\mq) = \Omega(\epsilon^2) |\me_{white}(\widehat{Q})| L_i$.
		
		\item  \textbf{Case 1:} $\mathcal{D}^- \cap \mz = \emptyset$. Then, it holds that:
		\begin{equation*}
			\begin{split}
				\Delta^+_{i+1}(\mq) &= \Phi(\mq^-) - \adm(\mq) \stackrel{\mbox{Obs.~\ref{obs:Dminus-vs-AdmX}}}{\geq }\Phi(\mq^-) - \md^- \\
				&\geq~ \sum_{\varphi \in \mz} \omega(\varphi)  = \Omega(|\mz|\epsilon L_i) = \Omega(t\epsilon L_i)  =  \Omega(\epsilon^2) |\me_{white}(\widehat{Q})| L_i 
			\end{split}
		\end{equation*}
	\end{itemize} 
	In both cases,  the lemma holds. \qed
\end{proof}

We note that out-going edges of $\me_{white}$ incident to nodes in $\mathcal{Q}$ are $\me_{white}(\hat{Q})$. Thus, the following lemma implies  Item (3) of \Cref{lm:Clustering-White}. 

\begin{lemma} $	\Delta^+_{i+1}(\mq) =  \Omega(|\me_{white}(\hat{Q})|\eps + |\mv(\mq)\cap \tilde{W}|)\eps L_i$.
\end{lemma}
\begin{proof} We will show that:
	 \begin{equation}\label{eq:blacknodes-C}
	 \Delta^+_{i+1}(\mq) =  \Omega(|\mathcal{Q} \cap \tilde{W}|) \epsilon L_i
	 \end{equation}
	 The lemma then follows from \Cref{eq:blacknodes-C} and \Cref{lm:H5}.
	 
	 Let $\bar{\mu}$ be a black node in $\Qbar$; $\bar{\mu}\in \Qbar\cap \bar{W}$. Let $\mathcal{T}_{\bar{\mu}}$ be the subtree of $\msttilde_{i}$ corresponding to $\bar{\mu}$. Note that $\mathcal{T}_{\bar{\mu}}$ is a subtree of $\mq$ as well. Recall that $\md^-$ is a path of  $\mq^-$. Let $\md_{\bar{\mu}} = \md^{-} \cap \mathcal{T}_{\bar{\mu}}$ be the subpath of $\md$ in the tree $\mathcal{T}_{\bar{\mu}}$. It is possible that $\md_{\bar{\mu}} = \emptyset$. 
	 \begin{claim}\label{clm:Tmu-delta} Let $\Delta(\mt_{\bar{\mu}}) = \Phi(\mt_{\bar{\mu}}) - \adm(\md_{\bar{\nu}})$ where $\Phi(\mt_{\bar{\mu}})  = \sum_{\varphi \in \mv(\mt_{\bar{\mu}})}\omega(\varphi) + \sum_{\mbe \in \msttilde_i\cap \me(\mt_{\bar{\mu}})}w(\mbe)$. Then:
	 	$$\Delta(\mt_{\bar{\mu}}) = \Omega(| \mv(\mt_{\bar{\mu}})|\epsi L_i).$$
	 \end{claim}
 	\begin{proof}
 		By Item (3) of \Cref{lm:tree-clustering},  $\adm(\mt_{\bar{\mu}} \setminus \md_{\bar{\mu}}) = \Omega( \md_{\bar{\mu}}) = \Omega(|\mv(\md_{\bar{\mu}})|\eps L_i)$. Clearly, $\adm(\mt_{\bar{\mu}} \setminus \md_{\bar{\mu}}) = \Omega(|\mv(\mt_{\bar{\mu}})| - |\mv(\md_{\bar{\mu}})|)\eps L_i$.  Thus, 
 		\begin{equation*}
 			\begin{split}
 			\Delta(\mt_{\bar{\mu}}) &\geq \adm(\mt_{\bar{\mu}} \setminus \md_{\bar{\mu}}) = (\adm(\mt_{\bar{\mu}} \setminus \md_{\bar{\mu}}))/2 +  \adm(\mt_{\bar{\mu}} \setminus \md_{\bar{\mu}})/2\\
 			& =  \Omega(|\mv(\mt_{\bar{\mu}})| - |\mv(\md_{\bar{\mu}})|)\eps L_i + \Omega(|\mv(\md_{\bar{\mu}})|\eps L_i)=   \Omega(|\mv(\mt_{\bar{\mu}})|\eps L_i)~,
 			\end{split}
 		\end{equation*}
 	as claimed. \qed
 	\end{proof}
	 
	 Since $\adm(\md)\leq \adm(\md^{-})$, we have:
	 \begin{equation*}
	 	\begin{split}
	 		\Delta^+_{i+1} (\mq) &\geq  \Phi(\mq) - \adm(\md^{-})  =\Phi(\mq-) - \adm(\md^{-})\\
	 		& = \sum_{\bar{\mu}\in \Qbar\cap \bar{W}}\Delta(\mt_{\bar{\mu}})  \\
	 		&=  \sum_{\bar{\mu}\in \Qbar\cap \bar{W}} \Omega(| \mv(\mt_{\bar{\mu}})|\epsi L_i) \qquad \mbox{(by \Cref{clm:Tmu-delta})}\\
	 		&= \Omega(|\mv(\mq)\cap \tilde{W}|\epsi L_i),
	 	\end{split}
	 \end{equation*}
 Thus, \Cref{eq:blacknodes-C} holds as claimed. \qed
\end{proof}

\end{document}